\pdfoutput=1

\documentclass[11pt]{article}

\usepackage[T1]{fontenc}
\usepackage[utf8]{inputenc}
\usepackage{lmodern}
\usepackage{microtype}
\usepackage[a4paper,margin=30mm]{geometry}
\usepackage{amsmath,amssymb,amsthm,mathtools}
\usepackage{aliascnt}
\usepackage{booktabs}
\usepackage{array}
\usepackage{graphicx}
\usepackage{capt-of}
\usepackage{pgfplots}
\pgfplotsset{compat=1.18}
\usepgfplotslibrary{groupplots}
\usetikzlibrary{positioning}
\usepackage{enumitem}
\usepackage{float}
\usepackage{xcolor}
\usepackage[hidelinks]{hyperref}
\usepackage[nameinlink,capitalise,noabbrev]{cleveref}

\hypersetup{
  pdftitle={Test or Run? Scheduling Jobs of Unknown Length},
  pdfauthor={Václav Rozhoň},
  pdfsubject={Online scheduling with testing under explorable uncertainty},
  pdfkeywords={scheduling with testing, explorable uncertainty, online algorithms,
    randomized algorithms, instance optimality, competitive analysis,
    total completion time}
}

\allowdisplaybreaks
\setlist[itemize]{leftmargin=*,topsep=4pt,itemsep=2pt}
\setlist[enumerate]{leftmargin=*,topsep=4pt,itemsep=2pt}
\floatstyle{ruled}
\newfloat{algorithm}{tbp}{loa}
\floatname{algorithm}{Algorithm}
\newcounter{algorithmicline}
\newenvironment{algorithmic}{%
  \setcounter{algorithmicline}{0}%
  \begin{list}{\footnotesize\arabic{algorithmicline}:}{%
    \usecounter{algorithmicline}%
    \setlength{\labelwidth}{2.2em}%
    \setlength{\labelsep}{0.5em}%
    \setlength{\leftmargin}{2.7em}%
    \setlength{\itemsep}{1.5pt}%
    \setlength{\parsep}{0pt}%
    \setlength{\topsep}{3pt}%
  }\small
}{\end{list}}
\newcommand{\State}{\item}
\newcommand{\Statex}{\item[]}
\crefname{algorithm}{algorithm}{algorithms}
\Crefname{algorithm}{Algorithm}{Algorithms}

\newtheorem{theorem}{Theorem}[section]
\newaliascnt{lemma}{theorem}
\newtheorem{lemma}[lemma]{Lemma}
\aliascntresetthe{lemma}
\newaliascnt{proposition}{theorem}
\newtheorem{proposition}[proposition]{Proposition}
\aliascntresetthe{proposition}
\newaliascnt{corollary}{theorem}
\newtheorem{corollary}[corollary]{Corollary}
\aliascntresetthe{corollary}
\newaliascnt{claim}{theorem}
\newtheorem{claim}[claim]{Claim}
\aliascntresetthe{claim}
\crefname{lemma}{lemma}{lemmas}
\Crefname{lemma}{Lemma}{Lemmas}
\crefname{proposition}{proposition}{propositions}
\Crefname{proposition}{Proposition}{Propositions}
\crefname{corollary}{corollary}{corollaries}
\Crefname{corollary}{Corollary}{Corollaries}
\crefname{claim}{claim}{claims}
\Crefname{claim}{Claim}{Claims}
\theoremstyle{definition}
\newaliascnt{definition}{theorem}

\aliascntresetthe{definition}
\newaliascnt{conjecture}{theorem}

\aliascntresetthe{conjecture}
\crefname{definition}{definition}{definitions}
\Crefname{definition}{Definition}{Definitions}
\crefname{conjecture}{conjecture}{conjectures}
\Crefname{conjecture}{Conjecture}{Conjectures}
\theoremstyle{remark}
\newaliascnt{remark}{theorem}
\newtheorem{remark}[remark]{Remark}
\aliascntresetthe{remark}
\crefname{remark}{remark}{remarks}
\Crefname{remark}{Remark}{Remarks}

\newcommand{\ALG}{\operatorname{ALG}}
\newcommand{\OPT}{\operatorname{OPT}}

\newcommand{\RdetOT}{R_{\mathsf{OT}}^{\mathrm{det}}}
\newcommand{\RrandOT}{R_{\mathsf{OT}}^{\mathrm{rand}}}
\newcommand{\RdetRO}{R_{\mathsf{RO}}^{\mathrm{det}}}
\newcommand{\RrandRO}{R_{\mathsf{RO}}^{\mathrm{rand}}}
\newcommand{\RdetBO}{R_{\mathsf{BO}}^{\mathrm{det}}}
\newcommand{\RrandBO}{R_{\mathsf{BO}}^{\mathrm{rand}}}
\newcommand{\Rquad}{R_{\mathrm{quad}}}
\newcommand{\PhiOT}{\Phi_{\mathsf{OT}}}
\newcommand{\PhiBE}{\Phi_{\mathsf{BE}}}
\newcommand{\PhiBO}{\Phi_{\mathsf{BO},u}}
\newcommand{\PhiRO}{\Phi_{\mathsf{RO},u}}
\newcommand{\uDet}[1]{u_{#1}^{\mathrm{det}}}
\newcommand{\uRand}[1]{u_{#1}^{\mathrm{rand}}}

\newcommand{\eps}{\varepsilon}

\newcommand{\one}{\mathbf{1}}
\newcommand{\defeq}{\mathrel{:=}}
\newcolumntype{L}[1]{>{\raggedright\arraybackslash}p{#1}}
\newcommand{\EE}{\mathbb{E}}
\newcommand{\SPT}{\operatorname{SPT}}
\newcommand{\Area}{\operatorname{Area}}
\newcommand{\Prob}{\mathbb{P}}
\newcommand{\Stat}{\mathrm{stat}}

\title{Test or Run? Scheduling Jobs of Unknown Length}

\author{Václav Rozhoň}
\date{August 2026}

\begin{document}
\maketitle

\begin{abstract}
A machine faces many jobs whose lengths are hidden.  Spending one unit of
time to inspect a job may reveal a short job that should be finished now, or
it may reveal nothing useful while every other job waits.  When should the
machine keep looking, and when should it start working?

We study natural variants of this question and provide optimal algorithms in
both the worst-case and instance-optimal frameworks. The resulting algorithms
are often quite simple, which may make them useful in practice.
\end{abstract}

\tableofcontents
\clearpage

\section{Introduction}

Suppose that we have many jobs waiting for one machine.  We want to minimize
the sum of the times at which they finish.  The length of each job is written
on a folded piece of paper.  If we knew all the lengths, the right schedule
would be clear: we would simply run the jobs from shortest to longest.  But
the papers are folded, and opening one of them takes one unit of machine time.
Should we open another paper, or should we start working?

Opening a paper may reveal a tiny job that we can finish immediately.  It may
also reveal a huge job that is best left for the end.  On the other hand,
looking has a cost.  If one hundred jobs are still unfinished, one unit spent
opening a paper adds one hundred units to the sum of completion times.  The
machine must therefore balance two uses of its time: learning which jobs are
short and completing jobs that are already understood.

Variants of this question have appeared in both stochastic and adversarial
scheduling with
testing~\cite{LeviMagnantiShaposhnik2019,DurrErlebachMegowMeissner2020,DogeasErlebachLiang2024}.
This work improves earlier results and studies new variants of this problem.  We start
with the cleanest version, in which every paper has to be opened.
We then allow a job to be run without opening its paper. We also consider the variant where opening a paper actually corresponds to \emph{optimizing} the running time -- it may or may not speed up the execution of the job. Together, this leads to several variants of this problem. However, the same algorithmic
ideas---random sampling, thresholding, and shortest-first execution---give
exact asymptotic answers in all considered models.

\subsection{Warmup: every job has to be tested}
\label{sec:intro-obligatory}

Consider the following problem from the area of online algorithms.  There are $n$ jobs.  Testing job $i$ takes
one unit of time and reveals a number $p_i\geq0$.  The job then needs another
$p_i$ units of processing; if $p_i=0$, it completes with the test.  In the simplest variant of this problem, every job has to be tested: testing is \emph{obligatory}.  The formal model
is given in \cref{sec:model}.

To analyze scheduling in this model, we compare an online scheduler with the shortest-first schedule that knows
all values $p_i$ in advance.  The \emph{competitive ratio} is the online cost
divided by this offline optimum.  We care about its asymptotic value.  Since
every job must be tested, even the optimum has cost $\Omega(n^2)$, so we
ignore additive terms of size $o(n^2)$.

The obligatory-testing model for total completion time was introduced by
Dogeas, Erlebach, and Liang~\cite{DogeasErlebachLiang2024}.  Their work and a
later lower bound of Liang and Liang~\cite{LiangLiang2026} left a gap between
the best deterministic upper and lower bounds.  They also left the randomized
case open.  We close both gaps.

\begin{theorem}[Obligatory testing, informal version of
\protect\cref{thm:det-obligatory,thm:rand-obligatory}]
\label{thm:intro-obligatory}
In the obligatory-testing model, the best deterministic asymptotic
competitive ratio is $\RdetOT=1.57\ldots$.  Against an adversary
that fixes the input before the algorithm draws its random bits, the best
randomized ratio is $\RrandOT=4/3$.  The two optimal algorithms are sketched in
\cref{alg:intro-obligatory-det,alg:intro-obligatory-rand}.
\end{theorem}

Here and below, \emph{shortest first} means processing revealed jobs in
nondecreasing order of $p_i$.  The following schematic pseudocode gives a
little more detail about the two algorithms.  The formulas are not meant to
be memorable.

\begin{algorithm}[H]
\caption{}
\label{alg:intro-obligatory-det}
\begin{algorithmic}
\State Put $c=0.57\ldots$ and $R=1+c=\RdetOT$.
\State Before test $i+1$, where $0\le i<n$, let $c_i$ be the number of
positive jobs processed immediately and $d_i$ the number deferred.
\State Put
\[
 y_i=\frac{d_i-R(c_i+d_i)}{n-i}.
\]
\State Use threshold $T_i=1$ if $y_i\le-1$; otherwise use
\[
 T_i=1+\frac12\ln\frac{c+2}{c-2y_i}.
\]
\State Test the next job.  Process it now if $p\le T_i$; otherwise defer it.
\State After all tests, finish deferred jobs shortest first.
\end{algorithmic}
\end{algorithm}

\begin{algorithm}[H]
\caption{}
\label{alg:intro-obligatory-rand}
\begin{algorithmic}
\State Set $B_n=O(n^{1/20})$ and test $O(n^{3/4})$ random jobs to estimate
the empirical job-length distribution $\widehat D$.
\State For each candidate threshold $0\le t\le B_n$, put
\[
 a(t)=\widehat D([0,t]),\qquad
 m(t)=\int_{[0,t]}p\,d\widehat D(p),
\]
and define
\[
 \widehat\tau
 =\min_{\substack{0\le t\le B_n\\a(t)>0}}
   \frac{1+m(t)}{a(t)}.
\]
\State If the displayed minimum is empty or larger than $B_n$, test every
remaining job and then finish all revealed jobs shortest first.
\State Otherwise process the sampled jobs with $p\le\widehat\tau$.  Keep
testing jobs, processing them immediately when $p\le\widehat\tau$ and
deferring them otherwise.
\State After all tests, finish deferred jobs shortest first.
\end{algorithmic}
\end{algorithm}

Both algorithms follow the same principle: process short jobs as soon as
possible and leave long jobs for the end.  The difficult part is choosing the
threshold.  A deterministic algorithm has no representative view of the
unseen jobs.  Its decisions are therefore more complicated: in fact, its
threshold must change with the history.  The exact update is given in
steps 3 and 4 of \cref{alg:intro-obligatory-det}.  The matching lower bound
in \cref{sec:lower} uses an adversary that reveals jobs with decreasing
positive values of $p_i$ while hiding many zero jobs until the end.

Randomization gives the algorithm a global view of the input.  It first tests
a sublinear random sample and estimates the distribution of job lengths.
Because the labels are privately shuffled, the remaining empirical
distribution stays close to the sampled one throughout the run.  The
algorithm therefore faces essentially the same distribution at every step
and can use one fixed threshold.

The choice of that threshold has a simple interpretation.  For a candidate
$t$, a full round of tests together with the immediate processing of outcomes
at most $t$ uses expected work $1+m(t)$ per original job and completes a fraction
$a(t)$ of the jobs.  We call this repeated test-and-process action the
\emph{testing module} for $t$.  Thus
\[
 \frac{1+m(t)}{a(t)}
\]
is the expected work per immediate completion of this testing module.
The algorithm minimizes this quantity and uses its optimum
$\widehat\tau$ as the actual threshold.

The matching randomized lower bound is correspondingly simple: half of the
jobs have length zero and half have length two, and their labels are shuffled.
Until the shuffle is uncovered, one unit of work completes only about half a
job.  Counting the area under the number of unfinished jobs gives the lower
bound $4/3$.  The formal argument is in \cref{sec:random-obligatory}.

\subsection{One algorithm for every bounded input multiset}
\label{sec:intro-obligatory-instance}

Let us return to \cref{alg:intro-obligatory-rand}.  The ratio $4/3$ describes
only the hardest input.  It says much less on an easy one.  Suppose, for
example, that all processing times are very large compared with the unit
test---a plausible regime in many applications.  Then the straightforward
strategy that first tests everything and subsequently runs the jobs shortest
first is nearly optimal.  The worst-case guarantee of
\cref{alg:intro-obligatory-rand} still tells us only that its cost is at most
$1.33\ldots$ times the optimum.

In fact, the same algorithm satisfies a much stronger per-input guarantee:
it is \emph{instance-optimal} on every class with a fixed upper bound on job
lengths.

\begin{theorem}[Obligatory-testing instance optimality, informal version of
\protect\cref{thm:obligatory-instance}]
\label{thm:intro-obligatory-instance}
Fix $L<\infty$ and the obligatory model.  The algorithm $\mathcal A_n^*$ in
\cref{alg:intro-obligatory-rand} has the following property: for every
input $p\in[0,L]^n$ and every online algorithm $\mathcal A'$,
\[
 \EE\ALG_{\mathcal A_n^*}\le \EE\ALG_{\mathcal A'}+o_L(n^2),
\]
where both costs are measured on $p$ and include the private shuffle and the
algorithms' remaining randomness.
The competing algorithm $\mathcal A'$ may even be designed specifically for
$p$ and told all its job lengths in advance, though not their shuffled assignment
to the public labels.
\end{theorem}

Assumptions in the theorem are necessary.  Without a fixed bound $L$
there is no randomized instance-optimal algorithm, and without randomization
there is no deterministic one; see \cref{thm:ot-instance-impossibility}.

Two points about the theorem are worth emphasizing.  First, the comparison
is made separately on every input.  We must therefore compete not only
with one fixed online algorithm, but also with algorithms optimized for this
particular $p$.  Allowing $\mathcal A'$ to know the entries of $p$ makes this
strongest possible online benchmark explicit.

Second, the existence of some instance-optimal algorithm is perhaps not
surprising.  One could imagine a meta-algorithm that samples the input to
estimate the empirical distribution of its job lengths and then runs an
optimal algorithm for that estimate.  The
substantive part of the theorem is that no complicated, unimplementable meta-algorithm is
needed.  The optimal algorithm is itself a simple
threshold algorithm from \cref{alg:intro-obligatory-rand}.

\subsection{Generalizations of the model}
\label{sec:intro-optional-map}

We next weaken the obligatory model in two directions.  First, we may allow a
job to run without being tested.  Second, the unit preliminary action may have
two different benefits.  It may reveal the true running time $p_i$ before the
job is run.  It may also be an optimization procedure that replaces a common
raw running time $u$ by a shorter time $p_i\leq u$.  For example, a generic
piece of code solves every task within time $u$, but an AI agent can spend one
unit of time adapting the code to a particular input and make it faster.

These two benefits are independent.  An optimization may reveal its new
running time $p_i$, giving the scheduler both a faster job and information
about where the job belongs.  Or it may leave $p_i$ hidden until the optimized
job finishes.

These choices give the following four models.

\begin{center}
\small
\setlength{\tabcolsep}{4pt}
\renewcommand{\arraystretch}{1.12}
\begin{tabular}{@{}|L{0.18\linewidth}|L{0.24\linewidth}|L{0.34\linewidth}|L{\dimexpr0.24\linewidth-8\tabcolsep-5\arrayrulewidth\relax}|@{}}
\hline
model & without the unit action & with the unit action & where discussed \\
\hline
obligatory testing & cannot be run & reveals $p_i$ and unlocks the job
 & Intro~\ref{sec:intro-obligatory};
   \cref{sec:obligatory,sec:random-obligatory,sec:obligatory-instance} \\
\hline
blind execution & runs for its hidden length $p_i$ & reveals $p_i$ before execution
 & Intro~\ref{sec:intro-blind}; \cref{sec:blind} \\
\hline
blind optimization & runs for raw time $u$ & replaces $u$ by a still-hidden $p_i\leq u$
 & Intro~\ref{sec:intro-blind-optimization}; \cref{sec:blind-optimization} \\
\hline
revealing optimization & runs for raw time $u$ & reveals $p_i$ and replaces $u$ by $p_i$
 & Intro~\ref{sec:intro-revealing-optimization}; \cref{sec:optional,sec:common-upper-instance,sec:random-common-upper} \\
\hline
\end{tabular}
\end{center}

The last three rows separate the two benefits of the unit preliminary action.
Blind execution studies information without speedup and is related to the
stochastic model of Levi, Magnanti, and
Shaposhnik~\cite{LeviMagnantiShaposhnik2019}.  Blind optimization studies
speedup without advance information and is related to the model of Damerius
et al.~\cite{DameriusEtAl2023}.  Revealing optimization, introduced by
D{\"u}rr et al.~\cite{DurrErlebachMegowMeissner2020}, contains both benefits.

The three optional models have somewhat different optimal algorithms and
input-specific benchmarks, so we state their results separately.  Their
randomized algorithms nevertheless share one design principle: privately
shuffle the jobs, use a sublinear sample to estimate the input distribution,
and optimize the rest of the schedule using that information.

\subsection{Information without speedup: blind execution}
\label{sec:intro-blind}

Suppose an untested job can simply be run until it finishes.  Its running
time is its true hidden length $p_i$.  Testing does not make the job shorter;
it only tells us $p_i$ before we decide when to run it.  We call this
\emph{blind execution}.

Without an upper bound on the job lengths, no randomized algorithm has a
finite competitive ratio, even if we ignore an additive $o(n^2)$ term; see
\cref{thm:blind-unbounded-impossible} for a simple counterexample.  For every
fixed bound $L$, however, sampling again gives one algorithm that is optimal
on every bounded input.

\begin{theorem}[Blind-execution instance optimality, informal version of
\protect\cref{thm:blind-instance}]
\label{thm:intro-blind-instance}
Fix $L<\infty$ and the blind-execution model.  There is one randomized algorithm $\mathcal A^*$ from \cref{alg:intro-blind} such that,
for every input $p\in[0,L]^n$ and every online algorithm
$\mathcal A'$,
\[
 \EE\ALG_{\mathcal A^*}\le \EE\ALG_{\mathcal A'}+o_L(n^2).
\]
\end{theorem}

To describe the algorithm, let $\widehat D$ denote the estimated
job-length distribution and put
$\widehat\mu=\int p\,d\widehat D(p)$.  As in
\cref{alg:intro-obligatory-rand}, define
\[
 a(t)=\widehat D([0,t]),\qquad
 m(t)=\int_{[0,t]}p\,d\widehat D(p),
 \qquad
 \widehat\tau
 =\min_{\substack{t\ge0\\a(t)>0}}\frac{1+m(t)}{a(t)}.
\]

The optimal algorithm is more subtle than
\cref{alg:intro-obligatory-rand}, because running a job blindly can be better
than testing it.  Sometimes, after the initial sample, the best choice is to
run every remaining job blindly.  Even when further testing is useful, it
need not be useful all the way to the end of the input.

For a simple illustration, suppose that jobs have only two possible lengths.
The obligatory-testing algorithm would process short outcomes immediately and
defer long ones.  Now consider the last untouched job.  If it is short, we
want to run it next.  If it is long, it can still be run next, as the first
job in the arbitrarily ordered long tail.  Thus in either case the useful
action is to run the job; testing it first only adds one unit of work.  This
is why the value of testing can disappear before every job has been tested.

In general, the algorithm tests only a fraction $q$ of the jobs and uses the
threshold $\widehat\tau$ from \cref{alg:intro-obligatory-rand} to decide
which tested jobs to process immediately.  After testing stops, deferred
jobs with $\widehat\tau\le p<\widehat\mu$ are processed before the untouched
jobs: their known lengths are smaller than the expected duration
$\widehat\mu$ of a blind execution.  The untouched jobs are then run
blindly, and tested jobs with $p\ge\widehat\mu$ form the final tail.  The
optimal fraction $q$ minimizes the explicit quadratic
$F_{\mathsf{BE},\widehat D}$ from \cref{eq:blind-F}.

\begin{algorithm}[H]
\caption{}
\label{alg:intro-blind}
\begin{algorithmic}
\State Shuffle the jobs and run $o(n)$ of them blindly; use their observed
lengths to estimate $\widehat D$ and $\widehat\mu$.
\State Compute
$\widehat\tau=\min_{t\ge0,\,a(t)>0}(1+m(t))/a(t)$.
\State If $\widehat\tau\ge\widehat\mu$, put $q=0$.  Otherwise choose
$q\in\arg\min_{0\le z\le1}F_{\mathsf{BE},\widehat D}(z)$ from \cref{eq:blind-F}. 
\State Test a random $q$-fraction of the remaining jobs.  Whenever a test
reveals $p<\widehat\tau$, process that job immediately; defer all other
outcomes.
\State After the tests stop, process the deferred jobs with
$\widehat\tau\le p<\widehat\mu$, shortest first.
\State Run every untouched job blindly.
\State Finish the remaining tested jobs shortest first.
\end{algorithmic}
\end{algorithm}

\subsection{Blind optimization}
\label{sec:intro-blind-optimization}

The two remaining models start from the same operation.  Every job takes a
known raw time $u$, and one unit of optimization reduces its running time to
$p_i\leq u$.  They differ only in whether optimization reveals $p_i$.  We
call $u$ the \emph{raw-execution cap}, or simply the \emph{cap}.  A
clairvoyant scheduler completes a job by the shorter of raw execution in
time $u$ and optimization followed by processing in time $1+p_i$; thus $u$
is also the \emph{offline cap} on the resulting job length.  We start with
the case in which optimization does not reveal $p_i$.  Before execution,
all optimized jobs still look the same.

This is the simpler case.  Let us first describe the hardest inputs for both
deterministic and randomized algorithms.  Roughly a $1/\sqrt u$ fraction of
the jobs cannot be shortened at all: they have $p_i=u$.  The remaining jobs
are optimized to length zero.  Even if an online algorithm knows these
proportions, it cannot distinguish the two types before committing to a job.
An optimized job therefore has expected length of order
$(1/\sqrt u)u=\sqrt u$, giving total completion cost
$\Theta(n^2\sqrt u)$.  The clairvoyant optimum first optimizes the zero-length
jobs and then runs the hard jobs raw, paying only $\Theta(n^2)$.

Randomization improves the asymptotic ratio by a factor of two, because a
private shuffle prevents the hard jobs from being concentrated at the
beginning of the schedule.  The exact ratios are as follows.

\begin{theorem}[Blind-optimization worst-case curves, informal version of
\protect\cref{thm:bo-det-curve,thm:bo-curve}]
\label{thm:intro-bo-det-curve}
\label{thm:intro-bo-curve}
For the blind-optimization model as $u\to\infty$, the exact deterministic and randomized
asymptotic competitive ratios satisfy
\[
 \RdetBO(u)=\sqrt u+o(1),
 \qquad
 \RrandBO(u)=\frac{\sqrt u}{2}+\frac12+o(1).
\]
Their complete formulas are given in
\cref{thm:bo-det-curve,thm:bo-curve}.
\end{theorem}

This poor worst-case behavior is driven by a very specific family of inputs,
which again raises the question of whether the worst-case ratio is the right
way to evaluate this model.  Instance optimality gives a more informative
answer.  For every fixed $u$, the simple algorithm in
\cref{alg:intro-blind-optimization} is instance-optimal.  The structure of
the optimal algorithm is particularly simple: at the leading $n^2$ scale, it
always either optimizes every job or optimizes none.  The algorithm simply
samples the input and chooses between these two options.

\begin{theorem}[Blind-optimization instance optimality, informal version of
\protect\cref{thm:bo-instance}]
\label{thm:intro-bo-instance}
Consider the blind-optimization model with some $0<u<\infty$.  There is one
randomized algorithm $\mathcal A^*$ such that, for every input
$p\in[0,u]^n$ and every online algorithm $\mathcal A'$,
\[
 \EE\ALG_{\mathcal A^*}\le \EE\ALG_{\mathcal A'}+o_u(n^2).
\]
\end{theorem}

\begin{algorithm}[H]
\caption{}
\label{alg:intro-blind-optimization}
\begin{algorithmic}
\State Shuffle the jobs, optimize and run $o(n)$ of them,
and let $\widehat\mu$ be the mean observed length.
\State If $1+\widehat\mu<u$, optimize and then run every remaining job.
\State Otherwise run every remaining job raw for time $u$.
\end{algorithmic}
\end{algorithm}

\subsection{Revealing optimization}
\label{sec:intro-revealing-optimization}

Now assume that optimization also reveals $p_i$.  Thus the unit action buys
both information and speed.  We call this \emph{revealing optimization}.

The two extreme regimes are easy to understand.  If $u\leq1$, the test alone
costs at least as much as raw execution, so we should never test.  If $u$
gets very large, raw execution is unattractive, so the problem approaches
obligatory testing.  The interesting question is what happens between these
regimes.

This model was introduced by D{\"u}rr et
al.~\cite{DurrErlebachMegowMeissner2020}.  Their results showed that
intermediate values of $u$ are the difficult ones.  We determine the exact
competitive ratio for every $u$, for both deterministic and randomized
algorithms.

\begin{theorem}[Revealing-optimization curves, informal version of
\protect\cref{thm:optional-curve,thm:random-common-upper}]
\label{thm:intro-common-upper-curves}
For every constant $u>0$, the exact deterministic and randomized asymptotic
competitive ratios $\RdetRO(u)$ and $\RrandRO(u)$ are the six-piece and
four-piece curves derived in
\cref{sec:optional,sec:random-common-upper}, respectively (see also
\cref{fig:intro-curves-det,fig:intro-curves-rand}).
The deterministic curve reaches
$1.86\ldots$ and then falls to the obligatory value
$1.57\ldots$.  The randomized curve reaches $1.62\ldots$ and
eventually becomes $4/3$.
\end{theorem}

\begin{figure}[H]
\centering
\begingroup
\tracinglostchars=0
\pgfmathdeclarefunction{introRhoI}{1}{%
  \pgfmathparse{%
    (-(((#1-1)^3)+((#1-1)^2)+1)
      +sqrt((((#1-1)^3)+((#1-1)^2)+1)^2
        +4*(#1-1)*(((#1-1)^2)+(#1-1)+1)*((#1-1)+2)))
      /(2*(#1-1))}%
}
\begin{tikzpicture}
\begin{axis}[
  width=0.88\textwidth,
  height=0.48\textwidth,
  xmin=0.2,
  xmax=7.3,
  ymin=0.97,
  ymax=2,
  xlabel={raw execution time $u$},
  ylabel={competitive ratio},
  xtick={1,2,3,4,5,6,7},
  ytick={1,1.25,1.5,1.75,2},
  tick label style={font=\footnotesize},
  label style={font=\small},
  grid=major,
  grid style={black!15},
  axis line style={black!65},
  tick style={black!65},
  axis x line*=bottom,
  axis y line*=left,
  samples=100,
  no markers,
  clip=true,
  legend style={
    font=\footnotesize,
    draw=none,
    fill=none,
    at={(0.5,1.04)},
    anchor=south,
    legend columns=3,
    /tikz/every even column/.append style={column sep=12pt},
  },
]

\addlegendimage{violet!82!black,densely dotted,line width=1.5pt}
\addlegendentry{\textsc{Raw}}
\addlegendimage{teal!78!black,dashed,line width=1.5pt}
\addlegendentry{\textsc{ForcedPrefixUTE}}
\addlegendimage{red!78!black,dash dot,line width=1.5pt}
\addlegendentry{\textsc{AdaptiveThreshold}}

\addplot[violet!82!black,densely dotted,line width=1.5pt,domain=0.2:1]
  {1};
\addplot[violet!82!black,densely dotted,line width=1.5pt,
  domain=1:1.86676039917]
  {x};

\addplot[teal!78!black,dashed,line width=1.5pt,
  domain=1.86676039917:3.14789903570] {introRhoI(x)};
\addplot[teal!78!black,dashed,line width=1.5pt,
  domain=3.14789903570:3.61803398875] {1+1/sqrt(x-1)};

\addplot[red!78!black,dash dot,line width=1.5pt,smooth] coordinates {
  (3.61803399,1.61803399)
  (3.63837920,1.61566099)
  (3.65923107,1.61328800)
  (3.68063111,1.61091500)
  (3.70262731,1.60854201)
  (3.72527568,1.60616902)
  (3.74864222,1.60379602)
  (3.77280560,1.60142303)
  (3.79786084,1.59905003)
  (3.82392445,1.59667704)
  (3.85114178,1.59430404)
  (3.87969805,1.59193105)
  (3.90983519,1.58955805)
  (3.94187917,1.58718506)
  (3.97628649,1.58481206)
  (4.01372962,1.58243907)
  (4.05526944,1.58006608)
  (4.10275400,1.57769308)
  (4.15995362,1.57532009)
  (4.23729031,1.57294709)
  (4.50524150,1.57057410)
};
\addplot[red!78!black,dash dot,line width=1.5pt,
  domain=4.50524149579:7.3] {1.57057409665};

\addplot[black!45,densely dotted,thin] coordinates {(1,0.97) (1,2)};
\node[font=\footnotesize,rotate=90,anchor=south]
  at (axis cs:1,1.18) {$\uDet{1}$};
\addplot[black!45,densely dotted,thin]
  coordinates {(1.86676039917,0.97) (1.86676039917,2)};
\node[font=\footnotesize,rotate=90,anchor=south]
  at (axis cs:1.86676039917,1.18) {$\uDet{2}$};
\addplot[black!45,densely dotted,thin]
  coordinates {(3.14789903570,0.97) (3.14789903570,2)};
\node[font=\footnotesize,rotate=90,anchor=south]
  at (axis cs:3.14789903570,1.18) {$\uDet{3}$};
\addplot[black!45,densely dotted,thin]
  coordinates {(3.61803398875,0.97) (3.61803398875,2)};
\node[font=\footnotesize,rotate=90,anchor=south]
  at (axis cs:3.61803398875,1.18) {$\uDet{4}$};
\addplot[black!45,densely dotted,thin]
  coordinates {(4.50524149579,0.97) (4.50524149579,2)};
\node[font=\footnotesize,rotate=90,anchor=south]
  at (axis cs:4.50524149579,1.18) {$\uDet{5}$};
\end{axis}
\end{tikzpicture}
\endgroup
\caption{The three deterministic algorithms, each shown only on the range
where it gives the best guarantee.}
\label{fig:intro-curves-det}
\end{figure}
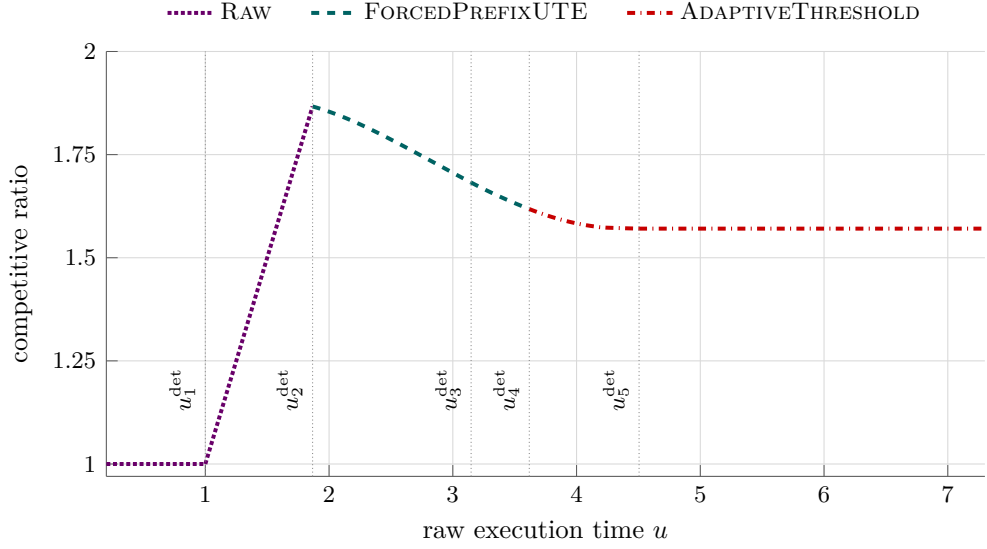

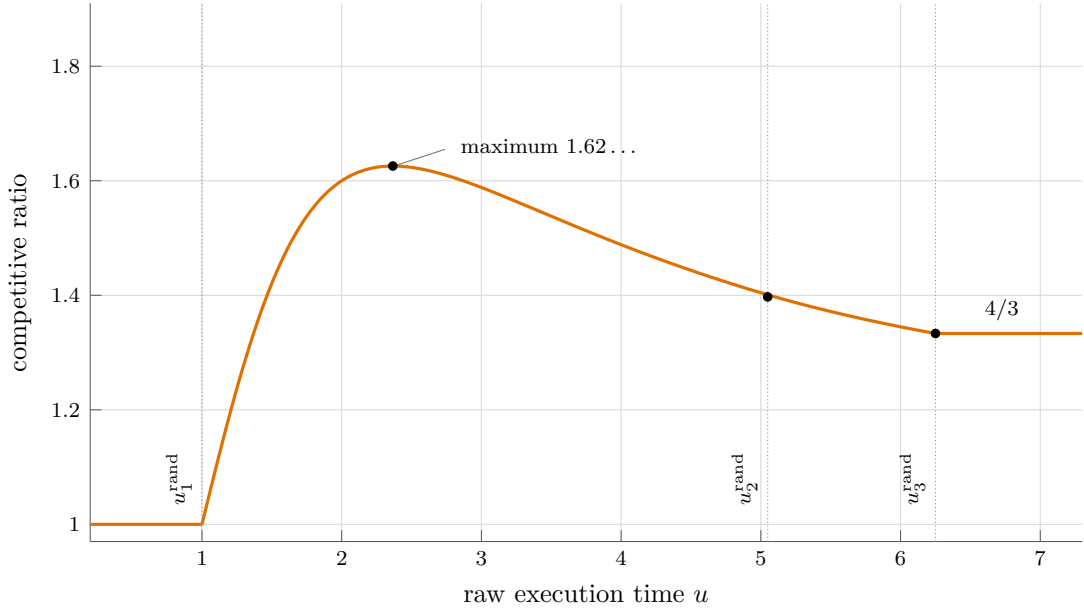
\begin{figure}[H]
\centering
\begingroup
\tracinglostchars=0
\begin{tikzpicture}
\begin{axis}[
  width=0.98\linewidth,
  height=0.58\linewidth,
  xmin=0.2,
  xmax=7.3,
  ymin=0.97,
  ymax=1.91,
  xlabel={raw execution time $u$},
  ylabel={competitive ratio},
  xtick={1,2,3,4,5,6,7},
  ytick={1,1.2,1.4,1.6,1.8},
  tick label style={font=\scriptsize},
  label style={font=\small},
  grid=major,
  grid style={black!15},
  axis line style={black!65},
  tick style={black!65},
  axis x line*=bottom,
  axis y line*=left,
  samples=120,
  no markers,
  clip=true,
]

\addplot[orange!88!black,very thick,domain=0.2:1] {1};
\addplot[orange!88!black,very thick,domain=1:5.04891733952]
  {x^3/(x^2+(x-1)^3)};
\addplot[orange!88!black,very thick,domain=5.04891733952:6.25]
  {1+1/(2*(sqrt(x)-1))};
\addplot[orange!88!black,very thick,domain=6.25:7.3] {4/3};
\addplot[black!45,densely dotted,thin] coordinates {(1,0.97) (1,1.91)};
\node[font=\scriptsize,rotate=90,anchor=south]
  at (axis cs:1,1.08) {$\uRand{1}$};
\addplot[black!45,densely dotted,thin]
  coordinates {(5.04891733952,0.97) (5.04891733952,1.91)};
\node[font=\scriptsize,rotate=90,anchor=south]
  at (axis cs:5.04891733952,1.08) {$\uRand{2}$};
\addplot[black!45,densely dotted,thin]
  coordinates {(6.25,0.97) (6.25,1.91)};
\node[font=\scriptsize,rotate=90,anchor=south]
  at (axis cs:6.25,1.08) {$\uRand{3}$};
\addplot[only marks,mark=*,mark size=1.6pt,black] coordinates {
  (2.36602540378,1.62575238458)
  (5.04891733952,1.39719767662)
  (6.25,1.33333333333)
};
\node[font=\scriptsize,anchor=west] at (axis cs:2.78,1.66)
  {maximum $1.62\ldots$};
\draw[black!55,thin] (axis cs:2.74,1.655) -- (axis cs:2.39,1.627);
\node[font=\scriptsize,anchor=west] at (axis cs:6.53,1.375) {$4/3$};

\end{axis}
\end{tikzpicture}
\endgroup
\caption{The exact randomized revealing-optimization curve.}
\label{fig:intro-curves-rand}
\end{figure}

The deterministic curve is realized by three explicit algorithm families,
used on consecutive ranges of $u$.  We describe them here and defer the exact
parameters to \cref{sec:optional}.

The first algorithm, \textsc{Raw}$(u)$, simply fixes an arbitrary job order
and executes every job untested for the public raw time $u$.  On its relevant
range, its ratio is $1$ for $u\le\uDet{1}$ and $u$ for
$\uDet{1}\le u\le\uDet{2}$; see \cref{fig:intro-curves-det}.

The second algorithm is \textsc{ForcedPrefixUTE}, where UTE stands for
\emph{uniform-threshold execution}.  As specified in
\cref{alg:intro-ro-ute}, it tests every job and normally processes only
outcomes below a fixed threshold.  It makes one deliberate exception: every
outcome in an initial $b(u)$-fraction of the test order is processed
immediately.  Intuitively, a deterministic adversary may reveal long jobs
first and shorter jobs only later.  Deferring all early long jobs would
create a tail that waits through most of the later tests.  The forced prefix
limits this risk by testing and immediately processing a controlled initial
fraction, at the cost of sometimes processing a long job too early.  After
the forced prefix, the same threshold is applied to every tested job.
The exact parameter $b(u)$ is defined in
\cref{eq:opt:ute-parameters}; it is nonnegative on
$[\uDet{2},\uDet{3}]$ and vanishes at $\uDet{3}$.

\begin{algorithm}[htbp]
\caption{}
\label{alg:intro-ro-ute}
\begin{algorithmic}
\Statex \textbf{Range:} $\uDet{2}\le u\le\uDet{4}$.
\State Fix the test order $1,\ldots,n$ and set $Q\gets\varnothing$.
\State Set $\theta_u\gets\min\{1,u-1\}$.  If
$\uDet{2}\le u\le\uDet{3}$, set $b\gets b(u)$; for
$\uDet{3}\le u\le\uDet{4}$, set $b\gets0$.
\State For $i=1,\ldots,n$, test job $i$ and observe $p_i$.
\State \quad If $i\le\lfloor bn\rfloor$ or $p_i\le\theta_u$, process
job $i$ immediately; otherwise add it to $Q$.
\State Process the jobs of $Q$ in nondecreasing revealed processing time.
\end{algorithmic}
\end{algorithm}

Finally, \textsc{AdaptiveThreshold} is the same history-dependent algorithm as
\cref{alg:intro-obligatory-det}.  For finite $u$, only its parameter $c$
changes: on $\uDet{4}\le u\le \uDet{5}$ it is chosen by the
equations in \cref{eq:opt:cm-parametrization}, after which it equals
the obligatory-testing value $0.57\ldots$.  The three algorithms cover
the deterministic curve in \cref{fig:intro-curves-det} in that order:
\textsc{Raw} first, \textsc{ForcedPrefixUTE} in the middle, and
\textsc{AdaptiveThreshold} for large $u$.  The exact formulas and matching
constructions are in \cref{sec:optional}.

For randomized algorithms, private sampling again gives instance optimality
with an algorithm similar in shape to \cref{alg:intro-blind} -- instance-optimality of this algorithm also implies that this single algorithm defines the curve in \cref{fig:intro-curves-rand}. The exact formulas and
matching constructions are in \cref{sec:random-common-upper}.

\begin{theorem}[Revealing-optimization instance optimality, informal version of
\protect\cref{thm:common-upper-instance}]
\label{thm:intro-common-upper-instance}
Fix $0<u<\infty$.  There is one randomized algorithm $\mathcal A^*$ such that,
for every input $p\in[0,u]^n$ and every online algorithm
$\mathcal A'$,
\[
 \EE\ALG_{\mathcal A^*}\le \EE\ALG_{\mathcal A'}+o_u(n^2).
\]

\end{theorem}

\begin{algorithm}[H]
\caption{}
\label{alg:intro-common-upper}
\begin{algorithmic}
\State Shuffle the jobs, test $O(\sqrt n)$ of them, and use the
observed lengths to estimate the input distribution.
\State Choose a threshold $t$ maximizing
$\widehat a/(1+\widehat\ell)$, where
$\widehat a=\widehat D([0,t])$ and
$\widehat\ell=\int_{[0,t]}p\,d\widehat D(p)$.
\State Put $\widehat\tau=(1+\widehat\ell)/\widehat a$.  If
$\widehat\tau\ge u$, set $q=0$; otherwise choose $q$ minimizing
$F_{\mathsf{RO},u,\widehat D}(q)$ from \cref{eq:cui-F}.
\State Test a random $q$-fraction; finish outcomes below
$\widehat\tau$ at once.
\State Finish all other tested jobs, shortest first.
\State Run every untouched job raw for time $u$.
\end{algorithmic}
\end{algorithm}

Notice that the two exact revealing-optimization curves from \cref{fig:intro-curves-det,fig:intro-curves-rand} are not monotone.
This is not a contradiction.  For very small $u$, testing is
useless, while as $u\to\infty$ the model converges to the obligatory-testing
setting of \cref{sec:obligatory}.  Both extreme cases are comparatively easy:
the online cost can be charged directly to the optimum.  Intermediate values
are harder because raw execution gives the clairvoyant scheduler additional
useful options.

Taken together, the revealing- and blind-optimization curves
quantify the value of advance information: with revelation the ratios remain
bounded as $u$ grows, whereas without revelation they grow as
$\Theta(\sqrt u)$.

\subsection{Techniques}

The proofs repeatedly use similar ideas.  We summarize the main ones here.

\paragraph{Remaining-job and pair accounting.}
Every unit of machine work delays every job that is still unfinished. 
Total completion time can therefore be computed as a summation over all pairs of jobs, summing the length of the one going first. 

This view can help us optimize the numerical parameters of the algorithms. For example, if an algorithm tests a fraction $q$ of the jobs, a pair of jobs may contain two tested jobs, one tested and
one untested job, or two untested jobs.  The numbers of such pairs are
proportional to $q^2$, $q(1-q)$, and $(1-q)^2$, respectively. The leading cost is therefore a quadratic function of $q$ -- the optimal parameter can thus be found solving a quadratic equation In general, computing the optimal parameters of our algorithms often leads to a lot of algebra -- going through it may be quite tiring, but it arises from arguments like this one. 

\paragraph{Completion density.}
Our algorithms can typically be viewed as a sequence of a few
\emph{modules}.  A module repeatedly performs the same action;
for example, it may test a random job, process it below a threshold, and
defer it otherwise.  Each module spends some work and completes some mass of
jobs.  Its \emph{completion density} is the completed mass per unit of work.
Intuitively, higher-density modules should come first. Our
algorithms typically identify the relevant modules and order them by
completion density.

\paragraph{A private shuffle.}
The randomized algorithms start by shuffling the labels and inspecting a
small prefix.  Its sublinear cost is negligible at the $n^2$ scale, while
concentration makes it a reliable estimate of the input distribution.
Because the labels were shuffled privately, the distribution of the
remaining jobs stays close to this estimate throughout most of the schedule.
This stability makes the randomized setting substantially simpler; in particular, we introduce a \emph{fluid model} -- a substantially simplified, idealized model, in which we disregard problems like that the distribution of remaining jobs changes over time. We analyze this simplified model and use concentration inequalities to reduce our models to it. 

\subsection{Related work}
\label{sec:related}

Our problems belong to the area of online algorithms and competitive
analysis; standard references are the textbook of Borodin and
El-Yaniv~\cite{BorodinElYaniv1998} and the survey of Buchbinder and
Naor~\cite{BuchbinderNaor2009}.  Classical examples include paging and list
update~\cite{SleatorTarjan1985}, online bipartite
matching~\cite{KarpVaziraniVazirani1990}, and online facility
location~\cite{Meyerson2001}.  As in this literature, one algorithm must work
for every input and is compared with an offline optimum.  Our randomized
results use the oblivious-adversary version of this comparison.  

\pagebreak[3]

Instance optimality is one approach in the broader area of beyond-worst-case
analysis; see the volume edited by Roughgarden~\cite{Roughgarden2021}.
Related notions have been studied for aggregation~\cite{FaginLotemNaor2003},
adaptive set operations~\cite{DemaineLopezOrtizMunro2000}, computational
geometry~\cite{AfshaniBarbayChan2017}, and sampling and sequential
estimation~\cite{NarayananRozhonTetekThorup2024}.  In graph algorithms,
universal optimality asks one algorithm to match the best possible bound for
every fixed graph topology; Dijkstra's algorithm with an appropriate heap has
this property~\cite{HaeuplerHladikRozhonTarjanTetek2024}.

Scheduling with tests belongs to optimization with \emph{explorable
uncertainty}, where uncertain data can be queried at a cost; the query
viewpoint goes back at least to Kahan~\cite{Kahan1991}.  Levi, Magnanti, and
Shaposhnik study a stochastic scheduling model and derive prefix-testing
algorithms followed by an order based on expected processing time per unit
weight~\cite{LeviMagnantiShaposhnik2019,LeviMagnantiShaposhnik2024}.  Our
four-block blind-execution algorithm has a similar shape, but our inputs are
adversarial, need not converge to a distribution, and are not announced to
the universal algorithm.

D{\"u}rr et al.~\cite{DurrErlebachMegowMeissner2020} introduced the
adversarial revealing-optimization model.  Their summary already
contains a deterministic lower bound $1.85\ldots$, a randomized lower bound
$1.62\ldots$, a randomized upper bound $1.74\ldots$, an arbitrary-uniform
deterministic upper bound $1.93\ldots$, and an extreme-uniform upper bound
$1.86\ldots$.  In particular, $1.62\ldots$ and $1.86\ldots$ numerically anticipate the
global maxima of our randomized and deterministic curves, although they
previously appeared respectively as a lower bound and as an upper bound for
a restricted input class.  Our contribution is the matching randomized
upper bound, the exact per-$u$ diagrams for arbitrary $p_i\in[0,u]$, the
improved arbitrary-uniform deterministic guarantee, and the general
instance-specific fluid characterization.

For clarity, the main comparison is:
\begin{center}
\small
\renewcommand{\arraystretch}{1.12}
\begin{tabular}{@{}L{0.22\linewidth}L{0.22\linewidth}L{0.22\linewidth}L{0.25\linewidth}@{}}
\toprule
setting & previous lower & previous upper & this paper \\
\midrule
$\mathsf{OT}$ deterministic
 & $3/2$~\cite{LiangLiang2026}
 & $1.58\ldots$~\cite{DogeasErlebachLiang2024}
 & exact $1.57\ldots$ \\
$\mathsf{OT}$ randomized
 & no matching result & no matching result & exact $4/3$ \\
$\mathsf{RO}$ deterministic
 & $1.85\ldots$~\cite{DurrErlebachMegowMeissner2020} globally
 & $1.93\ldots$~\cite{DurrErlebachMegowMeissner2020} arbitrary-uniform
 & exact per-$u$ curve;
   global maximum $1.86\ldots$ \\
$\mathsf{RO}$ randomized
 & $1.62\ldots$~\cite{DurrErlebachMegowMeissner2020}
 & $1.74\ldots$~\cite{DurrErlebachMegowMeissner2020}
 & exact per-$u$ curve;
   global maximum $1.62\ldots$ \\
\bottomrule
\end{tabular}
\end{center}
Related work allows nonuniform test
times~\cite{AlbersEckl2021,LiuLiuWongZhang2023,KrekelbergEtAl2026}, multiple
machines~\cite{GongChenHayashi2024}, or weights~\cite{BuldSchulz2026}.

Dogeas, Erlebach, and Liang~\cite{DogeasErlebachLiang2024} introduced
obligatory testing for total completion time.  For uniform tests they proved
a deterministic upper bound of $1.58\ldots$ and the lower bound
$\sqrt2$.  Liang and Liang~\cite{LiangLiang2026} later raised the lower bound
to $3/2$.  Our matching value is $1.57\ldots$, and our randomized
result gives the oblivious value $4/3$.  More strongly, the resulting
fixed-threshold algorithm is asymptotically optimal separately for every
bounded input.

The processing-time oracle of Dufoss{\'e} et
al.~\cite{DufosseEtAl2022} is the closest comparison to blind execution, but
it studies a deterministic minimax game with binary hidden durations and
different announced parameters.  Its admissible length set, adversary,
benchmark, and treatment of randomization differ from ours, so neither model
directly subsumes the other.  If
no tests are available and processing times are revealed only at completion,
one obtains the classical nonclairvoyant model of Motwani, Phillips, and
Torng~\cite{MotwaniPhillipsTorng1994}, in which the scheduler does not know a
job's length before it completes.  These semantics differ from raw
execution in revealing optimization, which always lasts the declared time
$u$.

\subsection{Roadmap}
\label{sec:roadmap}

\Cref{sec:model} formally defines the four models, their information
structures and adversaries, the corresponding offline optima, and our
asymptotic competitive-ratio convention.

\Cref{sec:common-fluid-envelopes} develops the common bounded-input
fluid optimum and the transfer from its completion curve to a learned finite
schedule.  Its specializations in
\cref{sec:common-upper-instance,sec:obligatory-instance,sec:blind} give the
instance-optimal benchmarks for revealing optimization, obligatory testing,
and blind execution.

\Cref{sec:optional} proves the exact deterministic revealing-optimization
curve, and \cref{sec:obligatory} derives the obligatory endpoint with bounds
whose constants do not depend on $u$.

\Cref{sec:random-common-upper} maximizes the revealing benchmark against the
clairvoyant optimum and proves the exact randomized curve.  Its final
subsection, \cref{sec:random-obligatory}, gives the unbounded obligatory
endpoint.

Finally, \cref{sec:blind-optimization} studies speedup without revelation and
proves the deterministic and randomized blind-optimization curves together
with randomized instance optimality.

All main exact-ratio, instance-optimality, and impossibility results in this
paper have been proved in Lean~4 and are available at
\url{https://github.com/vaclavrozhon/obligatory-testing}.

\section{Models and common preliminaries}\label{sec:model}

The four models differ in small but important ways.  This section
collects their common building blocks so that the later arguments need not
repeat them.  We first define the models and adversaries, then develop the
offline pair identities and asymptotic conventions, and finally record the
probabilistic tools used throughout the paper.  The common fluid problem is
developed separately in \cref{sec:common-fluid-envelopes}.

\subsection{Models}

There are $n$ jobs, all available at time zero, and one machine.  Operations
are nonpreemptive: once an operation starts, it runs to completion.  At most
one operation can run at a time.  Job $j$ has
a hidden value $p_j\ge0$.  The four models differ in what can be done before
this value is known.

\begin{description}[leftmargin=*,style=nextline]
\item[Obligatory testing $\mathsf{OT}$.]
The values $p_j$ are arbitrary finite nonnegative numbers.  Every job must
first be tested for one unit.  The test reveals $p_j$, after which its
processing operation of length $p_j$ may be performed immediately or
delayed.

\item[Blind execution $\mathsf{BE}$.]
The values $p_j$ are arbitrary finite nonnegative numbers.  A job may be run
blindly for its actual duration $p_j$;
the value becomes known only when that execution finishes.  Alternatively,
the job may be tested for one unit, revealing $p_j$, and its processing
operation may then be delayed.  Testing provides information but does not
shorten the job.

\item[Blind optimization $\mathsf{BO}(u)$.]
Here $p_j\in[0,u]$.  A job may be run raw for the known time $u$.
Alternatively, one unit of optimization changes its running time from $u$
to $p_j$ without revealing $p_j$; the optimized job is subsequently run
blindly, possibly after a delay.  Its value becomes observable only when that
execution finishes.  Raw execution reveals no hidden value.

\item[Revealing optimization $\mathsf{RO}(u)$.]
Again $p_j\in[0,u]$, and raw execution takes the known time $u$.  The other
option is a unit test that reveals $p_j$, followed, possibly after a delay,
by processing for $p_j$.  Thus the preliminary action provides both the
speedup and information about its size.
\end{description}

The notation $\mathsf{BE}$ has no model parameter: blind execution permits
arbitrary finite nonnegative processing times.  For its bounded
instance-optimal theorem we restrict the inputs to $[0,L]$ and make $L$
public.  In the last two models, by contrast, $u$ is itself an available
running time and hence part of every scheduling decision.  We call it the
\emph{raw-execution cap}, or simply the \emph{cap}.  Since the offline
effective length in these models is $\min\{u,1+p_j\}$, we also call $u$ the
\emph{offline cap} when discussing the clairvoyant benchmark.  The
number $n$
and any theorem-specific bound $L$ or model parameter $u$ are known to the
algorithm.

The completion time $C_j$ is the time at which job $j$ finishes.  If a test
reveals $p_j=0$, the job completes at the end of that test.  A blind zero job
in $\mathsf{BE}$ completes as soon as it is selected, whereas an optimized
zero job in $\mathsf{BO}(u)$ first pays its unit optimization cost.  The
objective is
\[
                              \sum_{j=1}^n C_j.
\]

Obligatory testing can be viewed as the limit of revealing optimization in
which raw execution becomes unavailable.  We use this only as intuition.
The obligatory theorems and their algorithms are proved directly; deriving
them from revealing optimization would require arguing about both $u$ and $n$ going to infinity, which we avoid.

We write $\ALG_{\mathcal A}(I)$ for the objective value of an online
algorithm $\mathcal A$ on an instance $I$, and $\OPT(I)$ for the value of an
offline optimum that knows all processing times.  We omit algorithm
subscripts or instance arguments when they are clear.
When a sequence of costs has the form $cn^2+o(n^2)$, we call $c$ its
\emph{leading coefficient}; terms contributed by individual jobs rather
than pairs are typically only $O(n)$.

\subsection{Input, information and adversaries}
\label{par:input-convention}

For a deterministic algorithm, the adversary may choose all processing times
and their assignment to labels with full knowledge of the algorithm.  The
randomized results use an \emph{oblivious} adversary: the input is fixed before
the algorithm draws its private random seed.

An input $p=(p_1,\ldots,p_n)$ consists of job lengths chosen by the
adversary.  For a randomized algorithm, this input is then relabeled by a
uniform permutation.  We abbreviate the expectation over this shuffle and
all further private randomness by $\EE\ALG_{\mathcal A}(p)$.  Thus randomized
performance depends only on the multiset of job lengths, not on their displayed
order.  This ``private shuffle'' is an action of the algorithm at time zero:
the oblivious adversary fixes the labeled input first, and the algorithm
then privately permutes the otherwise indistinguishable labels.  It is not a
random-order input assumption or a benevolent ordering supplied by the
environment.  We distinguish three levels of information.
\begin{itemize}
\item A \emph{clairvoyant offline algorithm} knows the entire labeled input,
including the assignment of processing times to labels.  Its minimum cost is
$\OPT(p)$.
\item An \emph{announced online algorithm} knows the multiset
$\{p_1,\ldots,p_n\}$, but not its assignment to labels.  Thus it knows the
input distribution exactly, but not which value belongs to the next job it
chooses.
\item An \emph{unannounced}, or \emph{fully online}, algorithm initially knows
only the public parameters, such as $n$, $L$, and $u$, and receives no
input-specific information.  Our instance-optimal algorithms belong to this
class.
\end{itemize}
Both online classes are \emph{nonanticipating}: each action may depend on the
public parameters, private randomness, and everything revealed so far, but
not on a hidden value that the model has not yet exposed.
The \emph{transcript} of a run is the resulting ordered sequence of actions
and observations.

Mathematically we allow arbitrary private probability spaces.  Expected-cost
comparisons are understood for integrable real-valued run costs; a policy with
infinite expected cost cannot violate a lower bound.  The Lean development
integrates over the general seed law directly, so no finite-support
approximation of the private randomness is needed.

For some lower bounds it is convenient to describe a finite distribution
over inputs before selecting one fixed input.  If every deterministic
algorithm has average payoff at least $L$ under this distribution, then the
same is true after conditioning on the private seed of a randomized
algorithm.  Averaging over that seed shows that at least one fixed input in
the support gives the randomized algorithm expected payoff at least $L$.
This is the finite one-sided form of Yao's principle~\cite{Yao1977}.

We write
\begin{equation}
 D[p]\defeq\frac1n\sum_{i=1}^n\delta_{p_i}
 \label{eq:empirical-distributions}
\end{equation}
for the empirical distribution of the input, where $\delta_x$ denotes one
unit of probability concentrated at $x$.  It is unchanged by the
shuffle.

Our instance-optimality results compare a single unannounced algorithm against
all announced online algorithms on this shuffled input.  This is ordinary
instance optimality when the input is viewed as a multiset, and order-oblivious
instance optimality when public-label order is treated as part of the input; see
\cite{AfshaniBarbayChan2017,NarayananRozhonTetekThorup2024}.
Every such comparison is in expectation, at the leading $n^2$ scale: the
universal algorithm does not know the multiset, while its announced
competitor may be tailored to that multiset but not to the hidden placement.

\subsection{Offline optima and pair accounting}

The offline optimum knows all values $p_j$.  It therefore chooses the shorter
available way to complete each job and orders the resulting job blocks
shortest first.  Associate with job $j$ the model-dependent
\emph{effective length}
\begin{equation}
 \lambda_j\defeq
 \begin{cases}
  1+p_j, & \mathsf{OT},\\
  p_j, & \mathsf{BE},\\
  \min\{u,1+p_j\}, & \mathsf{BO}(u)\text{ or }\mathsf{RO}(u).
 \end{cases}
 \label{eq:model-effective-length}
\end{equation}

\begin{lemma}[Effective lengths and the offline pair identity]
\label{lem:offline}\label{lem:optional-offline}
Fix any one of the four models and let
$\lambda_{(1)}\le\cdots\le\lambda_{(n)}$ be its effective lengths from
\cref{eq:model-effective-length}.  A clairvoyant (fully informed) optimum
completes the jobs
in this order and has value
\begin{align}
 \OPT
 &=\sum_{k=1}^n\sum_{i=1}^k\lambda_{(i)},
 \label{eq:unified-opt-prefix}\\
 &=\sum_i\lambda_i+\sum_{i<j}\min\{\lambda_i,\lambda_j\},
 \label{eq:optional-opt-pair}\\
 &=\frac12\left(
      \sum_i\lambda_i+\sum_{i,j}\min\{\lambda_i,\lambda_j\}
    \right).
 \label{eq:unified-opt-symmetric}
\end{align}
In $\mathsf{OT}$ it tests every job.  In $\mathsf{BE}$ it runs every job
blindly.  In $\mathsf{BO}(u)$ and $\mathsf{RO}(u)$ it uses raw mode when
$u\le1+p_j$ and the unit preliminary action followed by processing otherwise.
\end{lemma}

\begin{proof}
Let $C_{[1]}\le\cdots\le C_{[n]}$ be the ordered completion times of any
feasible schedule.  Completing job $j$ consumes at least $\lambda_j$ units of work
belonging to that job.  Hence, by time $C_{[k]}$, the machine has performed
the required work of some $k$ jobs, and
\[
 C_{[k]}\ge\sum_{i=1}^k\lambda_{(i)}.
\]
Choose the mode specified in the lemma, keep each resulting block contiguous,
and execute the blocks in nondecreasing order of $\lambda_j$.  This schedule is
feasible in the relevant model and attains equality for every $k$.  Summing
over $k$ and then collecting the contribution of each unordered pair gives
the first two formulas.  Separating the terms with $i=j$ from the two
orientations of every pair with $i\ne j$ gives the symmetric form.
\end{proof}

For later pair calculations, if $X\subseteq[n]$, we write
\begin{equation}
 \SPT(X)\defeq
 \sum_{i\in X}p_i
 +\sum_{\substack{i,j\in X\\i<j}}\min\{p_i,p_j\}
 \label{eq:spt-block-cost}
\end{equation}
be the exact cost of processing that block in shortest-processing-time
(SPT) order.  Thus \(\SPT(X)\) is the processing-time part of the offline
pair formula when the jobs in $X$ are executed without preliminary actions.

\subsection{Competitive-ratio conventions}

Every asymptotic statement in this paper allows a uniform additive
$o(n^2)$ term.  Formally, an algorithm family $\mathcal A$ has asymptotic
competitive ratio at most $R$ if there is a deterministic sequence
$\eps(n)\to0$ such that every $n$-job instance satisfies
\begin{equation}
 \EE\ALG_{\mathcal A}(I)
 \le R\cdot\OPT(I)+\eps(n)n^2.
 \label{eq:cr-infty}
\end{equation}
For a deterministic algorithm the expectation is omitted.  The optimal ratio
in a model is the infimum of such $R$ over all admissible algorithm families.
We also call this the \emph{size-asymptotic} ratio to emphasize that the
model parameters are fixed while the number of jobs $n$ tends to infinity.

In obligatory testing every effective length is at least one, and in either
optimization model it is at least $\min\{u,1\}$.  Hence
\cref{eq:unified-opt-prefix} gives $\OPT=\Omega(n^2)$ in these models, and the
definition above is equivalent to the usual worst-case ratio after taking
$n\to\infty$.  This equivalence fails for blind execution, where $\OPT$ can
be $o(n^2)$ in the special case where most of the jobs have length zero.  Accordingly,
\cref{thm:blind-instance} makes the absolute statement
$\EE\ALG=n^2\PhiBE(D[p])+o_L(n^2)$, where the benchmark $\PhiBE$ is defined
in \cref{sec:blind-benchmark}.  It does not claim a bounded ratio to the
clairvoyant optimum when most jobs have length zero.

\subsection{Concentration and predictable sampling}
\label{sec:probabilistic-tools}
\label{sec:standard-concentration}

We collect here the concentration estimates used in the randomized sections.
All of them are finite-space statements.  We use without further mention the
union bound, Markov's inequality
$\Prob(Y\ge t)\le \EE Y/t$ for $Y\ge0$, and the tail identity
$\EE Y=\int_0^\infty\Prob(Y\ge t)\,dt$.

We begin with Chernoff bounds for sums of independent Bernoulli variables.

\begin{lemma}[Chernoff bounds]
\label{lem:chernoff}
Let $X=\sum_{i=1}^m X_i$, where the $X_i$ are independent Bernoulli random
variables, and write $\mu=\EE X$.  For every $\delta>0$,
\begin{equation}
 \Prob\bigl(X\ge(1+\delta)\mu\bigr)
 \le \left(\frac{e^\delta}{(1+\delta)^{1+\delta}}\right)^\mu
 \le \exp\!\left(-\frac{\delta^2\mu}{2+\delta}\right).
 \label{eq:chernoff-upper}
\end{equation}
For $0\le\delta\le1$,
\begin{equation}
 \Prob\bigl(X\le(1-\delta)\mu\bigr)
 \le \exp\!\left(-\frac{\delta^2\mu}{2}\right).
 \label{eq:chernoff-lower}
\end{equation}
\end{lemma}

These are the standard exponential-moment bounds of
Chernoff~\cite{Chernoff1952}.

Most of our samples are drawn from a fixed input without replacement.  The
next lemma gives both a tail bound and the finite-population variance used
for such samples.

\begin{lemma}[Sampling without replacement]
\label{lem:sampling-without-replacement}
Let $x_1,\ldots,x_N\in[0,1]$, let $X_1,\ldots,X_N$ be a uniformly random
permutation of these numbers, and put
$\bar x=N^{-1}\sum_i x_i$.  For $1\le k\le N$ and $t\ge0$,
\begin{equation}
 \Prob\left(\left|\sum_{i=1}^kX_i-k\bar x\right|\ge t\right)
 \le 2\exp\!\left(-\frac{2t^2}{k}\right).
 \label{eq:without-replacement-hoeffding}
\end{equation}
Moreover, if $\widehat x_k=k^{-1}\sum_{i=1}^kX_i$ and
$\sigma^2=N^{-1}\sum_i(x_i-\bar x)^2$, then, for $N>1$,
\begin{equation}
 \operatorname{Var}(\widehat x_k)
 =\frac{N-k}{k(N-1)}\sigma^2
 \le\frac{1}{4k}.
 \label{eq:without-replacement-variance}
\end{equation}
\end{lemma}

\begin{proof}
The estimate \cref{eq:without-replacement-hoeffding} is Hoeffding's
sampling-without-replacement inequality~\cite[Section~6]{Hoeffding1963}.
The variance identity is the elementary finite-population correction.
\end{proof}

Applying the scalar variance bound to class indicators gives a concise
estimate for an empirical histogram.

\begin{lemma}[Histogram error without replacement]
\label{lem:without-replacement-histogram}
Partition a population of $N$ occurrences into $d$ classes.  Let $p$ be the
vector of class frequencies in the population and let $\widehat p$ be the
corresponding frequency vector among the first $k$ occurrences in a uniformly
random permutation.  Then
\begin{equation}
 \EE\|\widehat p-p\|_1\le\sqrt{\frac d k}.
 \label{eq:without-replacement-histogram}
\end{equation}
\end{lemma}

\begin{proof}
For a class of mass $p_j$, \cref{eq:without-replacement-variance} and
Cauchy--Schwarz give
$\EE|\widehat p_j-p_j|\le\sqrt{p_j/k}$.  Summing over the classes and using
$\sum_j\sqrt{p_j}\le\sqrt d$ proves the claim.
\end{proof}

An adaptive algorithm may stop at a data-dependent time.  We therefore also
need one event that controls every class at every possible prefix.

\begin{lemma}[Uniform prefix histogram bounds]
\label{lem:uniform-prefix-histogram}
In the same population, let $H_{k,j}$ count class $j$ among the first $k$
entries of a uniformly random permutation.  Then
\begin{align}
 \Prob\left(
   \max_{0\le k\le N}\max_{1\le j\le d}|H_{k,j}-kp_j|\ge t
 \right)
 &\le 2d(N+1)\exp\!\left(-\frac{2t^2}{N}\right),
 \label{eq:uniform-prefix-hoeffding}\\
 \EE\max_{0\le k\le N}\max_{1\le j\le d}|H_{k,j}-kp_j|
 &\le C\sqrt{N\ln(2d(N+1))}
 \label{eq:uniform-prefix-mean}
\end{align}
for a universal constant $C$.
\end{lemma}

\begin{proof}
Apply \cref{eq:without-replacement-hoeffding} to every binary class indicator
and union-bound over $k$ and $j$ to obtain
\cref{eq:uniform-prefix-hoeffding}.  Integrating that tail bound proves
\cref{eq:uniform-prefix-mean}.
\end{proof}

Predictable sampling introduces dependent increments rather than a fixed
prefix.  The following maximal inequalities provide the needed uniform
control.

\begin{lemma}[Martingale maximal inequalities]
\label{lem:martingale-maximal}
Let $(M_k,\mathcal F_k)_{k=0}^N$ be a real martingale with $M_0=0$; that is,
$\mathcal F_k$ records the information available by time $k$ and
$\EE[M_k\mid\mathcal F_{k-1}]=M_{k-1}$.
If $\EE|M_N|^2<\infty$, then Doob's inequality gives
\begin{equation}
 \EE\max_{0\le k\le N}|M_k|^2\le4\EE|M_N|^2.
 \label{eq:doob-L2}
\end{equation}
If instead $|M_k-M_{k-1}|\le b_k$ almost surely for deterministic $b_k$,
then maximal Azuma--Hoeffding gives, for every $t\ge0$,
\begin{equation}
 \Prob\left(\max_{0\le k\le N}|M_k|\ge t\right)
 \le2\exp\!\left(-\frac{t^2}{2\sum_{k=1}^N b_k^2}\right).
 \label{eq:maximal-azuma}
\end{equation}
\end{lemma}

We use the classical forms due to Doob~\cite{Doob1953} and
Azuma~\cite{Azuma1967}; the maximal version of the latter follows by stopping
the usual exponential bound when the displayed threshold is first crossed.

We finally record a specialized estimate that permits data-dependent weights
while sampling a uniformly permuted finite population and is uniform over
prefixes.

\begin{lemma}[Predictable sampling from a random permutation]
\label{lem:predictable-permutation-sampling}
Let $X_1,\ldots,X_N$ be a uniform random permutation of fixed vectors
$x_1,\ldots,x_N\in[0,B]^d$, and let
$\bar x=N^{-1}\sum_i x_i$.  Suppose $A_k\in[0,1]$ is chosen using only the
previous draws, before $X_k$ is exposed.  Then for
$\delta,\eta\in(0,1)$, with probability at least
$1-\eta$, simultaneously for every coordinate $\ell$ and every
$m\le(1-\delta)N$,
\begin{equation}
 \left|\sum_{k\le m}A_k(X_{k,\ell}-\bar x_\ell)\right|
 \le(1+\delta^{-1})E_{N,d,B,\eta},
 \qquad
 E_{N,d,B,\eta}\defeq
 B\sqrt{2N\ln\!\left(\frac{4d(N+1)}{\eta}\right)}.
 \label{eq:predictable-permutation-sampling}
\end{equation}
The statement remains valid after conditioning on independent private
randomness used to choose the predictable sequence.
\end{lemma}

\begin{proof}
For each coordinate, apply
\cref{eq:without-replacement-hoeffding} after scaling by $B$, and union-bound
over coordinates and prefix lengths.  Except with probability at most
$\eta/2$, every centered prefix sum has absolute value at most
$E=E_{N,d,B,\eta}$.  Before draw $k\le(1-\delta)N$, the mean
$\widetilde x_{k-1,\ell}$ of
the remaining population therefore satisfies
\[
 |\widetilde x_{k-1,\ell}-\bar x_\ell|\le \frac{E}{\delta N}.
\]
For each $\ell$,
\[
 M_{m,\ell}=\sum_{k\le m}A_k(X_{k,\ell}-\widetilde x_{k-1,\ell})
\]
is a martingale with increments bounded by $B$.  The maximal
Azuma--Hoeffding inequality \cref{eq:maximal-azuma}, followed by a union
bound over $\ell$, puts all $|M_{m,\ell}|$ below $E$ except with probability at
most $\eta/2$.  The accumulated difference between the remaining-population
mean and the original mean is at most $E/\delta$, which proves
\cref{eq:predictable-permutation-sampling}.  Conditioning and then averaging
proves the final assertion.
\end{proof}

\section{The fluid model and instance optimality}
\label{sec:common-fluid-envelopes}

In this section we prove the instance-optimality results for revealing
optimization, obligatory testing, and blind execution.  All three proofs use
the same fluid model.  We first solve this model exactly, then show that its
value gives the leading cost on every bounded input.

Let us explain the main idea.  Fix an input and privately shuffle its
values among the job labels.  Suppose the algorithm inspects a prefix whose
length may depend on the jobs seen so far, but not on values that are still
hidden.  The number of encountered jobs of each type is close to the prefix
length times its proportion in the whole input.  Thus, if the algorithm tests
a fraction $t$ of all jobs, it sees approximately a fraction $t$ of every
processing-time class.

The fluid model simply makes this approximation exact.  We replace $n$ jobs
by one unit of divisible mass and divide machine work by $n$.  For example,
testing mass $0.01$ reveals exactly a $0.01$ copy of the input distribution.
The resulting deterministic problem describes the leading $n^2$ term of
total completion time.  The approximation by actual shuffled jobs is proved
later in this section.

\subsection{The fluid model}

Here is a simple example.  Suppose that half of the jobs have length zero and
half have length two.  If a block tests mass $0.2$ and immediately completes
only the zero jobs, it uses $0.2$ units of work, completes mass $0.1$, and
leaves mass $0.1$ of length-two jobs for later.  The fluid model records such
aggregate blocks instead of individual jobs.

We consider only finitely many processing-time types, and a fluid algorithm
consists of finitely many divisible blocks.  Within each block, work and
completions grow proportionally.  We fix the total tested fraction $q$;
the remaining fraction $1-q$ is completed by another mode that uses $v$
units of work per job.  Formally, let
\begin{equation}
 D=\sum_{i=0}^m d_i\delta_{p_i},
 \qquad 0=p_0<p_1<\cdots<p_m,
 \qquad \sum_{i=0}^m d_i=1,
 \label{eq:fluid-finite-distribution}
\end{equation}
where $d_0$ may be zero.  Fix $q\in[0,1]$ and $v\ge0$.

A fluid schedule is a finite sequence of the following operations.

\begin{description}[leftmargin=*,style=nextline]
 \item[Testing block of mass $z$ with selector $g$.]
 The vector $g=(g_1,\ldots,g_m)\in[0,1]^m$ specifies which fraction of each
 positive type is processed directly after being revealed.  The block tests
 mass $z$, completes all revealed zero jobs and the selected positive jobs,
 and places the remaining revealed positive jobs in deferred queues.  It
 uses work
 \[
   z\left(1+\sum_{i=1}^m p_i d_i g_i\right)
 \]
 and completes mass
 \[
   z\left(d_0+\sum_{i=1}^m d_i g_i\right).
 \]
 It adds $zd_i(1-g_i)$ to the deferred queue of type $i$.  The total mass of
 all testing blocks may not exceed $q$.

 \item[Deferred-processing block of type $i$.]
 Processing mass $z$ from the currently available deferred queue of type
 $i$ uses work $p_i z$ and completes mass $z$.

 \item[Alternative-completion block.]
 Completing mass $z$ in the other mode uses work $vz$ and completes mass
 $z$.  The total mass of these blocks may not exceed $1-q$.
\end{description}

The schedule is complete if it tests total mass $q$, uses the other mode on
mass $1-q$, and eventually empties every deferred queue.  A positive-work
block is traversed linearly; a zero-work block completes its mass at once.
Consequently every finite fluid schedule has a piecewise-linear completion
curve $C(x)$, the mass completed after $x$ units of work.

The area under the unfinished mass is its total completion cost:
\begin{equation}
 \Area(C)\defeq\int_0^\infty(1-C(x))\,dx.
 \label{eq:completion-curve-area}
\end{equation}
Indeed, if a mass element $z$ completes after work $T(z)$, the elementary
layer-cake identity gives
\[
 \int T(z)\,dz
 =\int_0^\infty \operatorname{mass}\{z:T(z)>x\}\,dx
 =\Area(C).
\]
Thus continuous work is only a convenient coordinate for drawing the
completion curve; the schedule itself makes finitely many block decisions.
The shaded region in \cref{fig:fluid-completion-curve} shows its total
completion cost.

\begin{figure}[htbp]
\centering
\begin{tikzpicture}[x=1.35cm,y=3.1cm]
  \fill[teal!12]
    (0,1)--(4.3,1)--(4.3,0.98)--(3.6,0.82)--(3.0,0.72)--
    (2.35,0.38)--(1.55,0.38)--(0.85,0.16)--(0,0)--cycle;
  \draw[->,black!70] (0,0)--(4.55,0) node[right] {normalized work $x$};
  \draw[->,black!70] (0,0)--(0,1.08) node[above] {mass};
  \draw[black!45,densely dashed] (0,1)--(4.35,1)
    node[right,font=\small] {all jobs};
  \draw[teal!75!black,very thick]
    (0,0)--(0.85,0.16)--(1.55,0.38)--(2.35,0.38)--
    (3.0,0.72)--(3.6,0.82)--(4.3,0.98);
  \node[teal!70!black,font=\small,anchor=south west] at (1.55,0.42)
    {completion curve $C(x)$};
  \node[font=\small,align=center] at (2.75,0.9)
    {$\Area(C)=\int(1-C(x))\,dx$};
\end{tikzpicture}
\caption{A fluid completion curve.  Total completion cost is the shaded
area between completed mass and total mass.}
\label{fig:fluid-completion-curve}
\end{figure}
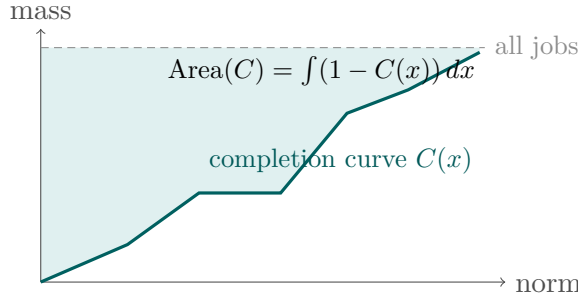

\paragraph{The best testing block.}

We next determine which positive types should be processed in the same block
in which they are tested.  For the finite fluid model, all integrals below
are finite sums under \cref{eq:fluid-finite-distribution}.  The proof is no
harder for an arbitrary finite-mean distribution, so we state the lemma in
that form.

For a selector $g$, let
\begin{equation}
 a_D(g)\defeq D(\{0\})+\int_{(0,\infty)}g(p)\,dD(p),
 \qquad
 w_D(g)\defeq 1+\int_{(0,\infty)}p g(p)\,dD(p).
 \label{eq:common-module-aw}
\end{equation}
These are the completed mass and work when one unit of mass is tested.
Their ratio $a_D(g)/w_D(g)$ is the block's \emph{completion density}.

For example, suppose half the mass has length zero and half has length one.
Tests alone complete mass $1/2$ using work $1$.  Processing the length-one
outcomes as well completes the full mass using work $3/2$, improving the
completion density from $1/2$ to $2/3$.  For longer positive outcomes the
same decision can reduce the density.  At the optimum, a positive type is
included precisely when its length is below the work per completion of the
resulting block.

\begin{lemma}[Maximum-density testing block]
\label{lem:maximum-density-module}
Let $D$ be a probability distribution on $[0,\infty)$ with finite mean.
There is a unique $\tau_D\ge1$ satisfying
\begin{equation}
 \int_{[0,\infty)}(\tau_D-p)^+\,dD(p)=1.
 \label{eq:common-threshold-equation}
\end{equation}
Choose $g_*$ to select every positive $p<\tau_D$, reject every
$p>\tau_D$, and select any fraction of the mass at $p=\tau_D$.  Write
$a_*=a_D(g_*)$ and $w_*=w_D(g_*)$.  Then
\begin{equation}
 \frac{a_*}{w_*}=\max_g\frac{a_D(g)}{w_D(g)}=\frac1{\tau_D},
 \qquad
 \tau_D a_D(g)\le w_D(g)\quad\text{for every }g.
 \label{eq:common-density-dual}
\end{equation}
\end{lemma}

\begin{proof}
The function
$h(t)=\int(t-p)^+\,dD(p)$ is continuous and nondecreasing, satisfies
$h(0)=0$, and tends to infinity because
$h(t)\ge t-\int p\,dD(p)$.  Once positive it is strictly increasing, so
$h(t)=1$ has a unique solution.  The inequality $h(t)\le t$ gives
$\tau_D\ge1$.

For an arbitrary selector $g$,
\[
 \begin{aligned}
  \tau_D a_D(g)-\int_{(0,\infty)}p g(p)\,dD(p)
   &=\tau_D D(\{0\})
     +\int_{(0,\infty)}(\tau_D-p)g(p)\,dD(p)\\
   &\le\int_{[0,\infty)}(\tau_D-p)^+\,dD(p)=1.
 \end{aligned}
\]
This is the second inequality in \cref{eq:common-density-dual}.  The
threshold selector $g_*$ attains equality because it keeps exactly the
positive terms; selecting mass at $p=\tau_D$ changes both sides equally.
Hence $w_*=\tau_Da_*$ and $g_*$ has maximum density $1/\tau_D$.
\end{proof}

For the rest of this section, return to the finite distribution in
\cref{eq:fluid-finite-distribution}.
The positive mass not completed in the maximum-density testing block forms
the finite residual measure
\begin{equation}
 d\nu(p)=(1-g_*(p))\one_{p>0}\,dD(p),
 \qquad \nu(\{0\})=0.
 \label{eq:common-module-data}
\end{equation}
It has total mass $1-a_*$ and support in $[\tau_D,\infty)$.  Thus every
deferred type has at least as much work per completion as the testing block.

\paragraph{Ordering the blocks.}

Once testing activity has been compressed into one block, only an ordering
problem remains.  The next elementary lemma is the divisible version of the
shortest-processing-time rule.

Represent a finite block of completion mass $c_k$ and work $v_k$ per
completion by $c_k\delta_{v_k}$.  For a collection
$\mathcal M=\sum_kc_k\delta_{v_k}$ of total mass one, define
\begin{equation}
 K_{\mathcal M}(x)\defeq
 \max\left\{\sum_kz_k:
       0\le z_k\le c_k,\ \sum_kv_kz_k\le x\right\}.
 \label{eq:divisible-spt-envelope}
\end{equation}
This finite fractional-knapsack problem asks for the largest mass that the
blocks can complete using work $x$.  Its value as a function of $x$ is their
completion curve.  We also write
\begin{equation}
 \SPT(\xi)\defeq
 \frac12\iint\min\{v,w\}\,d\xi(v)d\xi(w).
 \label{eq:divisible-spt-cost}
\end{equation}

\begin{lemma}[Shortest-first order for divisible blocks]
\label{lem:divisible-spt-area}
The curve $K_{\mathcal M}$ is realized by processing the blocks in
nondecreasing order of $v_k$, and
\begin{equation}
 \Area(K_{\mathcal M})=\SPT(\mathcal M).
 \label{eq:divisible-spt-area}
\end{equation}
\end{lemma}

\begin{proof}
Consider two consecutive blocks with masses $c,c'$ and work per completion
$v\le v'$.  Putting the $v$-block first changes the area, relative to the
opposite order, by $cc'(v-v')\le0$.  Repeated neighboring exchanges sort all
blocks by nondecreasing work per completion.

In a block of mass $c_k$ and work per completion $v_k$, unfinished mass
decreases linearly.  Once the blocks are sorted, its area contribution is
\[
 v_kc_k\left(1-\sum_{h<k}c_h-\frac{c_k}{2}\right).
\]
Summing these terms and collecting each unordered pair gives
\cref{eq:divisible-spt-area}.
\end{proof}

\subsection{The optimal fluid schedule}

We can now describe an optimal fluid schedule for the fixed distribution
$D$, tested fraction $q$, and alternative cost $v$.  We compare complete
curves, not just their areas: after every amount of work, our schedule has
completed at least as much mass as any other fluid schedule with the same
$D$, $q$, and $v$.

First describe the accounting state of an arbitrary schedule after it has
used work $x$.  Let $t$ be its tested mass, $b$ its mass completed in the
other mode, and
\[
 c=(c_1,\ldots,c_m)
\]
the vector of tested positive mass already processed in each class.  The
postponed mass of type $i$ is then $d_it-c_i$, so no additional state
variable is needed.  An algorithm may retain a much richer internal history,
but its aggregate state necessarily satisfies
\[
 0\le t\le q,\qquad 0\le b\le1-q,\qquad 0\le c_i\le d_it.
\]
It has completed mass $d_0t+b+\sum_i c_i$ and used work
$t+vb+\sum_i p_ic_i$.  Therefore define the finite-dimensional envelope
\begin{equation}
 \mathcal C^v_{D,q}(x)\defeq
 \max\left\{
 d_0t+b+\sum_{i=1}^m c_i:
 \begin{array}{l}
 0\le t\le q,\quad 0\le b\le1-q,\\
 0\le c_i\le d_it\quad(1\le i\le m),\\
 t+vb+\sum_{i=1}^m p_ic_i\le x
 \end{array}
 \right\}.
 \label{eq:common-fluid-envelope}
\end{equation}

The candidate schedule is completely explicit.

\begin{algorithm}[H]
\caption{}
\label{alg:fluid-spt}
\begin{algorithmic}
\State Compute the maximum-density selector $g_*$, its completed mass $a_*$,
and its work per completion $\tau_D$.
\State Create a testing block of completion mass $qa_*$ and work per
completion $\tau_D$; running it tests mass $q$ and uses selector $g_*$.
\State For every $p_i>0$, create a deferred block of mass
$qd_i(1-g_*(p_i))$ and work per completion $p_i$.
\State Create an alternative-completion block of mass $1-q$ and work per
completion $v$.
\State Order all nonempty blocks by nondecreasing work per completion,
placing the testing block before a deferred block in case of a tie, and
execute them in that order.
\end{algorithmic}
\end{algorithm}

Equivalently, the collection of blocks created by the algorithm is
\begin{equation}
 \mathcal M_{D,q,v}\defeq
 qa_*\delta_{\tau_D}+q\nu+(1-q)\delta_v.
 \label{eq:common-envelope-items}
\end{equation}

\begin{theorem}[Pointwise optimality of the finite fluid schedule]
\label{lem:test-module-envelope}
For every complete fluid schedule $S$ that tests mass $q$,
\begin{equation}
 C_S(x)\le \mathcal C^v_{D,q}(x)
 =K_{\mathcal M_{D,q,v}}(x)
 \qquad\text{for every }x\ge0.
 \label{eq:fluid-pointwise-optimality}
\end{equation}
The schedule $\textnormal{\textsc{FluidSPT}}(D,q,v)$ attains equality for
every $x$.  Consequently it minimizes fluid total completion time and
\begin{equation}
 \Area(\mathcal C^v_{D,q})=\SPT(\mathcal M_{D,q,v}).
 \label{eq:common-envelope-area}
\end{equation}
\end{theorem}

Before proving the theorem, let us summarize the argument.  The envelope
contains the state of every fluid algorithm.  The density inequality then
replaces all of its testing activity by the maximum-density testing block.
What remains is a collection of independent blocks, which should be run
shortest first.

\begin{claim}[Every schedule obeys the envelope]
\label{clm:fluid-envelope-upper}
For every complete fluid schedule $S$ and every $x\ge0$,
$C_S(x)\le\mathcal C^v_{D,q}(x)$.
\end{claim}

\begin{proof}
At work $x$, let $t,b,c_1,\ldots,c_m$ be the aggregate state described
above.  Every tested mass has type distribution $D$, so at most $d_it$ mass
of positive type $i$ has been revealed and hence $c_i\le d_it$.  The test,
alternative, and positive-processing work is respectively
$t$, $vb$, and $\sum_i p_ic_i$.  Thus the aggregate state is feasible in
\cref{eq:common-fluid-envelope}, whose objective is exactly its completed
mass.
\end{proof}

\begin{claim}[Replacing all tests by one testing block]
\label{clm:fluid-compression}
For every $x\ge0$,
$\mathcal C^v_{D,q}(x)\le K_{\mathcal M_{D,q,v}}(x)$.
\end{claim}

\begin{proof}
Fix a feasible point $(t,b,c_1,\ldots,c_m)$ in
\cref{eq:common-fluid-envelope}.  For each positive type, put
\[
 s_i=\min\{c_i,d_itg_*(p_i)\},
 \qquad r_i=c_i-s_i.
\]
The tests, tested zeros, and masses $s_i$ complete
\[
 y_{\rm test}=d_0t+\sum_i s_i
\]
using work
\[
 x_{\rm test}=t+\sum_i p_is_i.
\]
If $t>0$, apply \cref{eq:common-density-dual} to the selector
$s_i/(d_it)$, taking it to be zero when $d_i=0$.  This gives
\[
 \tau_Dy_{\rm test}\le x_{\rm test},
 \qquad y_{\rm test}\le a_*t\le a_*q.
\]
The same inequalities are trivial for $t=0$.  Hence the testing block of
\cref{alg:fluid-spt} can complete mass $y_{\rm test}$ using no more than
$x_{\rm test}$ work.

Moreover,
\[
 0\le r_i\le d_it(1-g_*(p_i))
        \le qd_i(1-g_*(p_i)),
\]
so the remaining processed mass fits inside the deferred block of type $i$.
The mass $b$ fits inside the alternative block.  We have represented the
same total completed mass by items from $\mathcal M_{D,q,v}$ using no more
than the original work $x$.  The definition of
$K_{\mathcal M_{D,q,v}}$ proves the claim.
\end{proof}

\begin{claim}[The shortest-first curve is feasible]
\label{clm:fluid-realization}
The schedule $\textnormal{\textsc{FluidSPT}}(D,q,v)$ is a complete fluid
schedule and has completion curve $K_{\mathcal M_{D,q,v}}$.
\end{claim}

\begin{proof}
The alternative block can be executed wherever its value $v$ places it in
the shortest-first order.  Running any fraction of the testing block means
testing the corresponding fraction of its mass and using $g_*$ on the
revealed types.  By \cref{eq:common-module-data}, every nonempty deferred
block has processing time at least $\tau_D$.  The prescribed tie rule
therefore places the testing block before all deferred work, so each deferred
queue has been exposed before the schedule uses it.  All block capacities
are respected, and their total completion mass is
\[
 qa_*+q(1-a_*)+(1-q)=1.
\]
Thus the schedule is complete.  By
\cref{lem:divisible-spt-area}, its shortest-first execution realizes
$K_{\mathcal M_{D,q,v}}$.
\end{proof}

\begin{proof}[Proof of \cref{lem:test-module-envelope}]
The first inequality in \cref{eq:fluid-pointwise-optimality} is
\cref{clm:fluid-envelope-upper}.  Claims
\ref{clm:fluid-compression} and~\ref{clm:fluid-realization} give the equality
and show that one schedule realizes it simultaneously for every work budget.
Integrating the pointwise inequality and applying
\cref{lem:divisible-spt-area} proves
\cref{eq:common-envelope-area} and optimality.
\end{proof}

\begin{center}
\fbox{\begin{minipage}{0.9\linewidth}
\textbf{Fluid takeaway.}
For fixed $D,q,v$, one maximum-density testing block followed by the
remaining blocks in shortest-work-per-completion order dominates every
other feasible fluid schedule at every work budget.  Instance optimality
will follow by learning this finite block plan from a private random sample.
\end{minipage}}
\end{center}

\paragraph{The three models.}
The theorem above fixes $D$, $q$, and $v$.  In obligatory testing we use
$q=1$.  In revealing optimization with raw time $u$, we use $v=u$ and choose
$q\in[0,1]$.  In blind execution, an untouched job has average processing
time
\[
 \mu_D=\int p\,dD(p),
\]
so we use $v=\mu_D$ and again choose $q\in[0,1]$.  Blind optimization is not
covered by this model, since its preliminary operation does not reveal the
processing time.

It remains to connect this fluid problem with finite inputs.  For an
input $p$, the empirical distribution $D[p]$ is already finitely
supported.  We only round it to a smaller grid so that one concentration
event controls all types simultaneously.  The finite transfer has the
following form:
\begin{center}
\small
fluid lower envelope $\longrightarrow$ plan learned from a subsample
$\longrightarrow$ finite nonanticipating schedule.
\end{center}
The next subsection proves both directions of this approximation.

\subsection{From fluid schedules to finite inputs}
\label{sec:common-finite-transfer}

We now prove that the fluid optimum gives the correct leading cost on finite
inputs.  The lower bound turns an arbitrary adaptive run into a feasible
fluid schedule.  For the upper bound, we learn the empirical distribution
from a sublinear sample and run the corresponding fluid algorithm in finite
batches.

Both directions are easiest to state for the following common model.

Fix $L<\infty$ and a known function
\[
 r:[0,L]\longrightarrow[0,L]
\]
with $|r(p)-r(p')|\le|p-p'|$.  A job first touched in the \emph{direct mode}
runs nonpreemptively for $r(p)$ and completes.  The choice of this mode must
be made before $p$ is known; revealing $p$ after completion may be granted
to the algorithm for free.  Alternatively, the job may be tested for one
unit, after which its processing operation of length $p$ can be scheduled.

Let $Q$ be either $[0,1]$, when both first-touch modes are available, or
$\{1\}$, when every job must first be tested.  We call this the $(r,Q)$
model.  For a distribution $D$, put
\begin{equation}
 v_r(D)\defeq\int r(p)\,dD(p).
 \label{eq:common-direct-mean}
\end{equation}
Since the direct mode is chosen before seeing $p$, a well-mixed unit of jobs
completed this way uses work $v_r(D)$.  We define
\begin{equation}
 \Phi_{r,Q}(D)\defeq
 \min_{q\in Q}\Area\bigl(\mathcal C^{v_r(D)}_{D,q}\bigr).
 \label{eq:common-instance-benchmark}
\end{equation}
When $Q=\{1\}$, the alternative block has zero mass and $r$ is immaterial.

The three specializations used later are
\begin{equation}
 \begin{array}{c|c|c}
  \text{model}&r(p)&Q\\
  \hline
  \mathsf{OT}&\text{irrelevant}&\{1\}\\
  \mathsf{RO}(u)&u&[0,1]\\
  \mathsf{BE}&p&[0,1].
 \end{array}
 \label{eq:common-instance-specializations}
\end{equation}
Thus revealing optimization and blind execution differ only in the amount
of work used by the direct mode.

\begin{theorem}[Finite-to-fluid transfer]
\label{thm:common-fluid-instance-transfer}
Fix $L$, $r$, and $Q$ as above.  There are randomized nonanticipating
algorithms $\mathcal L_n$, depending only on $n,L,r,Q$, and a deterministic
sequence
\begin{equation}
 \eps_L(n)
 =O_L\!\left(n^{-1/6}\sqrt{\ln(n+2)}\right)
 \longrightarrow0
 \label{eq:common-transfer-rate}
\end{equation}
such that every input $p\in[0,L]^n$ satisfies
\begin{equation}
 \EE\ALG_{\mathcal L_n}(p)
 \le n^2\Phi_{r,Q}(D[p])+\eps_L(n)n^2.
 \label{eq:common-transfer-upper}
\end{equation}
Conversely, every randomized nonanticipating algorithm $\mathcal A$ in the
$(r,Q)$ model, even if it is told the multiset of job lengths in $p$, satisfies
\begin{equation}
 \EE\ALG_{\mathcal A}(p)
 \ge n^2\Phi_{r,Q}(D[p])-\eps_L(n)n^2.
 \label{eq:common-transfer-lower}
\end{equation}
Both expectations include the private shuffle and the algorithms' remaining
randomness.
\end{theorem}

We prove the two inequalities separately.  The lower bound uses only the
random order of first touches.  The upper bound additionally uses a short
random subsample to choose the parameters of $\textsc{FluidSPT}$.

\subsubsection{Lower bound: comparing a finite run with the fluid curve}
\label{sec:common-lower-transfer}

Fix an algorithm and condition on its private seed.  List the jobs in the
order in which their labels are first touched.  Their values form a uniform
permutation $X_1,\ldots,X_n$ of the input, even though the algorithm may
choose the labels adaptively.  Moreover, the indicator $A_k$ that the $k$th
job is tested is chosen before $X_k$ is exposed.  We can therefore apply
\cref{lem:predictable-permutation-sampling} to $A_k$ and $1-A_k$.

Keep zero as a separate class, partition $(0,L]$ into
\begin{equation}
 K_n=\lfloor n^{1/6}\rfloor
 \label{eq:common-transfer-grid-size}
\end{equation}
cells, and round positive values upward to their cell endpoints.  Write
$D^+$ for the resulting empirical distribution and $h_n=L/K_n$ for the
mesh.  Take $\delta_n=n^{-1/6}$ and stop when
$\delta_n n+O(1)$ labels remain untouched.

Apply the predictable-permutation estimate simultaneously to the vector of
cell indicators selected by $A_k$ and to the scalar $r(X_k)$ selected by
$1-A_k$.  With probability at least $1-n^{-2}$, every earlier prefix obeys
\begin{align}
 \bigl|\#\{\text{tested jobs in cell }j\}
      -d_j\#\{\text{tests}\}\bigr|&\le\Gamma_n,
 \label{eq:common-transfer-cell-discrepancy}\\
 \bigl|\text{direct work}
      -v_r(D[p])\#\{\text{direct jobs}\}\bigr|&\le\Gamma_n,
 \label{eq:common-transfer-work-discrepancy}
\end{align}
for all cells and all such prefixes, where
\begin{equation}
 \frac{(K_n+1)\Gamma_n}{n}
 =O\!\left(n^{-1/6}\sqrt{\ln(n+2)}\right).
 \label{eq:common-transfer-discrepancy-rate}
\end{equation}
For $Q=\{1\}$ the second estimate is vacuous.

The concentration estimate stops before the last $\delta_n n+O(1)$ first
touches.  We handle this suffix by changing all of its jobs to one fixed
legal mode: we test them when $Q=\{1\}$ and complete them directly when
$Q=[0,1]$.  We delete operations made redundant by this change.  This
modifies only $\delta_n n+O(1)$ jobs.  Since each job uses at most $1+L$
work, remaining-job accounting gives
\begin{equation}
 C_{\rm original}\ge C_{\rm edited}
   -O_L(\delta_n n^2+n).
 \label{eq:common-transfer-suffix-edit}
\end{equation}
Let $q$ be the final tested fraction of the edited run; then $q\in Q$.

\begin{claim}[Comparison with the fluid completion curve]
\label{clm:common-envelope-repair}
On the event that
\cref{eq:common-transfer-cell-discrepancy,eq:common-transfer-work-discrepancy}
hold, define
\begin{equation}
 e_n\defeq\frac{(K_n+1)\Gamma_n}{n}+h_n+\delta_n.
 \label{eq:common-transfer-zeta}
\end{equation}
There is $\zeta_n=O_L(e_n)$ such that, simultaneously throughout the edited
run,
\begin{equation}
 \text{completed fraction after normalized work }x
 \le \mathcal C^{v_r(D^+)}_{D^+,q}(x+\zeta_n)+\zeta_n.
 \label{eq:common-transfer-envelope-repair}
\end{equation}
\end{claim}

\begin{proof}
At an operation boundary, let $t$ be tested mass, $b$ directly completed
mass, and $c_j$ the processed tested mass in positive cell $j$.  The cell
estimate gives
\[
 c_j\le d_jt+\Gamma_n/n,
\]
and gives the analogous bound for tested zeros.  Replace each $c_j$ by
$(c_j-\Gamma_n/n)^+$.  The resulting class capacities are exact and the
removed completed mass is at most $(K_n+1)\Gamma_n/n$.

The work estimate replaces the actual direct work by
$v_r(D[p])b$ with error at most $\Gamma_n/n$.  Upward rounding changes tested
processing work by at most $h_n$ and, because $r$ is one-Lipschitz, changes
the direct mean by at most $h_n$.  The edited suffix has mass
$O(\delta_n)$.  Increase the normalized work coordinate by these errors and
the completed mass by the discarded cell discrepancies.  The variables
$t,b,c_j$ then satisfy every constraint in
\cref{eq:common-fluid-envelope} with $v=v_r(D^+)$.  This proves the
displayed inequality with the error scale in
\cref{eq:common-transfer-zeta}.
\end{proof}

We also need to compare the rounded and original distributions.  Instead of
tracking how the optimal threshold changes, we run the same fluid algorithm
on the two coupled distributions.

\begin{lemma}[Zero-preserving rounding stability]
\label{lem:common-transfer-rounding-stability}
Suppose $D,D'$ admit a coupling that preserves zero and satisfies
$|P-P'|\le h$ almost surely.  Then
\begin{equation}
 |\Phi_{r,Q}(D)-\Phi_{r,Q}(D')|\le2h.
 \label{eq:common-transfer-rounding-stability}
\end{equation}
\end{lemma}

\begin{proof}
Use the same tested fraction, selectors, and block order for coupled mass
elements.  At every prefix, at most one unit of tested mass receives
processing, so its work changes by at most $h$.  Direct work per unit mass
also changes by at most $h$ because $r$ is one-Lipschitz.  Thus either
completion curve lies within a horizontal $2h$ shift of the other.  Their
areas differ by at most $2h$.  Exchange $D,D'$ and minimize over schedules.
\end{proof}

\begin{claim}[Lower bound for announced inputs]
\label{clm:common-announced-lower}
Every announced randomized algorithm in the $(r,Q)$ model satisfies
\begin{equation}
 \frac{\EE\ALG_{\mathcal A}(p)}{n^2}
 \ge\Phi_{r,Q}(D[p])
 -O_L\!\left(n^{-1/6}\sqrt{\ln(n+2)}\right).
 \label{eq:common-announced-lower}
\end{equation}
\end{claim}

\begin{proof}
The envelope reaches one after at most $1+L$ normalized work.  Integrating
\cref{eq:common-transfer-envelope-repair} therefore loses only
$O_L(\zeta_n)$ from its area.  By
\cref{lem:test-module-envelope,eq:common-instance-benchmark}, this area is at
least $\Phi_{r,Q}(D^+)$.  The sampling event fails with probability at most
$n^{-2}$, while every fluid value is at most $(1+L)/2$.  Finally apply the
suffix edit and \cref{lem:common-transfer-rounding-stability}.  Equations
\eqref{eq:common-transfer-discrepancy-rate} and
\eqref{eq:common-transfer-zeta} give the stated rate.  The argument is
uniform after fixing the algorithm seed, so averaging proves the randomized
claim.
\end{proof}

\subsubsection{Upper bound: learning the fluid algorithm}
\label{sec:common-learning-transfer}

For a distribution supported on our grid, a \emph{grid plan} records a
tested fraction, a selector on the processing-time classes, and the order of
the resulting testing, deferred, and direct blocks.  For a grid distribution
$D$ and plan $\pi$, let $F_{r,D}(\pi)$ be its fluid completion cost.
\Cref{lem:test-module-envelope} says that
\begin{equation}
 \min_\pi F_{r,D}(\pi)=\Phi_{r,Q}(D).
 \label{eq:common-template-optimum}
\end{equation}

We use plans defined on the \emph{entire} grid, not only on the support of
$D$.  Concretely, the maximum-density selector is extended canonically to an
empty class: include it below the reciprocal-density threshold, exclude it
above the threshold, and use the fixed convention ``include'' at equality.
Thus a short type absent from the sample still receives the correct action if
it occurs later.  Among multiple minimizing tested fractions and block
orders, choose the smallest tested fraction and then the lexicographically
first order, with testing blocks before deferred blocks at equal density and
grid index as the final tie-breaker.  These conventions make the learned
plan a deterministic function of the public histogram on every grid cell.

We need two simple facts about these plans.

\begin{claim}[Implementing a fixed grid plan]
\label{clm:common-finite-template-kernel}
For every fixed grid plan $\pi$, selecting the prescribed number of test
labels uniformly and executing its blocks gives
\begin{equation}
 \EE C_n(\pi)=n^2F_{r,D[p]}(\pi)+O_L(n),
 \label{eq:common-finite-template-kernel}
\end{equation}
uniformly over the grid, its frequencies, and $\pi$.
\end{claim}

\begin{proof}
Let $N$ be the number of labels to which the plan is applied.  Replace the
tested fraction $q$ by $q_N=\lfloor qN\rfloor/N$ and select exactly
$\lfloor qN\rfloor$ labels by a private uniform permutation.  This decision
is made before any of their values are exposed.  Whenever a tested value in
grid class $b$ is revealed, realize a fractional selector $g_b\in[0,1]$ by
an independent private mark of probability $g_b$; the mark may determine
whether the now-known processing operation is immediate or joins its
class queue.  Within every queue use the private order, and order queues by
the canonical block order above.  Direct-mode labels are likewise fixed
before their hidden values are exposed.  Hence every action is
nonanticipating.  A queue is used only after the tests that populate it, as
guaranteed by the feasible order in \textsc{FluidSPT}.

Expand total completion time into one-job and unordered-pair contributions.
For two distinct jobs, the probabilities that zero, one, or both are tested
differ from the corresponding products of the tested fraction by $O(1/n)$.
The selector marks give the prescribed fractional class products exactly in
expectation.  Replacing $q$ by $q_N$, stable tie-breaking, and diagonal terms
change only $O_L(N)$ in total.  Every operation and pair contribution is
bounded in terms of $L$.
Summing the diagonal terms and these finite-population corrections gives
$O_L(n)$; the remaining terms are exactly the fluid block area.
\end{proof}

\begin{claim}[Stability of a fixed grid plan]
\label{clm:common-template-stability}
There is $C_L<\infty$ such that any two distributions $D,E$ on the same
grid and every plan $\pi$ defined for both distributions satisfy
\begin{equation}
 |F_{r,D}(\pi)-F_{r,E}(\pi)|
 \le C_L\|D-E\|_1.
 \label{eq:common-template-stability}
\end{equation}
The same inequality holds after minimizing over grid plans.
\end{claim}

\begin{proof}
For a fixed plan, expand the area into the contributions of individual
blocks and pairs of blocks.  Every contribution is bounded by a constant
depending only on $L$.  Replacing first one input of each pair and then the
other gives the displayed total-variation bound.  Explicitly, with
$\|\cdot\|_{\rm TV}=\frac12\|\cdot\|_1$,
\begin{equation}
 \|D\otimes D-E\otimes E\|_{\rm TV}
 \le \|(D-E)\otimes D\|_{\rm TV}
     +\|E\otimes(D-E)\|_{\rm TV}
 \le2\|D-E\|_{\rm TV}.
 \label{eq:common-product-tv}
\end{equation}
The bound is uniform over the number of grid classes and over the plan.
Applying it first to a
minimizer for $D$ and then to one for $E$ proves the assertion for minima.
\end{proof}

The learner uses the grid from
\cref{eq:common-transfer-grid-size}, privately permutes the labels, and tests
the first
\begin{equation}
 k_n=\lceil n^{1/2}\rceil
 \label{eq:common-transfer-sample-size}
\end{equation}
jobs.  It keeps positive sample jobs pending and forms their rounded
histogram $\widehat D_n$.  It chooses the canonical full-grid plan minimizing
$F_{r,\widehat D_n}$ and applies it to the unused labels, then finishes the
positive sample jobs.

For $n<4$, let the learner instead use any fixed complete legal schedule.
In the rest of the proof $n\ge4$, so $k_n<n$ and the remaining histogram
below is well defined.  The cost of the finitely many smaller inputs is
absorbed by the constant in \cref{eq:common-learning-upper}.

Let $D_n$ and $D_n^{\rm rem}$ be the full and remaining rounded histograms.
The sampling-without-replacement estimate
\cref{eq:without-replacement-histogram} and deletion of the sampled jobs give
\begin{equation}
 \EE\|\widehat D_n-D_n\|_1
 \le\sqrt{\frac{K_n+1}{k_n}},
 \qquad
 \|D_n^{\rm rem}-D_n\|_1
 \le\frac{2k_n}{n-k_n}.
 \label{eq:common-transfer-histogram}
\end{equation}

\begin{claim}[Learning from a random subsample]
\label{clm:common-learning-upper}
The learner above satisfies
\begin{equation}
 \frac{\EE\ALG_{\mathcal L_n}(p)}{n^2}
 \le\Phi_{r,Q}(D[p])+O_L(n^{-1/6})
 \label{eq:common-learning-upper}
\end{equation}
uniformly over $p\in[0,L]^n$.
\end{claim}

\begin{proof}
Condition on the sample.  The unused labels remain uniformly permuted, so
\cref{clm:common-finite-template-kernel} applies to the selected plan.
In particular, the decision to use the direct mode and the random tested
quota are fixed from the sample before an unused hidden value is seen;
selector marks are drawn only after a test reveals their grid class.  The
canonical threshold extension defines marks and queue positions even for
grid cells absent from the sample.  The sampled prefix and its final positive
tail use $O_L(k_n)$ work and therefore add $O_L(nk_n)$ to total completion
time.

The fixed-plan estimate is applied to the upward-rounded values.  Coupling
the physical and rounded executions with the same decisions changes tested
processing work by at most $h_n$ per unit mass and direct work by at most
$h_n$ per unit mass, because $r$ is one-Lipschitz.  Remaining-job accounting
therefore changes the normalized cost by at most $O(h_n)$.

Apply \cref{clm:common-template-stability} first to the learned plan, then to
an optimal plan for $D_n^{\rm rem}$, and finally to an optimal plan for
$D_n$.  This bounds the conditional fluid cost by
\[
 \Phi_{r,Q}(D_n)
 +O_L\bigl(
   \|\widehat D_n-D_n\|_1+
   \|D_n^{\rm rem}-D_n\|_1\bigr).
\]
Average and use \cref{eq:common-transfer-histogram}.  The three remaining
errors are
\[
 h_n=O_L(n^{-1/6}),\qquad
 \sqrt{(K_n+1)/k_n}=O(n^{-1/6}),\qquad
 k_n/n=O(n^{-1/2}).
\]
The rounding term is controlled by
\cref{lem:common-transfer-rounding-stability}, proving the claim.
\end{proof}

\begin{proof}[Proof of \cref{thm:common-fluid-instance-transfer}]
Claim~\ref{clm:common-learning-upper} gives
\cref{eq:common-transfer-upper}.  Claim~\ref{clm:common-announced-lower}
gives \cref{eq:common-transfer-lower}, since announcing the multiset only
strengthens the competing algorithm.  Taking the larger of their vanishing
errors gives \cref{eq:common-transfer-rate}.
\end{proof}

\subsection{Revealing-optimization instance optimality}
\label{sec:common-upper-instance}

We first consider revealing optimization.  A job that is not tested is run
for the public time $u$, independently of its hidden processing time.  Thus,
in the common model, $r(p)=u$.  It remains only to expand the fluid cost as a
quadratic in the tested fraction.

\begin{theorem}[Revealing-optimization instance optimality]
\label{thm:common-upper-instance}
For every fixed $u>0$ there exist randomized nonanticipating algorithms
$\mathcal A^*_{n,u}$, depending only on $n,u$, and a deterministic sequence
$\eps_u(n)\to0$ such that every input $p\in[0,u]^n$ satisfies
\begin{equation}
 \EE\ALG_{\mathcal A^*_{n,u}}(p)
 \le n^2\PhiRO(D[p])+\eps_u(n)n^2.
 \label{eq:cui-upper}
\end{equation}
Conversely, every randomized nonanticipating algorithm $\mathcal A$, even an
announced one, satisfies
\begin{equation}
 \EE\ALG_{\mathcal A}(p)
 \ge n^2\PhiRO(D[p])-\eps_u(n)n^2.
 \label{eq:cui-lower}
\end{equation}
One may take
\begin{equation}
 \eps_u(n)=O_u\!\left(n^{-1/6}\sqrt{\ln(n+2)}\right).
 \label{eq:cui-rate}
\end{equation}
\end{theorem}

\subsubsection{The revealing benchmark}
\label{sec:cui-fluid}

Let $D$ be a probability distribution on $[0,u]$.  Apply
\cref{lem:maximum-density-module}, choose its threshold selector $g_*$, and
write
\begin{equation}
 a=a_D(g_*),\qquad
 \ell=\int_{(0,u]}p g_*(p)\,dD(p),
 \qquad \tau_D=\frac{1+\ell}{a}.
 \label{eq:cui-tau}
\end{equation}
Let $\nu$ be the residual measure from
\cref{eq:common-module-data}.  It has mass $1-a$, first moment
\begin{equation}
 \ell_\nu=\int p\,d\nu(p)=\int p\,dD(p)-\ell,
 \qquad
 \SPT(\nu)=\frac12\iint\min\{p,p'\}\,d\nu(p)d\nu(p'),
 \label{eq:cui-residual-moments}
\end{equation}
and support in $[\tau_D,u]$.

For a fixed tested fraction $q$, revealing optimization is the common fluid
model with $v=u$.  Hence \cref{lem:test-module-envelope} gives the blocks
\[
 qa\delta_{\tau_D}+q\nu+(1-q)\delta_u.
\]
When $\tau_D<u$, shortest-first order runs the testing block, the residual
types in SPT order, and finally the raw block.  Expanding
\cref{eq:common-envelope-area} gives
\begin{equation}
 F_{\mathsf{RO},u,D}(q)
 =(1+\ell)\left(q-\frac{aq^2}{2}\right)
 +q\ell_\nu(1-q)+q^2\SPT(\nu)
 +\frac{u(1-q)^2}{2}.
 \label{eq:cui-F}
\end{equation}
The four terms are, respectively, the testing block, the residual work
delaying raw jobs, the internal residual SPT cost, and the raw block.

If $\tau_D\ge u$, every testing or residual block has work per completion at
least $u$, so pure raw execution pointwise dominates every positive tested
fraction.  We therefore define
\begin{equation}
 \PhiRO(D)=
 \begin{cases}
  \displaystyle\min_{0\le q\le1}F_{\mathsf{RO},u,D}(q),&\tau_D<u,\\[3pt]
  u/2,&\tau_D\ge u.
 \end{cases}
 \label{eq:cui-Phi}
\end{equation}

The minimization is explicit.  Expanding gives
\begin{equation}
 F_{\mathsf{RO},u,D}(q)=\frac u2+Aq+Bq^2,
 \quad
 A=1+\int p\,dD(p)-u,
 \quad
 B=-\frac{a(1+\ell)}2-\ell_\nu+\SPT(\nu)+\frac u2.
 \label{eq:cui-quadratic}
\end{equation}
For each work budget, convexly combining feasible states for two tested
fractions gives a feasible state for their convex combination.  Thus the
completion envelope is concave in $q$, its unfinished-mass area is convex,
and $B\ge0$.  Moreover,
\[
 \tau_D<u
 \quad\Longleftrightarrow\quad
 \int(u-p)\,dD(p)>1
 \quad\Longleftrightarrow\quad A<0.
\]
Consequently a minimizing tested fraction is
\begin{equation}
 q^*=1\quad(B=0),
 \qquad
 q^*=\min\{1,-A/(2B)\}\quad(B>0).
 \label{eq:cui-qstar}
\end{equation}

\begin{proof}[Proof of \cref{thm:common-upper-instance}]
Take $r(p)=u$ and $Q=[0,1]$ in
\cref{thm:common-fluid-instance-transfer}.  Then
$v_r(D)=u$, and \cref{eq:cui-F,eq:cui-Phi} show that
$\Phi_{r,Q}(D)=\PhiRO(D)$.  The upper and lower conclusions are
\cref{eq:cui-upper,eq:cui-lower}, and
\cref{eq:common-transfer-rate} gives \cref{eq:cui-rate}.
\end{proof}

\subsection{Obligatory-testing instance optimality}
\label{sec:obligatory-instance}

In obligatory testing every job is tested, so $q=1$.  The optimal fluid
algorithm consists only of the maximum-density testing block and its
deferred SPT tail.  We will additionally show how to learn this algorithm
without knowing any bound on the processing times: a slowly growing cutoff
gives one algorithm that is optimal on every fixed bounded class.

\begin{theorem}[Obligatory-testing instance optimality]
\label{thm:obligatory-instance}
There is a single randomized nonanticipating algorithm $\mathcal A^*$,
depending only on $n$, with the following property.  For every fixed
$L<\infty$ there is a deterministic sequence $\eps_L(n)\to0$ such that every
input $p\in[0,L]^n$ satisfies
\begin{equation}
 \EE\ALG_{\mathcal A^*}(p)
 \le n^2\PhiOT(D[p])+\eps_L(n)n^2.
 \label{eq:ot-instance-upper}
\end{equation}
Conversely, every randomized nonanticipating algorithm $\mathcal A$, even if
it is told the multiset of job lengths in $p$, satisfies
\begin{equation}
 \EE\ALG_{\mathcal A}(p)
 \ge n^2\PhiOT(D[p])-\eps_L(n)n^2.
 \label{eq:ot-instance-lower}
\end{equation}
One may take $\eps_L(n)=O_L(n^{-1/5})$.  Thus the same algorithm is
asymptotically instance optimal, uniformly on every bounded input class.
\end{theorem}

\subsubsection{The obligatory benchmark}

Let $D$ be a finitely supported distribution on $[0,\infty)$.  Apply
\cref{lem:maximum-density-module}, choose its threshold selector $g_*$, and
put
\begin{equation}
 a_*=a_D(g_*),\qquad w_*=w_D(g_*)=a_*\tau_D.
 \label{eq:ot-module-threshold}
\end{equation}
The residual measure $\nu$ from \cref{eq:common-module-data} has mass
$1-a_*$ and support in $[\tau_D,\infty)$.  Since $q=1$, the block collection
in \cref{eq:common-envelope-items} is
\[
                       a_*\delta_{\tau_D}+\nu.
\]
By \cref{lem:test-module-envelope}, the optimal fluid cost is therefore
\begin{equation}
 \boxed{
 \PhiOT(D)
 =w_*\left(1-\frac{a_*}{2}\right)+\SPT(\nu).}
 \label{eq:ot-Phi}
\end{equation}
Indeed, the first term is the area contributed by the testing block, both
over itself and over the residual mass, and the second is the internal SPT
area of the residual block.  This is exactly
$\Area(\mathcal C^0_{D,1})$, so \cref{lem:test-module-envelope} proves that
no other fluid schedule has smaller cost.

For an empirical distribution, the selector has a direct finite
description.  Sort the job occurrences by processing time and, among
prefixes $E$ containing all zero jobs, maximize
\begin{equation}
 \frac{|E|/n}{1+n^{-1}\sum_{i\in E}p_i}.
 \label{eq:ot-finite-density}
\end{equation}
The addition and removal inequalities include every occurrence below
$\tau_D$, exclude every occurrence above $\tau_D$, and allow either choice
at equality.  Test jobs in a uniformly random order, immediately process
the positive jobs in $E$, and finally process the remaining jobs in SPT
order.  Direct remaining-job accounting gives the exact identity
\begin{equation}
 \EE\ALG_{\mathrm{stat}}(p)
 =n^2\PhiOT(D[p])
  +\frac12\left(|E|+\sum_i p_i\right).
 \label{eq:ot-announced-upper}
\end{equation}
Thus the fluid value is already attainable up to $O_L(n)$ when the multiset
is announced.

The general lower-bound proof can be made slightly sharper here.  Every job
is first touched by a test, so the revealed values themselves form one
uniform random permutation.  There is no choice between two first-touch
modes and no unobserved suffix that has to be changed.

\begin{claim}[Lower bound for obligatory testing]
\label{clm:ot-announced-lower}
For every $p\in[0,L]^n$ and every announced randomized obligatory-testing
algorithm $\mathcal A$,
\begin{equation}
 \frac{\EE\ALG_{\mathcal A}(p)}{n^2}
 \ge\PhiOT(D[p])-O_L(n^{-1/5}).
 \label{eq:ot-announced-lower}
\end{equation}
\end{claim}

\begin{proof}
We use the argument from \cref{clm:common-announced-lower} with a finer
grid.  Partition $(0,L]$ into $K_n=\lfloor n^{1/5}\rfloor$ cells.  By
\cref{eq:uniform-prefix-hoeffding}, with probability at least $1-n^{-2}$,
the number of occurrences of every cell in every prefix differs from its
expectation by at most $O(\sqrt{n\ln(n+2)})$.

Consider the schedule after it has tested mass $t$.  For each cell, reduce
the processed mass by the above discrepancy divided by $n$.  The remaining
class masses are at most their fluid capacities, and the total mass removed
in this way is
\[
 O\!\left(K_n\sqrt{\frac{\ln(n+2)}n}\right)
 =o(n^{-1/5}),
\]
while the zero-preserving rounding mesh is $O_L(n^{-1/5})$.  There is no
direct-work error and every job occurs in the test order, so there is no
suffix to modify.  The resulting completion curve is below the fluid
envelope, up to the two errors above.  Integrating it and applying
\cref{lem:common-transfer-rounding-stability} proves the displayed bound.
The argument is uniform after conditioning on the algorithm's private seed,
so it also holds for randomized algorithms.
\end{proof}

\subsubsection{Learning without knowing the input bound}

We now specify the single algorithm promised in
\cref{thm:obligatory-instance}.  It uses a cutoff that tends to infinity and
a grid whose mesh tends to zero.

\begin{algorithm}[H]
\caption{}
\label{alg:ot-growing-cutoff-instance}
\begin{algorithmic}
\State If $n<16$, test every job, process all positive jobs in SPT order, and stop.
\State Set $k=\lfloor n^{3/4}\rfloor$, $d=\lfloor n^{1/4}\rfloor$,
       $B_n=32+n^{1/20}$, and $\eta=B_n/d$.
\State Privately permute the labels and test labels $1,\ldots,k$, leaving
       their positive processing operations pending.
\State Keep zero as its own category, round values in $(0,B_n]$ upward to
       the $B_n$-grid of mesh $\eta$, and put values above $B_n$ in an
       overflow category.
\State If the sample has no mass in a finite category, choose no finite
       category and set the maximum sample density to zero.
\State Otherwise, among finite category prefixes, choose one with maximum
       sample completion density as in \cref{eq:ot-finite-density}, and close
       it by including every finite category up to its reciprocal density.
\State If the maximum sample density is less than $1/B_n$, test all
       remaining jobs and process every positive job in SPT order.
\State Otherwise, process the pending sample jobs in the chosen categories;
       then test labels $k+1,\ldots,n$ in order and immediately process each
       positive outcome in a chosen category.
\State Finally process all remaining positive jobs in SPT order.
\end{algorithmic}
\end{algorithm}

The overflow category is never selected.  The closure step also includes
finite categories absent from the sample, preventing an unseen short value
from being incorrectly placed in the final queue.

\begin{claim}[Growing-cutoff learning upper bound]
\label{clm:ot-learning-upper}
For every fixed $L<\infty$ and every $p\in[0,L]^n$,
\begin{equation}
 \EE\ALG_{\mathcal A^*}(p)
 \le n^2\PhiOT(D[p])+O_L(n^{9/5}).
 \label{eq:ot-instance-upper-rate}
\end{equation}
\end{claim}

\begin{proof}
Fix $L$.  The small-$n$ branch has bounded cost, which can be absorbed by the
constant in \cref{eq:ot-instance-upper-rate}.  For all sufficiently large
$n$, $B_n>L$, so the overflow category
is empty.  The prefix containing the whole sample has density at least
$1/(1+L+\eta)>1/B_n$, and hence fallback does not occur.

Let $D^+$ be the rounded full histogram, $D^{\rm rem}$ the rounded histogram
of the $n-k$ unused jobs, and $\widehat D$ the sample histogram.  The
sampling-without-replacement estimate
\cref{eq:without-replacement-histogram} and deletion of the sampled jobs give
\begin{equation}
 \EE\|\widehat D-D^+\|_1\le\sqrt{\frac{d+2}{k}},
 \qquad
 \|D^{\rm rem}-D^+\|_1\le\frac{2k}{n-k}.
 \label{eq:ot-learning-histogram}
\end{equation}

The selected category prefix is the optimal grid plan for
$\widehat D$ by \cref{lem:maximum-density-module,lem:test-module-envelope}.
Conditional on the sample, the remaining labels are uniformly permuted.
The fixed-plan implementation
\cref{clm:common-finite-template-kernel} therefore evaluates their main
schedule, while testing the sample first and finishing its pending tail adds
only $O_L(nk+k^2)$ to total completion time.  The stability bound
\cref{clm:common-template-stability}, applied to the learned plan and
then to an optimizer, yields
\[
 \frac{\EE[\ALG_{\mathcal A^*}(p)\mid\widehat D]}{n^2}
 \le \PhiOT(D^+)
 +O_L\!\left(
   \|\widehat D-D^+\|_1+
   \|D^{\rm rem}-D^+\|_1+
   \frac{k}{n}\right).
\]
Using actual rather than upward-rounded processing times can only decrease
the cost.  Moreover, zero-preserving rounding and
\cref{lem:common-transfer-rounding-stability} give
$\PhiOT(D^+)\le\PhiOT(D[p])+2\eta$.
Average, use \cref{eq:ot-learning-histogram}, and observe that
\[
 \sqrt{d/k}=O(n^{-1/4}),\qquad
 k/n=O(n^{-1/4}),\qquad
 \eta=O(n^{-1/5}).
\]
Multiplication by $n^2$ proves the claim.  Enlarging the constant covers the
finitely many smaller values of $n$.
\end{proof}

\begin{proof}[Proof of \cref{thm:obligatory-instance}]
Claim~\ref{clm:ot-learning-upper} proves \cref{eq:ot-instance-upper}.
Claim~\ref{clm:ot-announced-lower} proves
\cref{eq:ot-instance-lower}, since announcing the multiset can only help the
competing algorithm.  Both use an $O_L(n^{-1/5})$ normalized error, and the
algorithm in \cref{alg:ot-growing-cutoff-instance} depends only on $n$.
\end{proof}

\subsubsection{Why randomization and boundedness are necessary}

Randomization and bounded inputs are both essential to the preceding
statement.

\begin{theorem}[When instance optimality is impossible]
\label{thm:ot-instance-impossibility}
\begin{enumerate}[label=\textnormal{(\roman*)}]
\item There is no deterministic instance-optimal algorithm, even when all job
lengths lie in $[0,2]$.  More precisely, for every deterministic algorithm
$\mathcal A$, even an announced one, and every even $n$, there is a fixed
labeling with $n/2$ zero jobs and $n/2$ length-two jobs for which
\begin{equation}
 \ALG_{\mathcal A}\ge\frac98n^2.
 \label{eq:ot-det-not-instance-optimal}
\end{equation}
On the same input, an announced randomized algorithm $\mathcal A'$ has
\[
 \EE\ALG_{\mathcal A'}=n^2+\frac34n.
\]

\item If processing times are not bounded by a fixed constant, there is no
randomized instance-optimal algorithm with one uniform additive $o(n^2)$
term.  More precisely, for every randomized algorithm $\mathcal A$
independent of the input and every sufficiently large $n$, there is an input
with values in $[0,n^2]$ and an announced randomized algorithm $\mathcal A'$
such that
\begin{equation}
 \EE\ALG_{\mathcal A}\ge\EE\ALG_{\mathcal A'}+\frac18n^2.
 \label{eq:ot-unbounded-not-instance-optimal}
\end{equation}
\end{enumerate}
\end{theorem}

\begin{proof}
For part~(i), write $n=2m$.  Against the fixed deterministic algorithm,
reveal length two on the first $m$ labels it tests and zero on the remaining
$m$.  Record each answer on the label where it was given.  Rerunning the
algorithm on this fixed labeling reproduces exactly the same transcript.
If it does not finish, the claim is immediate.  Otherwise, let $h$ be the
number of length-two jobs completed before the first zero is tested.  By
then all $m$ positive jobs have been tested, so at least $m+2h$ work has
elapsed.  Even if the remaining work is put in SPT order,
\begin{align*}
 \ALG_{\mathcal A}(p)
 &\ge \frac{3h(h+1)}2
 +(2m-h)(m+2h)+\frac{m(m+1)}2\\
 &\hspace{22mm}{}+m(m-h)+(m-h)(m-h+1)\\
 &=\frac92m^2+\frac12h^2+\frac32m+\frac12h
 \ge\frac98n^2.
\end{align*}
The comparison algorithm tests in a uniformly random order, completes zeros
at their tests, defers all length-two jobs, and finally processes them.
Remaining-job accounting gives cost $n^2+3n/4$.

For part~(ii), put $H=n^2$ and first run $\mathcal A$ on the input in which
every job has length $H$.  For a fixed private seed, let $X_k$ be the number
of completed jobs immediately before the $k$th test.  Every processing
operation contributes the same $Hn(n+1)/2$ in total, regardless of its
position.  Relative to the algorithm $\mathcal A'$ that tests and processes
each job immediately, the expected excess is
\begin{equation}
 \EE\ALG_{\mathcal A}-\EE\ALG_{\mathcal A'}
 =\EE\sum_{k=1}^n(k-1-X_k).
 \label{eq:ot-all-H-excess}
\end{equation}
If this is at least $n^2/8$, we are done.

Otherwise replace one uniformly random label by a zero.  Conditional on the
seed, its test rank $K$ is uniform on $\{1,\ldots,n\}$.  Until that test the
run agrees with the all-$H$ run, and \cref{eq:ot-all-H-excess} gives
\begin{equation}
 \EE X_K>\frac{3n}{8}-\frac12.
 \label{eq:ot-zero-rank-completions}
\end{equation}
Forcing the zero behind $X_K$ long jobs adds $HX_K$ to the clairvoyant SPT
lower bound.  An announced algorithm can instead test in a private random
order, defer positive jobs until it finds the zero, and then finish them;
moving its unit tests away from the corresponding processing operations
costs at most $n^2$.  Hence
\[
 \EE\ALG_{\mathcal A}-\EE\ALG_{\mathcal A'}
 \ge H\left(\frac{3n}{8}-\frac12\right)-n^2,
\]
which exceeds $n^2/8$ for all sufficiently large $n$.  Averaging over the
possible zero labels is exactly the private-shuffle expectation.
\end{proof}

\subsection{Blind-execution instance optimality}
\label{sec:blind}

We now consider blind execution.  A job that is not tested runs for its
actual processing time, so $r(p)=p$ and the average direct work is the mean
processing time.  The optimal fluid algorithm has four blocks: a
maximum-density testing block, tested medium jobs, jobs run blindly, and
tested long jobs.  We now compute its cost explicitly.

\begin{theorem}[Blind-execution instance optimality]
\label{thm:blind-instance}
Fix $L<\infty$.  There exist randomized nonanticipating algorithms
$\mathcal A^*_{n,L}$, depending only on $n,L$, and a deterministic sequence
$\eps_L(n)\to0$ such that every input $p\in[0,L]^n$ satisfies
\begin{equation}
 \EE\ALG_{\mathcal A^*_{n,L}}(p)
 \le n^2\PhiBE(D[p])+\eps_L(n)n^2.
 \label{eq:intro-blind-upper}
\end{equation}
Conversely, every randomized nonanticipating algorithm $\mathcal A$, even if
it is told the multiset of job lengths in $p$, satisfies
\begin{equation}
 \EE\ALG_{\mathcal A}(p)
 \ge n^2\PhiBE(D[p])-\eps_L(n)n^2.
 \label{eq:intro-blind-lower}
\end{equation}
The benchmark $\PhiBE$ is defined explicitly below, and one may take
\[
 \eps_L(n)=O_L\!\left(n^{-1/6}\sqrt{\ln(n+2)}\right).
\]
\end{theorem}

\subsubsection{The four-block benchmark}
\label{sec:blind-benchmark}

Let $D$ be a probability distribution on $[0,L]$ and put
\begin{equation}
 \mu=\int p\,dD(p).
 \label{eq:blind-mean}
\end{equation}
Thus blindly executing one unit of well-mixed job mass uses work $\mu$.
Apply \cref{lem:maximum-density-module}, choose its threshold selector $g_*$,
and let $\nu$ be the residual measure from
\cref{eq:common-module-data}.  Write
\[
 a=a_D(g_*),\qquad
 \ell=\int_{(0,L]}p g_*(p)\,dD(p).
\]
The maximum-density identity is
\begin{equation}
 a\tau_D-\ell=1,
 \qquad \tau_D=\frac{1+\ell}{a}.
 \label{eq:blind-test-module}
\end{equation}
Hence a full unit of tests together with the selected short outcomes is a
divisible block of completion mass $a$, work $1+\ell$, and work per
completion $\tau_D$.

For a fixed tested fraction $q$, blind execution is the common fluid model
with alternative work $v=\mu$.  By \cref{lem:test-module-envelope}, its
optimal block collection is
\begin{equation}
 qa\delta_{\tau_D}+q\nu+(1-q)\delta_\mu.
 \label{eq:blind-fluid-blocks}
\end{equation}
If $\tau_D\ge\mu$, the direct block is no slower per completion than any
testing block, so the optimum is pure blind execution and has cost $\mu/2$.
Suppose now that $\tau_D<\mu$.  Sorting the blocks in
\cref{eq:blind-fluid-blocks} by work per completion gives, in order,
\begin{equation}
 \begin{aligned}
 \underbrace{\text{tests and selected }p<\tau_D}_{\text{testing block}}
 &\ \longrightarrow\
 \underbrace{\text{tested }\tau_D<p<\mu}_{\text{SPT order}}\\[-2pt]
 &\ \longrightarrow\
 \underbrace{\text{blindly executed jobs}}_{\text{work }\mu}
 \ \longrightarrow\
 \underbrace{\text{tested }p>\mu}_{\text{SPT order}}.
 \end{aligned}
 \label{eq:blind-four-block-order}
\end{equation}

We next write the area of this completion curve explicitly.  To avoid
inessential boundary notation, first suppose that $D$ has no mass at
$\tau_D$ or $\mu$, and set
\begin{align*}
 \mathcal L&=\{p<\tau_D\},&
 \mathcal M&=\{\tau_D<p<\mu\},&
 \mathcal H&=\{p>\mu\},\\
 a&=D(\mathcal L),&
 \ell&=\int_{\mathcal L}p\,dD(p),&
 \ell_M&=\int_{\mathcal M}p\,dD(p),\\
 d&=D(\mathcal H),&
 D_M&=D|_{\mathcal M},&D_H&=D|_{\mathcal H}.
\end{align*}
For $\bar q=1-q$, remaining-mass accounting gives
\begin{align}
 F_{\mathsf{BE},D}(q)
 &= (1+\ell)\left(q-\frac{aq^2}{2}\right)
    +q\ell_M(\bar q+qd)+q^2\SPT(D_M)\notag\\
 &\quad+\mu\left(\bar q qd+\frac{\bar q^2}{2}\right)
    +q^2\SPT(D_H).
 \label{eq:blind-F}
\end{align}
The five terms describe, respectively, the testing block, the external and
internal costs of the tested medium jobs, the blind block, and the internal
cost of the tested long jobs.  If $D$ has mass at either boundary, split that
mass between the adjacent blocks; the value is unchanged.

Expansion yields
\begin{equation}
 F_{\mathsf{BE},D}(q)=\frac\mu2+Aq+Bq^2,
 \label{eq:blind-quadratic}
\end{equation}
where
\begin{align}
 A&=1+\ell+\ell_M+\mu(d-1)
   =1-\int(\mu-p)^+\,dD(p),
 \label{eq:blind-A}\\
 B&=-\frac{a(1+\ell)}2+\ell_M(d-1)+\SPT(D_M)
   +\mu\left(\frac12-d\right)+\SPT(D_H).
 \label{eq:blind-B}
\end{align}
The feasible set in \cref{eq:common-fluid-envelope} is affine in $q$.
Consequently its completion envelope is concave in $q$, its unfinished-mass
area is convex, and $B\ge0$.  Moreover, when $\tau_D<\mu$,
\cref{eq:common-threshold-equation} gives $A<0$.  A minimizing tested
fraction is therefore
\begin{equation}
 q^*=1\quad(B=0),
 \qquad
 q^*=\min\{1,-A/(2B)\}\quad(B>0).
 \label{eq:blind-qstar}
\end{equation}
We define
\begin{equation}
 \PhiBE(D)=
 \begin{cases}
  \displaystyle\min_{0\le q\le1}F_{\mathsf{BE},D}(q),&\tau_D<\mu,\\[3pt]
  \mu/2,&\tau_D\ge\mu.
 \end{cases}
 \label{eq:blind-Phi}
\end{equation}

\begin{proof}[Proof of \cref{thm:blind-instance}]
Take $r(p)=p$ and $Q=[0,1]$ in
\cref{thm:common-fluid-instance-transfer}.  Then $r$ is one-Lipschitz,
$v_r(D)=\mu$, and \cref{eq:blind-fluid-blocks,eq:blind-Phi} show that
$\Phi_{r,Q}(D)=\PhiBE(D)$.  The two conclusions and
the stated rate now follow from
\cref{eq:common-transfer-upper,eq:common-transfer-lower,eq:common-transfer-rate}.
\end{proof}

\subsubsection{An example requiring all four blocks}

Let
\[
 D(\{0\})=\frac15,\qquad
 D(\{9\})=\frac15,\qquad
 D(\{16\})=\frac35.
\]
Then $\mu=57/5$ and $\tau_D=5$.  Thus tested zeros complete during their
tests, tested jobs of length nine precede the jobs run blindly, and tested
jobs of length sixteen follow them.  Formula \cref{eq:blind-F} becomes
\[
 F_{\mathsf{BE},D}(q)=\frac{57}{10}-\frac{44}{25}q+\frac{11}{10}q^2.
\]
Its minimizer is $q^*=4/5$ and its value is $1249/250$.  If all tested nonzero
jobs are instead deferred until after the blind block, the best value is
$126/25$.  The gap is a positive constant at the fluid scale and therefore a
positive constant times $n^2$ for large finite inputs.  In particular, the
simpler order ``testing block, blind block, all tested residuals'' is not
instance optimal.

\subsubsection{Unbounded processing times}

The bound $L$ in \cref{thm:common-fluid-instance-transfer} is necessary.
Without it, blind execution has no finite worst-case ratio.

\begin{theorem}[No finite ratio for unbounded blind execution]
\label{thm:blind-unbounded-impossible}
There are no randomized algorithms $\mathcal A_n$, constant $R<\infty$, and
sequence $\eta(n)\to0$ such that every input
$p\in[0,\infty)^n$ satisfies
\begin{equation}
 \EE\ALG_{\mathcal A_n}(p)
 \le R\cdot\OPT(p)+\eta(n)n^2.
 \label{eq:blind-unbounded-impossible}
\end{equation}
\end{theorem}

\begin{proof}
Suppose that such a guarantee exists, and run $\mathcal A_n$ on the all-zero
input.  Its offline cost is zero, so if $Z_n$ denotes the expected online
cost, then $Z_n\le\eta(n)n^2=o(n^2)$.  We may assume that the algorithm
finishes every job.

Fix its private seed and order the labels by their first touch.  Every first
touch completes its zero job.  If the job at rank $h$ is tested, that test is
performed while $n-h+1$ jobs are unfinished and contributes $n-h+1$ to the
objective.  Consequently
\begin{equation}
 Z_n\ge\EE\sum_{h:\,\text{rank }h\text{ is tested}}(n-h+1).
 \label{eq:blind-zero-test-area}
\end{equation}

Now choose one label uniformly and change only its length from zero to
$H=n^2$.  Until this exceptional label is first touched, the algorithm sees
the zero-input history.  If its rank $h$ is touched blindly, its execution
occupies the machine for $H$ time units while $n-h+1$ jobs are unfinished.
Averaging over the exceptional label and the private seed gives
\begin{align}
 \EE\ALG_{\mathcal A_n}(p)
 &\ge \frac Hn\,\EE
   \sum_{h:\,\text{rank }h\text{ is blind}}(n-h+1)\notag\\
 &\ge\frac Hn\left(\frac{n(n+1)}2-Z_n\right)
  =H\left(\frac n2-o(n)\right).
 \label{eq:blind-one-long-average}
\end{align}
Some fixed choice of the exceptional label has at least this expected cost.
Its offline optimum runs all zero jobs first and the long job last, and hence
$\OPT(p)=H$.  With $H=n^2$, \cref{eq:blind-one-long-average} contradicts
\cref{eq:blind-unbounded-impossible} for all sufficiently large $n$.
\end{proof}

\section{Deterministic revealing optimization and its obligatory endpoint}
\label{sec:optional}

This section proves two exact deterministic results: the revealing-optimization
curve in \cref{thm:optional-curve} and the obligatory-testing ratio in
\cref{thm:det-obligatory}, obtained as its uniform unbounded endpoint.  These
are the formal deterministic versions of the statements in
\cref{thm:intro-common-upper-curves,thm:intro-obligatory}.  We begin with
revealing optimization at a finite common raw-execution time $u$.  All jobs
have this same known upper limit, and a
test reveals $p_j\in[0,u]$.  By
\cref{eq:model-effective-length,lem:offline}, the offline effective length
is $\min\{u,1+p_j\}$, and the offline optimum schedules these
effective lengths in nondecreasing order.  We write $\RdetRO(u)$ for the optimal
deterministic size-asymptotic ratio at this fixed common upper limit.

The exact ratio is given by six formulas on six intervals separated by five
transition points.  Put
\begin{equation}
 \varphi=\frac{1+\sqrt5}{2},
 \qquad
 \uDet{1}=1,
 \qquad
 \uDet{2}=\frac{1+\sqrt{3+2\sqrt5}}2
 =1.86\ldots .
 \label{eq:opt:u1-u2}
\end{equation}
Let $s_0>1$ be the root of
\begin{equation}
 s_0^3=(s_0+1)^2,
 \qquad
 \uDet{3}=1+s_0=3.14\ldots .
 \label{eq:opt:u3}
\end{equation}
The fourth transition point is
\begin{equation}
 \uDet{4}=\varphi+2=3.61\ldots .
 \label{eq:opt:u4}
\end{equation}
For $s=u-1$, define
\[
 T_s(z)=sz^2+(s^3+s^2+1)z-(s^2+s+1)(s+2),
\]
and let $\Rquad(u)$ be its positive root, so
\begin{equation}
 T_s(\Rquad(u))=0.
 \label{eq:opt:Rquad-equation}
\end{equation}
The positive-root characterization is the form used below; expanding it by
the quadratic formula gives an equivalent but uninformative radical.

At the other end, let $\uDet{5}>1$ solve
$\uDet{5}-3=\ln \uDet{5}$ and define
\begin{equation}
 \begin{aligned}
  r&=\frac2{\uDet{5}-1},
  \qquad
  \RdetOT=1+r=1.57\ldots,\\
  \uDet{5}&=4.50\ldots .
 \end{aligned}
\label{eq:opt:obligatory-constant}
\end{equation}
For each $\uDet{4}\le u\le \uDet{5}$, let $c=c(u)$ and $m=m(u)$ be the
unique solution of
\begin{equation}
 \operatorname{artanh}m-m
 =1-\frac1c+\frac12\ln\!\left(1+\frac2c\right),
 \qquad
 u=1+\frac2c-m,
 \label{eq:opt:cm-parametrization}
\end{equation}
with $r\le c\le1/\varphi$ and $0\le m\le1/\varphi$.

\begin{theorem}[Exact revealing-optimization curve]
\label{thm:optional-curve}
For every fixed $0<u<\infty$, the optimal deterministic size-asymptotic ratio in
the revealing-optimization model is
\begin{equation}
\boxed{
\RdetRO(u)=
\begin{cases}
1, &0<u\le\uDet{1},\\[1mm]
u, &\uDet{1}\le u\le \uDet{2},\\[1mm]
\Rquad(u), &\uDet{2}\le u\le \uDet{3},\\[1mm]
\displaystyle 1+\frac1{\sqrt{u-1}},
   &\uDet{3}\le u\le\uDet{4},\\[3mm]
1+c(u), &\uDet{4}\le u\le \uDet{5},\\[1mm]
\RdetOT, &u\ge \uDet{5}.
\end{cases}}
\label{eq:opt:full-curve}
\end{equation}
On every interval in \cref{eq:opt:full-curve}, there is an explicit
deterministic algorithm satisfying
\[
 \ALG_{\mathcal A}(I)\le \RdetRO(u)\cdot\OPT(I)+\,O_u(n).
\]
On $u\ge \uDet{5}$ the remainder is one absolute $O(n)$ term, uniform over all
finite $u$.
Conversely, for every deterministic algorithm $\mathcal A$ and every
$\eps>0$, there are arbitrarily large fixed instances $I$ such that
\[
 \ALG_{\mathcal A}(I)
 \ge\bigl(\RdetRO(u)-\eps\bigr)\cdot\OPT(I).
\]
For $u\ge \uDet{5}$, these lower-bound instances may be chosen with
$p_j<u-1$ for every job, so the offline cap $u$ is inactive.  Thus both
witnesses on the plateau are uniform in
the sense needed to pass separately to the obligatory endpoint.
\end{theorem}

\begin{theorem}[Deterministic obligatory testing]
\label{thm:det-obligatory}
Let $\uDet{5}>1$ be the unique solution of
$\uDet{5}-3=\ln \uDet{5}$, and set
\[
 \RdetOT=1+\frac{2}{\uDet{5}-1}
 =1.57\ldots .
\]
For every deterministic online algorithm $\mathcal B$ and every $\eps>0$,
there are arbitrarily large obligatory-testing instances $I$ such that
\[
 \ALG_{\mathcal B}(I)\ge(\RdetOT-\eps)\cdot\OPT(I).
\]
Conversely, there are a deterministic algorithm $\mathcal A$ and an absolute
constant $B$ such that every $n$-job obligatory instance satisfies
\[
 \ALG_{\mathcal A}(I)\le\RdetOT\cdot\OPT(I)+Bn.
\]
Consequently the optimal deterministic size-asymptotic competitive ratio is
exactly $\RdetOT$.
\end{theorem}

The proof of \cref{thm:optional-curve} combines two kinds of lower bounds.
Binary hidden-stopping instances give the formulas on
$0<u\le\uDet{4}$, while a continuous version of the harmonic construction
with an initial block of jobs satisfying $p_j=u$ gives the lower bound on
$\uDet{4}\le u\le\uDet{5}$ and the constant value for
$u\ge\uDet{5}$.  Three
explicit algorithms match them: raw execution, a forced-prefix uniform
threshold, and a history-dependent adaptive threshold.  They are introduced
informally in \cref{sec:intro-revealing-optimization} and specified formally
in \cref{alg:raw,alg:ute,alg:adaptive-threshold}.  The matching lower bounds
are developed in \cref{sec:det-branchwise-lower}.  The following table maps
each interval to its upper-bound witness.

\begin{center}
\footnotesize
\begin{tabular}{@{}L{0.16\linewidth}L{0.22\linewidth}L{0.17\linewidth}L{0.29\linewidth}@{}}
\toprule
range & algorithm & parameter choice & matching lower construction \\
\midrule
$0<u\le \uDet{2}$
 & \textsc{Raw} & none & binary hidden stopping \\
$\uDet{2}<u\le \uDet{3}$
 & \textsc{ForcedPrefixUTE} & $b=b(u)$ & binary hidden stopping \\
$\uDet{3}<u\le\uDet{4}$
 & \textsc{ForcedPrefixUTE} & $b=0$ & binary hidden stopping \\
$\uDet{4}<u<\uDet{5}$
 & \textsc{AdaptiveThreshold} & $c=c(u)$ & cap plus harmonic block \\
$\uDet{5}\le u<\infty$
 & \textsc{AdaptiveThreshold} & $c=r$ & harmonic core \\
\bottomrule
\end{tabular}
\end{center}

\subsection{Upper bounds}
\label{sec:det-algorithm-guarantees}

We prove the six interval guarantees using three algorithms.  Raw execution
covers $0<u\le\uDet{2}$, a forced-prefix uniform threshold covers
$\uDet{2}\le u\le\uDet{4}$, and an adaptive threshold covers
$u\ge\uDet{4}$.
Each of the next three parts treats one algorithm and ends with its
complete guarantee on the relevant range of $u$.

\subsubsection{Raw execution}
\label{sec:opt:low-upper}

We first handle the range in which testing is unnecessary.  Running every
job raw is optimal for $u\le1$ and gives the matching factor $u$ until the
first nontrivial transition point $\uDet{2}$.

\begin{algorithm}[H]
\caption{}
\label{alg:raw}
\begin{algorithmic}
\Statex \textbf{Input:} jobs $J$ and a finite common upper limit $u$.
\State Fix the job order $1,\ldots,n$.
\State \textbf{For} $i=1,\ldots,n$, execute job $i$ untested for $u$ time.
\end{algorithmic}
\end{algorithm}

\begin{lemma}[Raw-execution guarantee]
\label{lem:opt:raw-upper}
For every $0<u\le\uDet{2}$, \textnormal{\textsc{Raw}} satisfies
\begin{equation}
 \ALG_{\textnormal{\textsc{Raw}}}(I)\le \max\{1,u\}\cdot\OPT(I).
 \label{eq:opt:raw-upper}
\end{equation}
In particular, it is optimal for $0<u\le\uDet{1}$ and is $u$-competitive
for $\uDet{1}\le u\le\uDet{2}$.
\end{lemma}

\begin{proof}
The raw blocks all have length $u$, so
$\ALG_{\textnormal{\textsc{Raw}}}=u n(n+1)/2$.  If $u\le1$, every offline effective
length is $\min\{u,1+p_j\}=u$, and equality holds.  If $u>1$, every offline
effective length is at least one, so
$\OPT\ge n(n+1)/2$ and the ratio is at most $u$.
\end{proof}

\subsubsection{Forced-prefix uniform thresholds}
\label{sec:opt:forced-prefix-upper}

The next algorithm tests every job.  It processes a forced prefix regardless
of the observed values and thereafter processes only short outcomes
immediately.  We prove its complete guarantee on
$[\uDet{2},\uDet{3}]$ and then analyze the zero-prefix specialization,
namely $b=0$, which is optimal up to $\uDet{4}$.

Here UTE stands for \emph{uniform-threshold execution}: outside the forced
prefix, the algorithm uses the same threshold for every tested job.

\begin{algorithm}[H]
\caption{}
\label{alg:ute}
\begin{algorithmic}
\Statex \textbf{Input:} jobs $J$, finite $u>1$, and $b\in[0,1]$.
\State Fix the test order $1,\ldots,n$; set
       $k\gets\lfloor bn\rfloor$, $\theta_u\gets\min\{1,u-1\}$,
       and $Q\gets\varnothing$.
\State \textbf{For} $i=1,\ldots,n$ \textbf{do}:
\State \quad Test job $i$ for one unit and observe $p_i$.
\State \quad \textbf{If} $i\le k$ or $p_i\le\theta_u$, process job $i$
       immediately.
\State \quad \textbf{Otherwise}, add job $i$ to $Q$.
\State Process the jobs of $Q$ in nondecreasing revealed processing time.
\end{algorithmic}
\end{algorithm}

In this subsubsection, the first $k=\lfloor bn\rfloor$ positions in the fixed
test order are the \emph{prefix}, and the remaining positions are the
\emph{suffix}.  We call a tested job \emph{immediate} if it is processed at
once, and \emph{deferred} if it is placed in $Q$ for the final queue.

Recall from the beginning of \cref{sec:optional} that $\Rquad(u)$ denotes
the positive root specified by \cref{eq:opt:Rquad-equation}.

\begin{lemma}[Forced-prefix UTE guarantee]
\label{lem:opt:ute-endpoint-game}
For every $\uDet{2}\le u\le\uDet{3}$, define the algorithm parameter
\begin{equation}
 b=b(u)\defeq
 \frac{u^2-u+1-(u-1)^2\Rquad(u)}
      {u^2-u+1+(u-1)\Rquad(u)}.
 \label{eq:opt:ute-parameters}
\end{equation}
Then
\textnormal{\textsc{ForcedPrefixUTE}}$(J,u,b)$ satisfies
\begin{equation}
 \ALG\le\Rquad(u)\cdot\OPT+\,O_u(n).
 \label{eq:opt:ute-full-upper}
\end{equation}
\end{lemma}

We prove the lemma in three steps.  The first claim verifies that the root
$R$ is well-defined and that the prefix parameter is admissible.  The next
two claims prove the guarantee below and above $u=2$ by reducing the possible
processing times to boundary values in two different ways.  Throughout these claims,
put $R=\Rquad(u)$ and $s=u-1$.  In this notation, the
parameter from \cref{eq:opt:ute-parameters} is
\[
 b=\frac{s^2+s+1-s^2R}{s^2+s+1+sR}.
\]

\begin{claim}[Parameter admissibility]
\label{clm:opt:ute-admissibility}
The positive root $R$ is well-defined.  Moreover, $0\le b<1$, with $b=0$
only at $u=\uDet{3}$, and $R>5/3$ and $s^2R>16/15$ throughout the interval.
\end{claim}

\begin{proof}
Since $s>0$, we have $T_s(0)<0$ and $T_s(z)\to\infty$ as
$z\to\infty$.  Moreover, $T_s$ is strictly increasing for $z>0$, because
\[
 T_s'(z)=2sz+s^3+s^2+1>0.
\]
It therefore has a unique positive root $R$.

Since $T_s$ is increasing on positive arguments,
\begin{equation}
 T_s\!\left(\frac{s^2+s+1}{s^2}\right)
 =\frac{s^2+s+1}{s^3}\bigl(-s^3+s^2+2s+1\bigr)\ge0
 \qquad(s\le s_0),
\end{equation}
and hence $s^2+s+1-s^2R\ge0$, with equality only at $s=s_0$.
The denominator in $b$ is positive and exceeds its numerator by
$sR(s+1)$, proving $b<1$.

We will also use the following lower bound on $R$:
\begin{equation}
 9T_s(5/3)=6s^3-12s^2-2s-3<0
 \label{eq:opt:T-five-thirds}
\end{equation}
throughout the present interval.  Indeed, on $[4/5,1]$ the polynomial is
decreasing and has value $-1151/125$ at $4/5$; on $[1,s_0]$ it has one
interior minimum, while its value at $1$ is $-11$ and its value at $s_0$
is smaller than its value $-95/36$ at $13/6$.  Here
$s_0<13/6$ follows by substituting $13/6$ in
$s^3-(s+1)^2$.  Thus $R>5/3$, while $s>4/5$, and
$s^2R>16/15$.
\end{proof}

\begin{claim}[The range $\uDet{2}\le u\le2$]
\label{clm:opt:ute-low-range}
For every $\uDet{2}\le u\le2$,
\textnormal{\textsc{ForcedPrefixUTE}}$(J,u,b)$ satisfies
\[
 \ALG\le R\cdot\OPT+\,O_u(n).
\]
\end{claim}

\begin{proof}
We first consider a binary input in which every job is either zero
($p_i=0$) or long ($p_i=u$).  The first $bn$ positions are the forced
prefix.  A long job in this prefix is processed immediately; a long job in
the suffix is deferred; every zero job is processed immediately.

Fix the total number $L$ of long jobs, and write $k=bn$ for the moment;
integer rounding will be restored at the end.  Their positions do not affect
OPT, which knows all values and sorts by effective length.  We determine the
worst positions for ALG directly.

Suppose exactly $q$ long jobs lie in the prefix and put $D=L-q$.  Within the
prefix, moving a long job before a zero delays that zero and every later
immediate completion by the long processing time $u$.  Within the suffix,
moving a deferred long job before a zero delays the zero by one unit of test
time and does not change the final queue.  Thus, for fixed $q$, the worst
order is

\[
 \underbrace{u,\ldots,u}_{q}
 \underbrace{0,\ldots,0}_{k-q}
 \mathbin{|}
 \underbrace{u,\ldots,u}_{D}
 \underbrace{0,\ldots,0}_{n-k-D},
\]
where the vertical bar separates the forced prefix from the suffix.  The
completion cost of this order is
\begin{equation}
 C(q)=\frac{n(n+1)}2
 +D\left(n-k-\frac{D+1}{2}\right)
 +u\left(qn-\frac{q(q-1)}2+\frac{D(D+1)}2\right).
 \label{eq:opt:binary-order-cost}
\end{equation}
The first term is the unit-test baseline.  The second records the positions
of the tests of the $D$ long jobs in the suffix, and the last records all
delays caused by long processing.  Moving one long job from the suffix into
the prefix changes
$q$ to $q+1$ and $D$ to $D-1$, and hence
\[
 C(q+1)-C(q)=(u-1)(n-L)+k-q\ge0.
\]
Consequently a worst input puts as many long jobs in the prefix as possible,
then all remaining long jobs at the start of the suffix, and all zeros last.

Divide all job counts by $n$, and let $x=L/n$ be the fraction of long jobs.
The fraction of long jobs in the prefix is therefore $a=\min\{b,x\}$.  If
$x>b$, the remaining fraction $d=x-b$ consists of deferred long jobs; if
$x\le b$, then
$d=0$.  Ignore for the moment the $O_u(n)$ one-job and rounding terms, and
write $A=\ALG/n^2$ and $O=\OPT/n^2$ for the leading coefficients.  Direct
pair counting gives
\begin{align}
 O(x)&=\frac12+\frac{s}{2}x^2,
 \notag\\
 A(x)&=
 \begin{cases}
 \displaystyle
  \frac12+u\left(x-\frac{x^2}{2}\right),
     &0\le x\le b,\\[3mm]
 \displaystyle
  \frac12+u\left(b-\frac{b^2}{2}\right)
  +(1-b)(x-b)+\frac{s}{2}(x-b)^2,
     &b\le x\le1.
 \end{cases}
 \label{eq:opt:binary-upper-gap}
\end{align}
The term $1/2$ is the common unit-time baseline.  In $O(x)$, the extra
amount $s=u-1$ is paid only by pairs of long jobs.  The two cases for $A(x)$
distinguish whether all long jobs fit into the forced prefix.

Put $G(x)=A(x)-R\cdot O(x)$.  We show that its maximum is zero.
On $b\le x\le1$, substituting the second formula for $A(x)$ from
\cref{eq:opt:binary-upper-gap} makes $G(x)$ a strictly concave quadratic.
Its stationary point is
\[
 \alpha=\frac{s^2+s+1}{s^2+s+1+sR}.
\]
Moreover, $b=(s+1)\alpha-s$, and hence
$\alpha-b=s(1-\alpha)>0$.  Substitution gives $G(\alpha)=0$; after
cancelling the positive denominator, this is exactly
\cref{eq:opt:Rquad-equation}.  Since $b<\alpha$, its right derivative at $b$
is positive.  The left derivative satisfies
$G'_-(b)=G'_+(b)+s(1-b)>0$, and the first quadratic is concave, so it is
increasing on $[0,b]$.  Thus $G(x)\le0$ for every $x\in[0,1]$, including
the limiting case $b=0$: the two pieces agree at $b$, the first increases
up to that common value, and the second never exceeds its zero maximum.

We now reduce a general input to the binary case just analyzed.  In the
present range the threshold is $u-1\le1$.  Thus a suffix job is immediate
exactly when $p\le u-1$ and deferred otherwise.  First increase every
deferred suffix value to $u$.  This cannot decrease ALG, and it leaves OPT
unchanged because such a job already has effective length
$\min\{u,1+p\}=u$.

For the remaining jobs, construct a random binary input as follows,
independently from job to job.  Replace a prefix value $p$ by $u$ with
probability $p/u$ and by zero otherwise.  Replace an immediate suffix value
$p\le u-1$ by $u$ with probability $p$ and by zero otherwise.  The latter
probability is valid because $p\le u-1\le1$.  A suffix job rounded to $u$
changes from immediate to deferred; the comparison below includes this
change of action.

To compare online costs, consider a pair of positions $i<j$ in the test
order.  Its leading contribution to ALG is
\begin{equation}
 g_{ij}=\begin{cases}
 1+p_i,&i\text{ is immediate},\\
 2+p_j,&i\text{ is deferred and }j\text{ is immediate},\\
 2+\min\{p_i,p_j\},&i,j\text{ are deferred}.
 \end{cases}
 \label{eq:opt:pair-contribution}
\end{equation}
A prefix job remains immediate, and its rounding preserves its expected
processing value.  If $i$ is an immediate suffix job, then after rounding its
pair contribution is one when it becomes zero and at least two when it
becomes $u$; its expectation is therefore at least $1+p_i$.  If $i$ is
deferred and $j$ is an immediate suffix job, then $i$ has already been raised
to $u$, and rounding $j$ gives expected contribution
$2+up_j\ge2+p_j$.  Finally, a pair of deferred jobs can only become more
expensive.  These cases show that the expected leading online cost does not
decrease.

It remains to check that the same rounding does not increase OPT.  If
$\lambda'$ denotes the effective length after rounding, then for a prefix job
and an immediate suffix job, respectively,
\begin{equation}
 \EE\lambda'=1+\frac{u-1}{u}p\le \min\{u,1+p\},
 \qquad
 \EE\lambda'=1+(u-1)p\le1+p.
\end{equation}
A deferred suffix job has effective length $u$ both before and after the
first step.  Since the minimum is jointly concave,
\[
 \EE\min\{\lambda_i',\lambda_j'\}
 \le \min\{\EE\lambda_i',\EE\lambda_j'\}
 \le \min\{\lambda_i,\lambda_j\}.
\]
The offline pair identity therefore shows that expected OPT does not
increase.  Consequently the expected value of
$\ALG-\,R\cdot\OPT$ for the binary input is at least its value for the original
input, up to $O_u(n)$ one-job terms.  Some binary realization is at least as
bad as this expectation, and the binary bound proves the claim.
\end{proof}

\begin{claim}[The range $2\le u\le\uDet{3}$]
\label{clm:opt:ute-high-range}
For every $2\le u\le\uDet{3}$,
\textnormal{\textsc{ForcedPrefixUTE}}$(J,u,b)$ satisfies
\[
 \ALG\le R\cdot\OPT+\,O_u(n).
\]
\end{claim}

\begin{proof}
Here $u\ge2$, so the algorithm uses threshold one.  A prefix job is always
immediate, a suffix job with $p\le1$ is immediate, and a suffix job with
$p>1$ is deferred.  First keep this immediate/deferred pattern fixed and
vary one processing time at a time in the leading excess
$\ALG-\,R\cdot\OPT$.  On each interval between the threshold and the offline
cap $u$,
the pair formula \cref{eq:opt:pair-contribution} makes this excess convex in
that coordinate.  Hence a maximum occurs at an interval endpoint.  This
reduces
\begin{itemize}
\item a prefix value to $0$ or $u$;
\item an immediate suffix value to $0$ or $1$; and
\item a deferred suffix value to $1^+$ or $u$.
\end{itemize}
Here $1^+$ means a value tending to one from above, so the job remains
deferred.  For the last reduction, on $1<p<u-1$ the only nonlinear term is a
concave minimum multiplied by $1-R<0$, hence it is convex; on $[u-1,u]$ the
offline effective length is already the constant $u$, while the online cost
is nondecreasing in $p$.  In the rest of this proof, we call a job with
$p=u$ \emph{capped}.

We next choose the worst test order for fixed numbers of these endpoint
types.  OPT depends only on those numbers.  In ALG, swapping two neighboring
positions shows the following: capped jobs in the prefix may be moved before
zero jobs in the prefix; deferred jobs in the suffix may be moved before
immediate jobs in the suffix; and suffix jobs of value one may be moved
before zero jobs in the suffix.  These are three direct comparisons, not an
additional scheduling principle: substitute the processing times of the two
jobs being swapped in
\cref{eq:opt:pair-contribution}; all unaffected pair terms cancel and the
remaining difference is nonnegative.  The order among two deferred jobs is
irrelevant because the last line of \cref{eq:opt:pair-contribution} is
symmetric.

Divide each class size by $n$.  Let $a$ be the resulting fraction of prefix
jobs with $p=u$, $d$ the fraction of suffix jobs with $p=u$, $t$ the fraction
of deferred suffix jobs with $p=1^+$, and $m$ the fraction of immediate
suffix jobs with $p=1$.  All other jobs have value zero.  The constraints are
\begin{equation}
 0\le a\le b,
 \qquad d,t,m\ge0,
 \qquad d+t+m\le1-b.
\end{equation}
As in the preceding claim, let $A$ and $O$ denote the leading coefficients
of $\ALG/n^2$ and $\OPT/n^2$.  Summing
\cref{eq:opt:pair-contribution} over the ordered endpoint classes gives
\begin{align}
 A={}&\frac12+(s+1)\left(a-\frac{a^2}{2}\right)
 +(1-b)(d+t+m)-\frac{m^2}{2}+\frac{s}{2}d^2,
 \label{eq:opt:ute-A}\\
 O={}&\frac12+\frac{(a+d+t+m)^2}{2}
 +\frac{s-1}{2}(a+d)^2.
 \label{eq:opt:ute-O}
\end{align}
The formula for $O$ is especially transparent: every job contributes the
unit baseline; all four nonzero classes have one additional unit of effective
length; and the two capped classes $a,d$ have another $s-1=u-2$ units.
Replacing an immediate value-one suffix job by a deferred value-$1^+$ job
does not change $O$ in the limit and increases $A$.  Thus a worst reduced
instance has $m=0$.  Put $x=a+d$, the total fraction of capped jobs, and maximize
$A-R\cdot O$ over the remaining split between $a,d,t$.  We retain the stationary
point from the preceding claim,
\[
 \alpha=\frac{s^2+s+1}{s^2+s+1+sR}.
\]
For fixed $x$, the derivative with respect to the deferred $p=1^+$ fraction $t$
is $1-b-R(x+t)$.  Thus the best choice is
$t=(\chi-x)_+$, where
\[
 \chi=\frac{1-b}{R}
      =\frac{s(s+1)}{s^2+s+1+sR}.
\]
The identities
\[
 \alpha-\chi=\frac1{s^2+s+1+sR},
 \qquad
 \chi-b=\frac{s^2R-1}{s^2+s+1+sR}
\]
show that $b<\chi<\alpha$ on the present range.  At fixed
$x,t$, replacing a capped suffix job by a capped prefix job leaves $O$
unchanged and changes $A$ at rate $s(1-x)+b-a\ge0$.  We therefore fill the
prefix with capped jobs as far as the constraints permit.  Altogether,
\begin{equation}
 t=(\chi-x)_+,
 \qquad
 a=\min\{b,x\},
 \qquad d=(x-b)_+.
 \label{eq:opt:ute-reduction}
\end{equation}

It remains to optimize the single variable $x$.  If $x\ge\chi$, then $t=0$,
so only values zero and $u$ remain.  This is exactly the binary calculation
from \cref{clm:opt:ute-low-range}; its maximum gap is zero at
$x=\alpha$.  If $x<\chi$, then $t=\chi-x>0$, and substitution gives
\begin{equation}
 A-R\cdot O
 =-\frac{R-1}{2}+\frac{(1-b)^2}{2R}
  +a(s+b-sx)-\frac{a^2}{2}+\frac{s-R(s-1)}{2}x^2.
 \label{eq:opt:ute-positive-t}
\end{equation}
The coefficient $s-R(s-1)$ is positive on $1\le s\le s_0$.  Indeed,
$T_s(2)=s(s^2-s+1)>0$, so for $s\le2$ the equation $T_s(R)=0$ gives
$R<2$, and hence $s-R(s-1)>2-s\ge0$.  For
$2\le s\le s_0<13/6$,
\begin{equation}
 16T_s(7/4)=12s^3-20s^2+s-4>0;
\end{equation}
the polynomial is $14$ at $s=2$ and has positive derivative thereafter.
Thus $R<7/4$ and $s-R(s-1)>7/4-3s/4>1/8$.

On $0\le x\le b$, one has $a=x$, and
\begin{equation}
 A-R\cdot O=-\frac{R-1}{2}+\frac{(1-b)^2}{2R}+(s+b)x
 -\frac{s+1+R(s-1)}2x^2.
 \label{eq:opt:ute-small-x}
\end{equation}
Its derivative is decreasing, and the derivative at $b$ multiplied by the
positive denominator $s^2+s+1+sR$ is
\begin{equation}
 R\bigl(1+s^2+R s^2(s-1)\bigr)>0,
\end{equation}
so the maximum is at $b$.  On $b\le x\le\chi$,
\cref{eq:opt:ute-positive-t} is convex and its maximum is at $b$ or $\chi$.
At $x=\chi$ one has $t=0$, so this endpoint is already covered by the binary
calculation.  It remains to compare the other endpoint $x=b$ with it:
\begin{equation}
 F(b)-F(\chi)
 =\frac{\chi-b}{2}\,[2sb-(s-R(s-1))(b+\chi)],
\end{equation}
and direct substitution of $b$ and $\chi$ gives
\begin{equation}
 (s^2+s+1+sR)[(s-R(s-1))(b+\chi)-2sb]
 =-s+(1+s-s^3)R+s^2(s-1)R^2>0.
 \label{eq:opt:ute-small-certificate}
\end{equation}
To prove the remaining sign without numerical approximation, observe that
$T_s$ is increasing on positive arguments and
\begin{equation}
 9T_s(5/3)=6s^3-12s^2-2s-3<0,
\end{equation}
so $R>5/3$.  At $R=5/3$, the derivative of the right-hand side of
\cref{eq:opt:ute-small-certificate} satisfies
\begin{equation}
 3\left.\frac{\partial}{\partial R}
 \bigl[-s+(1+s-s^3)R+s^2(s-1)R^2\bigr]\right|_{R=5/3}
 =7s^3-10s^2+3s+3>0.
\end{equation}
The last polynomial is increasing for $s\ge1$.  Since the coefficient of
$R^2$ is nonnegative, the right-hand side of
\cref{eq:opt:ute-small-certificate} is therefore increasing for
$R\ge5/3$.  At $R=5/3$ its value satisfies
\begin{equation}
 9\bigl[-s+(1+s-s^3)(5/3)+s^2(s-1)(25/9)\bigr]
 =10s^3-25s^2+6s+15>0.
\end{equation}
For the last sign, write $s=3/2+z$; the expression becomes
$3/2-3z/2+20z^2+10z^3$.  It is at least $3/2$ for
$-1/2\le z\le0$, and for $z\ge0$ it is larger than
$3/2-9/320$.  Hence both possible maxima on $[b,\chi]$ are nonpositive.
Together with the binary case $x\ge\chi$, this proves $A-R\cdot O\le0$ for every
reduced configuration.  Restoring the one-job terms and rounding
$\lfloor bn\rfloor$ contributes only $O_u(n)$.
\end{proof}

We can now prove the lemma.

\begin{proof}[Proof of \Cref{lem:opt:ute-endpoint-game}]
By \Cref{clm:opt:ute-admissibility}, the displayed value of $b$ is a valid
algorithm parameter throughout $[\uDet{2},\uDet{3}]$.  The guarantees in
\Cref{clm:opt:ute-low-range,clm:opt:ute-high-range} cover this interval and
give \cref{eq:opt:ute-full-upper} with $R=\Rquad(u)$.
\end{proof}

At $u=\uDet{3}$, $b=0$.  For
$\uDet{3}\le u\le\uDet{4}$, the zero-prefix specialization gives the
ratio $1+1/\sqrt{u-1}$.

\begin{lemma}[Zero-prefix guarantee]
\label{lem:opt:zero-prefix}
Fix $\uDet{3}\le u\le\uDet{4}$ and set $b=0$ in \cref{alg:ute}.  Then
\begin{equation}
 \ALG\le\left(1+\frac1{\sqrt{u-1}}\right)\cdot\OPT+\,O_u(n).
 \label{eq:opt:zero-prefix-guarantee}
\end{equation}
\end{lemma}

\begin{proof}
Here $b=0$, so there is no forced prefix, and the threshold is one.  The
coordinatewise argument from \cref{clm:opt:ute-high-range} reduces every
input to four types: deferred jobs with $p=u$, deferred jobs with $p=1^+$,
immediate jobs with $p=1$, and zero jobs.  Divide their class sizes by $n$
and denote the first three resulting fractions by $d,t,m$, respectively;
put $s=u-1$.  Thus
$d,t,m\ge0$ and $d+t+m\le1$.  The leading online and offline coefficients
are
\begin{equation}
 A=\frac12+d+t+m-\frac{m^2}{2}+\frac{s d^2}{2},
 \qquad
 O=\frac12+\frac{(d+t+m)^2}{2}+\frac{(s-1)d^2}{2}.
\end{equation}
As before, replacing an immediate value-one job by a deferred value-$1^+$
job preserves $O$ in the limit and increases $A$.  We may therefore take
$m=0$ and write $y=d+t$ for the total positive fraction.  For fixed $y$, both
$A$ and $O$ are affine functions of $d^2$, so their ratio is monotone in
$d^2$.  Its maximum is attained at one of the two extreme compositions:
\begin{equation}
 \frac{1+2y}{1+y^2}
 \quad(d=0),
 \qquad
 1+\frac{2y}{1+s y^2}
 \quad(d=y).
\end{equation}
The first case has only value-$1^+$ positive jobs; the second has only
positive jobs with $p=u$.
For $d=y$, the desired bound is equivalent to
$(\sqrt{s}y-1)^2\ge0$.  For $d=0$, put $r=1/\sqrt{s}$.  The desired
inequality is
\[
 (1+r)y^2-2y+r\ge0.
\]
On the range of the lemma, $s=u-1\le\uDet{4}-1=\varphi^2$, and hence
$r(1+r)\ge1$.  The quadratic on the left has nonpositive discriminant and
is therefore nonnegative for every $y$.  Thus both extreme compositions are
bounded by $1+1/\sqrt{u-1}$, as required.  Binary instances attain this
value.
\end{proof}

\subsubsection{Adaptive thresholds for the final two parameter intervals}
\label{sec:optional-upper}

The final algorithm changes its threshold as work accumulates.  We prove its
guarantee on $[\uDet{4},\uDet{5}]$ and its uniform plateau
guarantee for $u\ge\uDet{5}$.  In the proof, a revealed value $p_i=u$ must
be treated separately from a smaller value placed in the final queue,
because the two have different pair contributions to the offline optimum.
Throughout this subsection, $R=1+c$, and we retain $\varphi$, $\uDet{5}$,
and $r$ from
\cref{eq:opt:u1-u2,eq:opt:obligatory-constant}.

For every $c>0$ define the history-dependent threshold
\begin{equation}
 a_c(y)=
 \begin{cases}
  1, & y\le-1,\\[1mm]
  \displaystyle
  1+\frac12\ln\!\left(\frac{c+2}{c-2y}\right),
      & -1\le y\le0.
 \end{cases}
 \label{eq:opt:threshold}
\end{equation}
The following single algorithm covers both parameter regimes
$\uDet{4}\le u\le\uDet{5}$ and $u\ge\uDet{5}$.  It does not use raw
execution, so its schedule does not depend on $u$.

The formula is best read backwards from the pair accounting.  The state
\[
 y=\frac{N_{\rm def}-R N}{x}
\]
is the normalized balance between positive jobs already deferred and the
$R$-weighted number of positive outcomes seen so far, per still-untested
label.  The adversarial endpoint types $M$ and $Q$ sit immediately below and
above the current threshold.  Equalizing their limiting local drifts forces
\begin{equation}
 a_c'(y)=\frac1{c-2y},
 \qquad a_c(-1)=1,
 \label{eq:opt:threshold-ode-motivation}
\end{equation}
whose solution is exactly \cref{eq:opt:threshold}.  Integrating the remaining
endpoint drifts produces the base potential below.  Jobs at the offline cap
$p=u$ have a different pair charge; on the middle interval their remaining
part is stored in a separate cap reserve, which vanishes at the plateau
$u=\uDet{5}$.  This derivation explains the threshold, the base potential,
and the additional reserve term before the verification starts.

\begin{algorithm}[H]
\caption{}
\label{alg:adaptive-threshold}
\begin{algorithmic}
\Statex \textbf{Input:} jobs $J$ and parameter $c>0$.
\State Fix the test order $1,\ldots,n$.
\State Set $N\gets0$, $N_{\mathrm{def}}\gets0$, and $Q\gets\varnothing$.
\State \textbf{For} $i=1,\ldots,n$ \textbf{do}:
\State \quad Set $x\gets n-i+1$ and $y\gets(N_{\mathrm{def}}-(1+c)N)/x$.
\State \quad Set $a\gets a_c(y)$ according to \cref{eq:opt:threshold}.
\State \quad Test job $i$ for one unit and observe $p_i$.
\State \quad \textbf{If} $p_i\le a$, process job $i$ immediately.
\State \quad \textbf{Otherwise}, add job $i$ to $Q$ and set $N_{\mathrm{def}}\gets N_{\mathrm{def}}+1$.
\State \quad \textbf{If} $p_i>0$, set $N\gets N+1$.
\State Process the jobs of $Q$ in nondecreasing order of their revealed
       processing times.
\end{algorithmic}
\end{algorithm}

For $c\in[r,1/\varphi]$, let $m=m(c)\in[0,1/\varphi]$ and $u=u(c)$ be
determined by the $(c,m)$-parametrization in
\cref{eq:opt:cm-parametrization}.
The right-hand side of the first equation is strictly increasing in $c$,
as is the left-hand side in $m$.  Consequently $m(c)$ is increasing and
$u(c)$ is decreasing.  At the two boundary values of the parameter $c$, the
parametrization gives
\[
 c=\frac1\varphi:\quad
 (m,u)=\left(\frac1\varphi,\uDet{4}\right),
 \qquad
 c=r:\quad (m,u)=(0,\uDet{5}).
\]

\begin{theorem}[Middle and high revealing upper bounds]
\label{thm:opt:optional-middle-upper}
For every finite $u\in[\uDet{4},\infty)$, the formally specified algorithm
\textnormal{\textsc{AdaptiveThreshold}} from
\cref{alg:adaptive-threshold}
satisfies
\[
 \ALG\le
 \begin{cases}
  (1+c)\cdot\OPT+\,O_u(n),
     & \uDet{4}\le u\le \uDet{5},
       \quad c\text{ given by \eqref{eq:opt:cm-parametrization}},\\[1mm]
  (1+r)\cdot\OPT+\,O(n),
     & u\ge \uDet{5}.
 \end{cases}
\]
The constant in the last $O(n)$ is absolute and uniform over
$u\in[\uDet{5},\infty)$.
In particular these are upper bounds on the deterministic
size-asymptotic competitive ratio.
\end{theorem}

We prove the theorem through six claims.  The first two rewrite
$\ALG-\,R\cdot\OPT$ as a sum of contributions from ordered pairs of jobs plus
$O_u(n)$ one-job terms, reduce the processing times to five boundary values,
and assign the resulting pair contributions to local charges $\ell_i$.
The next three claims construct and calibrate a nonnegative potential $W$
such that, for every endpoint type $T$ with state-update direction $v_T$,
\[
 \ell_T+\nabla W\mathbin{\cdot}v_T\le0.
\]
Here $\nabla W\mathbin{\cdot}v_T$ is the first-order change obtained by
differentiating $W$ in the direction of that state update.
The sixth claim uses a Taylor estimate to compare this directional derivative
with the actual potential change caused by one job, and then telescopes the
resulting finite inequalities.

\paragraph{Common setup.}
Fix a finite raw-execution cap $u\ge\uDet{4}$ and put $R=1+c$.
Use \textsc{AdaptiveThreshold}$(J,c)$ from
\cref{alg:adaptive-threshold}.  When $u\ge\uDet{5}$, set $c=r$.
For the analysis, immediately before a test let $x$ be the number of jobs not
yet tested, $N$ the number of earlier tests that revealed a positive value,
and $N_{\mathrm{def}}$ the number of those jobs placed in the final queue.
Thus
\begin{equation}
 y=\frac{N_{\mathrm{def}}-RN}{x}\le0
 \label{eq:opt:state-y}
\end{equation}
and the threshold is given by \cref{eq:opt:threshold}.  When
$\uDet{4}\le u\le\uDet{5}$, the parametrization
\cref{eq:opt:cm-parametrization} gives $c\ge r$, and hence
\begin{equation}
 a_c(y)\le a_c(0)\le a_r(0)=\frac1r<u-1.
 \label{eq:opt:threshold-below-cap}
\end{equation}
When $u\ge\uDet{5}$, we set $c=r$, so
\cref{eq:opt:threshold-below-cap} holds there as well.
Thus every revealed value $p_i\in[u-1,u]$, for which the offline effective
length has reached the offline cap $u$, is placed in the final queue.

The state $y$ never increases.  Indeed, the numerator in
\eqref{eq:opt:state-y} changes by $0$, $-R$, or $-c$ after a zero, a
positive job processed immediately after its test, or a positive job placed
in the final queue, respectively, while $x$ decreases.  This monotonicity
will allow us to allocate every pair cost when its relevant endpoint becomes
known.

\begin{claim}[Reduction to five endpoint types]
\label{clm:opt:adaptive-endpoint-reduction}
Fix a decision transcript of \textnormal{\textsc{AdaptiveThreshold}}: for each tested job,
fix whether its revealed value is zero or positive and whether the job is
processed immediately or placed in the final queue.  This also fixes the
state and threshold before every subsequent test.  Among all inputs
consistent with this transcript, the supremum of the pairwise term
$\Gamma_n$ from \cref{eq:opt:pair-objective} is unchanged if every job is
restricted to one of the five endpoint types in
\cref{eq:opt:five-endpoints}.  The symbols $0^+$ and $a_c(y)^+$ denote
one-sided limits at which the fixed transcript is preserved.
\end{claim}

\begin{proof}
Write
\[
 \bar p_i=\min\{p_i,u-1\}.
\]
By \cref{eq:model-effective-length,lem:offline}, the clairvoyant revealing
optimum is SPT on jobs of effective lengths
$1+\bar p_i=\min\{1+p_i,u\}$.  For two jobs tested in the order $i<j$, let
\begin{equation}
 e_{ij}=
 \begin{cases}
  (p_i-p_j)_+,
     &i\text{ is processed immediately},\\
  1+(p_j-p_i)_+,
     &i\text{ is placed in the final queue and }j\text{ is processed immediately},\\
  1,
     &i\text{ and }j\text{ are placed in the final queue}.
 \end{cases}
 \label{eq:opt:pair-excess}
\end{equation}
The first line applies whenever $i$ is processed after its test, regardless
of the action taken on $j$.  Pairwise completion-time accounting for the online
schedule gives
\begin{equation}
 \ALG
 =\sum_i(1+p_i)
   +\sum_{i<j}\bigl(1+\min\{p_i,p_j\}+e_{ij}\bigr),
 \label{eq:opt:alg-pair-decomposition}
\end{equation}
whereas the common offline pair identity \cref{eq:optional-opt-pair} gives
\begin{equation}
 \OPT
 =\sum_i(1+\bar p_i)
   +\sum_{i<j}\bigl(1+\min\{\bar p_i,\bar p_j\}\bigr).
 \label{eq:opt:opt-pair-decomposition}
\end{equation}
Collecting the pair contributions in $\ALG-\,R\cdot\OPT$, define
\begin{equation}
 \Gamma_n=\sum_{i<j}
 \left(e_{ij}+\min\{p_i,p_j\}
       -R\min\{\bar p_i,\bar p_j\}\right).
 \label{eq:opt:pair-objective}
\end{equation}
Equations \eqref{eq:opt:alg-pair-decomposition}--%
\eqref{eq:opt:pair-objective} yield the exact identity
\begin{equation}
 \ALG-\,R\cdot\OPT
 =\Gamma_n-c\binom n2
  +\sum_i\bigl[(1+p_i)-R(1+\bar p_i)\bigr].
 \label{eq:opt:gap-decomposition}
\end{equation}
The final sum contains one term per job, rather than one per pair, and is
nonpositive.  If $p_i<u-1$, so that the offline cap $u$ is inactive, its
$i$-th summand is $-c(1+p_i)$.  If $p_i\ge u-1$, its offline effective
length is $u$ and the summand is at most $1-cu\le0$.  Thus it is enough to
control $\Gamma_n$.

Now fix the decision transcript described in the claim.  It determines all
later threshold states because these states depend only on whether each
earlier value was zero or positive and whether it was processed immediately
or placed in the final queue.  Varying a processing value inside an interval
that preserves this transcript therefore does not change any later
threshold.  On each such interval,
\eqref{eq:opt:pair-objective} is convex in that value, so its maximum occurs
at a boundary.  The resulting five endpoint types are
\begin{equation}
\begin{array}{c|c|c}
 \text{type}&\text{value}&\text{action}\\ \hline
 Z&0&\text{process after the test}\\
 E&0^+&\text{process after the test}\\
 M&a_c(y)&\text{process after the test}\\
 Q&a_c(y)^+&\text{place in the final queue}\\
 U&u&\text{place in the final queue}.
\end{array}
 \label{eq:opt:five-endpoints}
\end{equation}
Here $0^+$ means a positive value tending to zero, and $a_c(y)^+$ means a
value tending to the current threshold from above.  The types $E$ and $Q$
are kept distinct from $Z$ and $M$, respectively, because approaching either
boundary from the other side produces a different state update or action.
For a job processed after its test, the positive
interval has endpoints $0^+$ and $a_c(y)$.  If a value $p_i<u-1$ is placed
in the final queue, its interval has endpoints $a_c(y)^+$ and $u-1$.  On
$[u-1,u]$ the offline contribution is constant and the online contribution
is nondecreasing, so $u$ is worst.  Thus the intermediate value $u-1$ can
be moved further to $u$.  This proves the claimed reduction without
changing any later threshold decision.
\end{proof}

\begin{claim}[Local pair allocation]
\label{clm:opt:adaptive-local-allocation}
Use the following mixed precharge/refund accounting.  When an endpoint of
type $M$ is processed immediately, precharge its processing time against
every label that is still untested.  When a later endpoint is exposed,
charge or refund the part of each pair not already covered by such a
precharge.  The resulting total charge is at most the sum of the local
quantities $\ell_i$ defined below.  To express them, distinguish earlier endpoints of
type $E$
from the positive endpoint types $M,Q,U$.  Let
\[
 P=\#\{\text{earlier }M,Q,U\},\qquad
 E=\#\{\text{earlier }E\},\qquad
 N_{\mathrm{def}}=\#\{\text{earlier }Q,U\},
\]
and let $K$ be the number of earlier endpoints of type $U$.  Let $I,d$ be the
numbers of earlier endpoints of types $M,Q$, respectively, and let $A_I,A_D$ be the
corresponding sums of their processing values.  Thus
$P=I+d+K$ and $N_{\mathrm{def}}=d+K$.

For a current endpoint of type $Z,E,M,$ or $Q$, with $a=a_c(y)$, the
allocated charges are bounded by
\begin{equation}
 \ell_Z=\ell_E=N_{\mathrm{def}},
 \qquad
 \ell_M=N_{\mathrm{def}}+a(x+N_{\mathrm{def}}-RP),
 \qquad
 \ell_Q=N_{\mathrm{def}}+a(N_{\mathrm{def}}-RP).
 \label{eq:opt:noncapped-local-charges}
\end{equation}
For a current endpoint $U$, representing $p=u$, direct pair counting gives
the exact charge
\begin{align}
 \ell_U
 &=N_{\mathrm{def}}+A_D+uK-R\bigl(A_I+A_D+(u-1)K\bigr)
 \notag\\
 &=d-RA_I-cA_D+K h_U,
 \qquad
 h_U=1+u-R(u-1).
 \label{eq:opt:cap-local-charge}
\end{align}
\end{claim}

\begin{proof}
We make the accounting explicit.  If an earlier type-$M$ endpoint has value
$a_i$, its precharge contributes $a_i$ to its pair with every later label.
For example, the actual pair charge for $M$ followed by $Z$ is $a_i$ even
though $\ell_Z=0$ when no job has been deferred; it is exactly this earlier
precharge that pays the pair.  This is why a purely later-endpoint charging
description would be false.

Now expose a current endpoint with threshold $a$.  Threshold monotonicity
gives $a_i\ge a$ for all earlier types $M$ and $Q$.  After subtracting the
type-$M$ precharges, \cref{eq:opt:pair-objective} gives the following
remaining charges.  A current $Z$ or $E$ receives one from every earlier
deferred endpoint, hence $N_{\rm def}$.  A current $M$ precharges $a$ against
its $x-1$ later labels and receives
$-Ra$ from every earlier $M$ and $1-ca$ from every earlier deferred
endpoint.  Replacing $x-1$ by $x$ only overcharges by $a\ge0$ and yields
\[
 N_{\rm def}+a(x+N_{\rm def}-RP)=\ell_M.
\]
A current $Q$ makes no precharge and receives the same earlier-endpoint
terms, giving
\[
 -RaI+(1-ca)N_{\rm def}
 =N_{\rm def}+a(N_{\rm def}-RP)=\ell_Q.
\]

Finally, for a current $U$, an earlier type $M$ of value $p$ has actual pair
charge $-cp$; after its precharge $p$, the remaining charge is $-Rp$.
An earlier $Q$ of value $p$ contributes $1-cp$, and an earlier $U$
contributes $h_U$.  Summing these terms gives exactly
\cref{eq:opt:cap-local-charge}.
Summing over the five endpoint types gives
\begin{equation}
 \Gamma_n\le\sum_i\ell_i;
 \label{eq:opt:pair-allocation-dominates}
\end{equation}
the difference is exactly the sum of nonnegative extra delays by which the
displayed $M$ charges overestimate the true pair contributions.
Formula
\eqref{eq:opt:cap-local-charge}, rather than a substitution $p=u$ into
$\ell_Q$, is essential before the $U$-endpoints have been put into a prefix.

Threshold monotonicity gives $A_I\ge aI$ and $A_D\ge ad$.  Comparing
\eqref{eq:opt:cap-local-charge} with the hypothetical $Q$-charge at the
same state gives the exact identity
\begin{align}
 \ell_U-\ell_Q
 &=R(aI-A_I)+c(ad-A_D)
   +K\,[1-c(u-1-a)]
 \notag\\
 &\le K\,[1-c(u-1-a)].
 \label{eq:opt:cap-versus-q}
\end{align}
\end{proof}

\begin{claim}[Base potential for types $Z,E,M,Q$]
\label{clm:opt:adaptive-base-potential}
There is a nonnegative potential $W_{\mathrm{base}}$ for which the
directional-derivative bound
\[
 \ell_T+\nabla W_{\mathrm{base}}\mathbin{\cdot}v_T\le0
\]
holds for every $T\in\{Z,E,M,Q\}$.
The type $U$, representing $p=u$, is handled separately in
\cref{clm:opt:adaptive-cap-potential}.
\end{claim}

\begin{proof}
Define
\begin{equation}
 \eta=\frac{N_{\mathrm{def}}-RP}{x},
 \qquad b=\frac Ex,
 \qquad y=\eta-Rb.
 \label{eq:opt:eta-b}
\end{equation}
Call $-1\le y\le0$ Region I, and in this region put
\begin{align}
 f_c(y)
 &=\frac{c+2}{4}
   +\frac{c-2y}{4}
      \left(\ln\frac{c-2y}{c+2}-1\right),
 \label{eq:opt:f-definition}\\
 h_c(y)&=a_c(y)(1+y),
 \notag\\
 G_c(y,b)&=f_c(y)+b h_c(y)+\frac R2b^2.
 \label{eq:opt:G-region-one}
\end{align}
Call $y\le-1$ Region II, and in this region set instead
\begin{equation}
 G_c(y,b)=g_c(\eta),
 \qquad
 g_c(\eta)=\frac{(1+\eta)_+^2}{2R}.
 \label{eq:opt:G-region-two}
\end{equation}
The two definitions are $C^1$ across $y=-1$.  Indeed, feasibility
implies $0\le b\le1/R$ there, and both sides have value $Rb^2/2$ and
derivatives
\[
 (G_y,G_b)=(b,Rb).
\]
They are also nonnegative: $f_c(-1)=0$,
$f_c'=a_c-1\ge0$, and $h_c,b,g_c\ge0$.

Define the base potential
\begin{equation}
 W_{\mathrm{base}}(x,P,E,N_{\mathrm{def}})=xN_{\mathrm{def}}+x^2G_c(y,b).
 \label{eq:opt:base-potential}
\end{equation}
In the unnormalized state vector $(x,P,E,N_{\mathrm{def}},K)$, revealing one
endpoint of type $T$ changes the coordinates in the direction $v_T$ shown
below:
\begin{equation}
\begin{array}{c|ccccc}
 &x&P&E&N_{\mathrm{def}}&K\\ \hline
 v_Z&-1&0&0&0&0\\
 v_E&-1&0&1&0&0\\
 v_M&-1&1&0&0&0\\
 v_Q&-1&1&0&1&0\\
 v_U&-1&1&0&1&1.
\end{array}
 \label{eq:opt:state-directions}
\end{equation}
Let
$J_T=\ell_T+\nabla W_{\mathrm{base}}\mathbin{\cdot}v_T$.  This is the
local charge plus the directional derivative of the potential; we must show
$J_T\le0$.  For the moment $T\in\{Z,E,M,Q\}$.  In Region I, direct
differentiation and
the identities
\begin{align}
 -2f_c+(y-R)f_c'+a_c(1+y)&=0,
 \notag\\
 1-2f_c+(y-c)f_c'+a_cy&=0
 \label{eq:opt:f-hjb-identities}
\end{align}
give
\begin{align}
 \frac{J_Z}{x}
 &=(-2f_c+yf_c')+b(-h_c+yh_c'),
 \notag\\
 \frac{J_E}{x}
 &=b[-h_c+(y-R)h_c'+R],
 \notag\\
 \frac{J_M}{x}
 &=b[-h_c+(y-R)h_c'+a_cR],
 \notag\\
 \frac{J_Q}{x}
 &=b[-h_c+(y-c)h_c'+a_cR].
 \label{eq:opt:region-one-drifts}
\end{align}
All four quantities are nonpositive.  For the first one,
\begin{equation}
 -2f_c+yf_c'=a_c(c-y)-R\le0;
 \label{eq:opt:zero-drift-sign}
\end{equation}
the right-hand side is zero at $y=-1$ and is nonincreasing because
$a_c'(c-y)-a_c\le1-a_c\le0$.  Also $h_c,h_c'\ge0$.  For the remaining
three drifts, observe that
\begin{equation}
 h_c+(c-y)h_c'
 =a_cR+(c-y)(1+y)a_c'\ge a_cR.
 \label{eq:opt:h-sign-identity}
\end{equation}
The $Q$-bracket in \eqref{eq:opt:region-one-drifts} is therefore
nonpositive; the $M$-bracket is the $Q$-bracket minus $h_c'$, and the
$E$-bracket is the $M$-bracket minus $R(a_c-1)$.

In Region II, substituting $W=x^2g_c(\eta)$ into the definitions of the four
drifts gives
\begin{align}
 J_Z/x=J_E/x&=-2g_c+\eta g_c',
 \notag\\
 J_M/x&=-2g_c+(\eta-R)g_c'+1+\eta,
 \notag\\
 J_Q/x&=1-2g_c+(\eta-c)g_c'+\eta.
 \label{eq:opt:region-two-drifts}
\end{align}
For $-1\le\eta\le0$, substitution of
$g_c'=(1+\eta)/R$ makes these values $-(1+\eta)/R$,
$-(1+\eta)/R$, and $0$, respectively.  For $\eta\le-1$ their signs
are immediate from $g_c=g_c'=0$.  Hence $J_T\le0$ for
$T\in\{Z,E,M,Q\}$, as claimed.
\end{proof}

\begin{claim}[Reserve potential for jobs with $p=u$]
\label{clm:opt:adaptive-cap-potential}
There is a nonnegative reserve term $W_{\mathrm{cap}}$, depending on the
number $K$ of earlier $U$ endpoints, such that the full potential
$W=W_{\mathrm{base}}+W_{\mathrm{cap}}$ satisfies
\[
 \ell_T+\nabla W\mathbin{\cdot}v_T\le0
 \qquad(T\in\{Z,E,M,Q,U\}).
\]
\end{claim}

\begin{proof}
Define
\begin{equation}
 Z(c)=2+\frac1c+\frac12\ln\!\left(1+\frac2c\right),
 \qquad
 \delta=Z(c)-u.
 \label{eq:opt:Z-delta}
\end{equation}
For $\uDet{4}\le u\le\uDet{5}$, with $(c,m)$ chosen by
\cref{eq:opt:cm-parametrization}, the definition of $\delta$ gives
\begin{equation}
 \delta=\operatorname{artanh}m.
 \label{eq:opt:delta-atanh}
\end{equation}
For $0\le t\le m$, set
\begin{equation}
 \beta(t)=ct\,[\delta-\operatorname{artanh}t],
 \qquad
 \mathcal V(t)=\int_t^m\beta(s)\,ds,
 \label{eq:opt:reserve-integral}
\end{equation}
and extend $\mathcal V$ by zero for $t\ge m$.  Since
$\beta(m)=0$, this extension is $C^1$.  With
\begin{equation}
 q=\frac Kx,
 \qquad
 \mu=\frac{q}{1+q},
 \qquad
 \mathcal H(q)=(1+q)^2\mathcal V(\mu),
 \label{eq:opt:reserve-H}
\end{equation}
define the full potential
\begin{equation}
 W=xN_{\mathrm{def}}+x^2\bigl(G_c(y,b)+\mathcal H(q)\bigr).
 \label{eq:opt:full-potential}
\end{equation}
Thus $W_{\mathrm{cap}}=x^2\mathcal H(q)$.  Both terms in $W$ are
nonnegative.

For an endpoint of type $Z,E,M,$ or $Q$, the value of $K$ is fixed while
$x$ decreases, so the reserve term contributes, after division by
$x$,
\begin{equation}
 L_0(q)=-2\mathcal H(q)+q\mathcal H'(q).
 \label{eq:opt:L0}
\end{equation}
An endpoint $U$ also increments $K$, and its contribution is
$L_0+\mathcal H'$.  On the range $\mu\le m$, where $\mathcal V$ may be
nonzero, differentiation of
\eqref{eq:opt:reserve-H} gives
\begin{equation}
 \mathcal H'=2(1+q)\mathcal V-\beta(\mu),
 \qquad
 L_0=-2(1+q)\mathcal V-q\beta(\mu)\le0.
 \label{eq:opt:reserve-derivatives}
\end{equation}
Thus the reserve term can only decrease the left-hand side of the
directional-derivative bound for each of $Z,E,M,Q$.

For a current endpoint $U$, feasibility $N_{\mathrm{def}}\le P$ and
$P\ge K$ gives
\begin{equation}
 y=\frac{N_{\mathrm{def}}-R(P+E)}x\le-\frac{cK}{x}=-cq.
 \label{eq:opt:y-cap-bound}
\end{equation}
While $\mu\le m$, we have $q\le m/(1-m)\le\varphi$ and hence
$cq\le1$.  Monotonicity of $a_c$, together with
\eqref{eq:opt:cap-versus-q}, yields
\begin{align}
 \frac{\ell_U-\ell_Q}{x}
 &\le cq\left[\delta-\frac12\ln(1+2q)\right]
 =:B(q).
 \label{eq:opt:cap-residual-B}
\end{align}
The change of variables in \eqref{eq:opt:reserve-H} satisfies
\begin{equation}
 \operatorname{artanh}\mu=\frac12\ln(1+2q),
 \qquad
 (1+q)\beta(\mu)=B(q).
 \label{eq:opt:B-beta}
\end{equation}
Consequently the extra positive charge of an endpoint of type $U$ is
cancelled exactly by the directional derivative of the reserve term:
\begin{equation}
 L_0+\mathcal H'+B(q)=0
 \qquad(\mu\le m).
 \label{eq:opt:reserve-cancellation}
\end{equation}
When $\mu\ge m$, the reserve term vanishes.  At
$q_*=m/(1-m)$, the bracket in
\eqref{eq:opt:cap-versus-q}, evaluated at the largest possible threshold
$a_c(-cq_*)$, is zero by
\eqref{eq:opt:delta-atanh} and \eqref{eq:opt:B-beta}.  For $q\ge q_*$,
\eqref{eq:opt:y-cap-bound} only decreases the threshold, so that bracket is
nonpositive.  The directional-derivative bound for type $U$ therefore holds
for every state and for arbitrary interleavings of type $U$ with the other
four types.
\end{proof}

\begin{claim}[Initial calibration]
\label{clm:opt:adaptive-calibration}
The full potential starts with exactly the amount needed to cancel the
constant pairwise term: $W_{\mathrm{initial}}=cn^2/2$.
\end{claim}

\begin{proof}
Initially $y=b=q=0$, and
\begin{equation}
 f_c(0)=\frac12-\frac c4\ln\!\left(1+\frac2c\right).
 \label{eq:opt:f-zero}
\end{equation}
Elementary integration in \eqref{eq:opt:reserve-integral} gives
\begin{equation}
 \mathcal V(0)=\frac c2(\delta-m).
 \label{eq:opt:V-zero}
\end{equation}
Indeed,

\[
 \int_0^m t\operatorname{artanh}t\,dt
 =\frac{(m^2-1)\operatorname{artanh}m+m}{2}.
\]
Using \eqref{eq:opt:cm-parametrization} and
\eqref{eq:opt:delta-atanh}, we obtain the exact calibration
\begin{equation}
 f_c(0)+\mathcal H(0)
 =f_c(0)+\mathcal V(0)=\frac c2.
 \label{eq:opt:initial-calibration}
\end{equation}
Thus $W_{\mathrm{initial}}=cn^2/2$.
\end{proof}

\begin{claim}[Finite transitions and telescoping]
\label{clm:opt:adaptive-finite}
Let $z_i$ be the state before job $i$ is revealed, write $W_i=W(z_i)$, and
let $T_i$ be its endpoint type.  For $\uDet{4}\le u\le\uDet{5}$, each actual
transition satisfies
\[
 W_{i+1}-W_i
 =\nabla W(z_i)\mathbin{\cdot}v_{T_i}+O_u(1).
\]
Together with
\cref{clm:opt:adaptive-base-potential,clm:opt:adaptive-cap-potential}, this
gives
\begin{equation}
 \ALG\le(1+c)\cdot\OPT+\,O_u(n).
 \label{eq:opt:adaptive-finite-claim}
\end{equation}
For every finite $u\ge\uDet{5}$, there is an absolute constant $B$,
independent of $u$, such that
\begin{equation}
 \ALG\le(1+r)\cdot\OPT+Bn.
 \label{eq:opt:adaptive-uniform-claim}
\end{equation}
\end{claim}

\begin{proof}
We first verify the needed regularity on feasible states.
In Region I, $y\ge-1$ implies
\begin{equation}
 -xy=R(P+E)-N_{\mathrm{def}}=cP+(P-N_{\mathrm{def}})+RE\le x.
 \label{eq:opt:feasible-compactness}
\end{equation}
All terms in the middle expression are nonnegative.  Hence
$P/x,N_{\mathrm{def}}/x,q\le1/c$ and $b\le1/R$.  The normalized arguments of
$G_c$ in \eqref{eq:opt:base-potential} therefore stay in a bounded set, and
all derivatives needed for the Taylor estimate of one transition are
bounded.  In
Region II the unnormalized base potential is simply
\begin{equation}
 xN_{\mathrm{def}}+\frac{(x+N_{\mathrm{def}}-RP)_+^2}{2R},
 \label{eq:opt:region-two-base-potential}
\end{equation}
which has bounded one-sided second derivatives.  The calculation following
\cref{eq:opt:G-region-two} shows that the two pieces are $C^1$ at $y=-1$,
and the positive-part square is $C^1$ at
$\eta=-1$.

For the additional term $W_{\mathrm{cap}}(x,K)=x^2\mathcal H(K/x)$, whenever
$0\le q\le q_*$ one has bounded $\mathcal H''$, and
\begin{equation}
 (W_{\mathrm{cap}})_{xx}=2\mathcal H-2q\mathcal H'+q^2\mathcal H'',
 \quad
 (W_{\mathrm{cap}})_{xK}=\mathcal H'-q\mathcal H'',
 \quad
 (W_{\mathrm{cap}})_{KK}=\mathcal H''.
 \label{eq:opt:reserve-hessian}
\end{equation}
For $q\ge q_*$, $W_{\mathrm{cap}}=0$; moreover
$\mathcal H(q_*)=\mathcal H'(q_*)=0$.  Thus the join is $C^1$ with
bounded one-sided Hessians.  At $x=0$, set $W_{\mathrm{cap}}=0$: for $K>0$ it is
already locally zero, while
$W_{\mathrm{cap}}(\theta x,\theta K)=\theta^2W_{\mathrm{cap}}(x,K)$ gives zero gradient
at $(x,K)=(0,0)$.  Altogether $W$ is continuously differentiable with a Lipschitz
gradient on all states that the algorithm can reach.  For
$\uDet{4}\le u\le\uDet{5}$, its regularity constant may depend on the fixed
value of $u$.
When $c=r$ and $m=0$, however, the additional term is identically zero and the base
potential depends only on the absolute constant $r$; its regularity constant is
therefore uniform over all $u\ge \uDet{5}$.  This includes the last transition
from $x=1$ to $x=0$.

It follows that each transition caused by one job has Taylor remainder
\begin{equation}
 \ell_i+W_{i+1}-W_i\le O_u(1).
 \label{eq:opt:finite-taylor}
\end{equation}
The endpoint $E=0^+$ is interpreted as a one-sided limit.  For a fixed sequence of
endpoint types, replace all $E$ coordinates by positive values $\eps$.  Its
state trajectory is unchanged, and the exact finite pair objective, being
a finite sum of affine and minimum terms, converges as
$\eps\downarrow0$.  Equivalently, retain an
$O(\eps n^2)$ error and take this limit for each fixed $n$.
Therefore the directional-derivative bound for the formal type $E$ transfers to
actual inputs without an asymptotic loss.

Summing \eqref{eq:opt:finite-taylor}, using $W\ge0$, and applying
\eqref{eq:opt:initial-calibration} gives
\begin{equation}
 \sum_i\ell_i\le\frac c2n^2+O_u(n).
 \label{eq:opt:ell-telescope}
\end{equation}
Together with \eqref{eq:opt:gap-decomposition} and
\eqref{eq:opt:pair-allocation-dominates}, this yields
\begin{align}
 \ALG-\,R\cdot\OPT
 &\le -c\binom n2+\frac c2n^2+O_u(n)
 =O_u(n).
 \label{eq:opt:adaptive-middle-upper}
\end{align}

At $u=\uDet{5}$, one has $c=r$, $m=0$, and
$Z(r)=\uDet{5}$.  The reserve term for type $U$ is identically zero,
while
\eqref{eq:opt:initial-calibration} remains valid.  If $u$ is increased
above $\uDet{5}$, the bracket controlling an endpoint of type $U$,
$1-r(u-1-a_r(y))$ decreases, and the contribution of an individual job of
type $U$ in
\eqref{eq:opt:gap-decomposition} also decreases.

Here is the promised uniform finite calculation.  The base potential and the
set of states reachable by the algorithm contain no $u$, so there is one
constant $C_{\rm step}$, depending only on $r$, for which, whenever
$u\ge \uDet{5}$, every transition, whether its endpoint is $U$ or not,
satisfies
\[
 \ell_i+W_{i+1}-W_i\le C_{\rm step}.
\]
Because $W_0=rn^2/2$ and $W_n=0$, telescoping gives
\begin{equation}
 \sum_i\ell_i\le\frac r2n^2+C_{\rm step}n.
 \label{eq:opt:uniform-ell}
\end{equation}
In the final sum over individual jobs in \cref{eq:opt:gap-decomposition},
every endpoint other
than $U$ has $p_i<u-1$ and contributes $-r(1+p_i)$.  The endpoint $U$ has
$p_i=u$ and contributes at most $1-r u\le0$.  Hence this sum is
nonpositive uniformly on
$[\uDet{5},\infty)$.  Combining
\cref{eq:opt:gap-decomposition,eq:opt:pair-allocation-dominates,eq:opt:uniform-ell}
now yields
\begin{align}
 \ALG-(1+r)\cdot\OPT
 &\le-r\binom n2+\frac r2n^2+C_{\rm step}n
   =\left(C_{\rm step}+\frac r2\right)n.
 \label{eq:opt:uniform-high-gap}
\end{align}
Thus, with $B=C_{\rm step}+r/2$,
\begin{equation}
 \ALG\le(1+r)\cdot\OPT+B n
 \qquad\text{for every finite }u\in[\uDet{5},\infty),
 \label{eq:opt:uniform-high-upper}
\end{equation}
and $B$ is independent of $u$.  This is a uniform estimate, not the result
of sending a parameter-dependent remainder to infinity.
This proves the claim.
\end{proof}

We can now prove the theorem.

\begin{proof}[Proof of \Cref{thm:opt:optional-middle-upper}]
\Cref{clm:opt:adaptive-endpoint-reduction,clm:opt:adaptive-local-allocation}
reduce the competitive excess to the local charges.  The potentials in
\cref{clm:opt:adaptive-base-potential,clm:opt:adaptive-cap-potential} satisfy
the required directional-derivative bounds, and
\cref{clm:opt:adaptive-calibration} gives their initial value.
\Cref{clm:opt:adaptive-finite} then yields both displayed finite bounds.

Finally, $u\ge\uDet{4}>1$, so every effective offline length is at least one
and $\OPT\ge n(n+1)/2$.  Hence the additive terms vanish in the
size-asymptotic ratios, proving the final assertion.
\end{proof}

\subsection{Lower bounds}
\label{sec:det-branchwise-lower}

The lower bounds use three adversarial constructions.  A hidden binary
stopping game covers $0<u\le\uDet{4}$.  For
$\uDet{4}\le u\le\uDet{5}$, we prepend a block of jobs with $p_j=u$ to a
scaled harmonic construction.  Finally, for $u\ge\uDet{5}$, we use the
obligatory-testing harmonic instances with $1+p_j<u$ for every job.  On
these instances the offline cap $u$ is inactive, and the transfer lemma
preserves the obligatory-testing lower bound.

\subsubsection{The binary stopping game and its piecewise lower bound}
\label{sec:optional-lower}

Binary instances $p_j\in\{0,u\}$ give the lower bounds $1$ for
$0<u\le\uDet{1}$, $u$ for $\uDet{1}\le u\le\uDet{2}$,
$\Rquad(u)$ for $\uDet{2}\le u\le\uDet{3}$, and
$1+1/\sqrt{u-1}$ for $\uDet{3}\le u\le\uDet{4}$.  The stopping point
must be hidden from the algorithm;
announcing it in advance gives a strictly weaker lower bound.
In this subsubsection, a job with $p_j=0$ is a \emph{zero job}, and one with
$p_j=u$ is a \emph{long job}.

\begin{lemma}[Hidden binary stopping with raw execution]
\label{lem:opt:hidden-stopping}
Fix $u>1$, a target $R\le u$, and a stopping fraction $\alpha\in(0,1)$.
Against a deterministic algorithm, there is a legal binary adversary whose
leading normalized excess $2(\ALG-\,R\cdot\OPT)/n^2$ at its stopping time is,
up to $O_u(1/n)$,
\begin{equation}
 (u-R)(1-\sigma^2)+\sigma^2 f_{u,R}(y,b),
 \label{eq:opt:stopping-decomposition}
\end{equation}
where $0\le b\le y\le1$ and
\begin{equation}
 0\le y-\alpha<\frac1{n-v}.
 \label{eq:opt:stopping-overshoot}
\end{equation}
Moreover,
\begin{equation}
 f_{u,R}(y,b)=
 1-R+2y+(u-1)(1-R)y^2
 +b^2+2b(u-1-uy).
 \label{eq:opt:stopping-f}
\end{equation}
Here $\sigma\in[0,1]$ is the fraction of jobs not already executed raw,
$y$ is the fraction of those jobs revealed to be long, and $b$ is the
fraction both revealed long and already completed.
Consequently, if
\[
 \min_{0\le b\le\alpha}f_{u,R}(\alpha,b)\ge0,
\]
then the adversary gives ratio at least $R-O_u(1/n)$.
\end{lemma}

\begin{proof}
Until the stopping time, answer every test of a fresh job with $p=u$; if the
algorithm instead runs a fresh job raw, assign that job the hidden value
zero.  Let $v$ be the number of raw first touches and $L$ the number of tests,
all of which have revealed long jobs.  Stop immediately after the first touch
for which
\begin{equation}
 L\ge\alpha(n-v).
 \label{eq:opt:stopping-line}
\end{equation}
All untouched jobs are then assigned value zero.  Only for the purpose of
lower-bounding the remaining cost, we grant the algorithm this future
information and allow an optimal continuation; in the actual interaction,
values are still revealed only by tests.  If the line is never crossed, every job
was run raw: if even one job had been tested, then after the last first touch
the tested fraction would be one.  On the all-zero hidden instance the ratio
is $u\ge R$.

At a crossing, let $e$ of the $L$ revealed long jobs already be complete,
let $d=L-e$ be tested but still pending, and let $x=n-v-e-d$ jobs remain
untouched.  To obtain a lower bound on the algorithm's cost, we give it the
benefit of rearranging its work into the following cheaper schedule: first
the $v$ raw completions, then the tests and processing of the $e$ completed
long jobs, then the $d$ tests whose long jobs remain pending, then the $x$
zero test-completions, and finally the processing of the $d$ pending long
jobs.

Here is why the rearrangement can only decrease total completion time.  If a
noncompleting unit test is immediately
before an available length-$u$ completion while $h$ jobs remain, the two
orders contribute $h+hu$ and $hu+(h-1)$, respectively; moving the
completion left saves one.  Raw completions are available from time zero,
and identical long jobs may be relabeled so that the $e$ completed ones are
tested first.  After the switch, a zero test-completion has length one and
therefore precedes a pending operation of length $u$.  Repeating only these
explicit neighboring swaps produces the displayed relaxed schedule and
respects every test-before-processing precedence constraint.

The exact costs of this relaxed schedule and of OPT are
\begin{align}
 A_{\rm ex}={}&
 u\left(vn-\frac{v(v-1)}2\right)
 +(1+u)\left(e(n-v)-\frac{e(e-1)}2\right) \notag\\
 &+d(n-v-e)+x(d+x)-\frac{x(x-1)}2
 +u\frac{d(d+1)}2,
 \label{eq:opt:stopping-Aexact}\\
 O_{\rm ex}={}&\frac{n(n+1)}2
 +(u-1)\frac{(e+d)(e+d+1)}2.
 \label{eq:opt:stopping-Oexact}
\end{align}
Put $\nu=v/n$, $\delta=(v+e+d)/n$, and $\lambda=e/n$.  Twice the quadratic
coefficients in \cref{eq:opt:stopping-Aexact,eq:opt:stopping-Oexact} are
\begin{align*}
 \mathcal A={}&1+2\delta(1-u\nu)+(u-1)\delta^2
 +2\nu(\nu+u-2)\\
 &+\lambda^2+2\lambda(\nu+u-1-u\delta),\\
 \mathcal O={}&1+(u-1)(\delta-\nu)^2.
\end{align*}
Now put $\sigma=1-\nu$, $y=(e+d)/(n-v)$, and $b=e/(n-v)$.
Thus $\sigma$ is the nonraw fraction of the original instance, while $y$ and
$b$ are the long and already-completed fractions within those nonraw jobs.
Direct expansion gives exactly \cref{eq:opt:stopping-decomposition} after
normalizing by $n^2/2$.  The first-crossing overshoot in
\cref{eq:opt:stopping-line} is less than one: a test increases
$e+d-\alpha(n-v)$ by one and a raw first touch increases it by $\alpha$.
This proves \cref{eq:opt:stopping-overshoot}, and in particular
$\sigma(y-\alpha)<1/n$.  The minimization in the lemma uses
$b\le\alpha$, whereas the actual value only satisfies $b\le y$ and may exceed
$\alpha$ by the one-job overshoot.  Put $\bar b=\min\{b,\alpha\}$.  Then
$0\le b-\bar b\le y-\alpha$.  Since $f_{u,R}$ is Lipschitz on the unit
square,
\begin{equation}
 \sigma^2 f_{u,R}(y,b)
 \ge \sigma^2 f_{u,R}(\alpha,\bar b)-O_u(1/n).
\end{equation}
Thus the stated minimum over $0\le\bar b\le\alpha$ suffices uniformly even
when $n-v$ is small.  The omitted one-job terms are $O_u(n)$ in the
unnormalized costs.
Finally, freeze every revealed or hidden answer on its job label.  Because
the algorithm is deterministic and a raw action never reveals its assigned
zero, rerunning it on this fixed input reproduces the same transcript.  The
adaptive description therefore yields a legal fixed hard instance.
\end{proof}

\begin{proposition}[Binary lower curve]
\label{prop:opt:binary-lower}
Every deterministic algorithm has size-asymptotic ratio at least
\begin{equation}
 B(u)=
 \begin{cases}
  1,&0<u\le1,\\
  u,&\uDet{1}\le u\le \uDet{2},\\
  \Rquad(u),&\uDet{2}\le u\le \uDet{3},\\[1mm]
  \displaystyle1+\frac1{\sqrt{u-1}},&u\ge \uDet{3}.
 \end{cases}
 \label{eq:opt:binary-curve}
\end{equation}
The lower bound uses only $p_j\in\{0,u\}$.
\end{proposition}

We optimize the stopping lemma in three steps.  First we minimize over the
fraction $b$ of revealed long jobs that the algorithm has already completed.
This produces two quadratic formulas, according to whether the minimizing
$b$ is positive or zero.  We then locate their transition points and choose
the stopping fraction $\alpha$ at which the relevant quadratic is largest.

\begin{claim}[Minimization over completed long jobs]
\label{clm:opt:binary-b-minimization}
For fixed revealed-long fraction $y$, minimizing
\cref{eq:opt:stopping-f} over the completed-long fraction $b$ gives
\begin{equation}
 b_*(y)=\max\{0,uy-(u-1)\}.
 \label{eq:opt:bstar}
\end{equation}
The upper constraint $b\le y$ never binds, because
$uy-(u-1)\le y$ for $y\le1$.
When $uy-(u-1)>0$, the minimum equals
\begin{equation}
 \min_b f_{u,R}(y,b)
 =1-R-(u-1)^2+2(u^2-u+1)y
  -\bigl(u^2-u+1+(u-1)R\bigr)y^2.
 \label{eq:opt:stopping-positive-branch}
\end{equation}
and its unconstrained maximum over $y$ occurs at
\[
 \alpha=\frac{u^2-u+1}{u^2-u+1+(u-1)R}.
\]
This maximum is zero exactly when
\begin{equation}
 \bigl(1-R-(u-1)^2\bigr)
 \bigl(u^2-u+1+(u-1)R\bigr)+(u^2-u+1)^2=0.
 \label{eq:opt:stopping-tangency}
\end{equation}
which is equivalent to \cref{eq:opt:Rquad-equation}.  When
$uy-(u-1)\le0$, the minimum equals
\begin{equation}
 \min_b f_{u,R}(y,b)
 =1-R+2y+(u-1)(1-R)y^2.
\end{equation}
and its unconstrained maximum is zero for
\begin{equation}
 R=1+\frac1{\sqrt{u-1}},
 \qquad
 \alpha=\frac1{\sqrt{u-1}}.
 \label{eq:opt:zero-b-tangency}
\end{equation}
This maximizing point indeed lies in the region where $b=0$ if and only if
$(u-1)^3\ge u^2$, with equality at $u=\uDet{3}$.
\end{claim}

\begin{proof}
The part of $f_{u,R}(y,b)$ that depends on $b$ is
$b^2+2b(u-1-uy)$, whose unconstrained minimizer is
$uy-(u-1)$.  Taking the maximum with zero gives
\cref{eq:opt:bstar}; its upper bound follows from
$uy-(u-1)\le y$.  Substitution gives the two displayed quadratics.
Differentiating each concave quadratic in $y$ gives the stated maximizing
value of $\alpha$.  Substitution of that value gives
\cref{eq:opt:stopping-tangency} in the positive-$b$ case and
\cref{eq:opt:zero-b-tangency} in the zero-$b$ case.  Expanding the former
gives \cref{eq:opt:Rquad-equation}; substituting the latter value of
$\alpha$ into $u\alpha-(u-1)\le0$ gives $(u-1)^3\ge u^2$.
\end{proof}

\begin{claim}[The transition from the ratio $u$]
\label{clm:opt:binary-first-transition}
For $R=u$, the maximum in
\cref{eq:opt:stopping-positive-branch} is nonnegative exactly while
\begin{equation}
 [u(u-1)]^2-u(u-1)-1\le0,
\end{equation}
whose positive equality is $u=\uDet{2}$.
\end{claim}

\begin{proof}
Substitute $R=u$ and the maximizing value of $\alpha$ from
\cref{clm:opt:binary-b-minimization} into
\cref{eq:opt:stopping-positive-branch}.  After multiplication by its
positive denominator, nonnegativity is equivalent to the displayed
inequality.  With $x=u(u-1)$, its equality is $x^2-x-1=0$, so
$x=\varphi$.  Since $u(u-1)$ is strictly increasing for $u>1$, the unique
positive solution is
\[
 u=\frac{1+\sqrt{1+4\varphi}}2
  =\frac{1+\sqrt{3+2\sqrt5}}2=\uDet{2}.
\]
\end{proof}

\begin{claim}[Stopping fractions for the three nontrivial intervals]
\label{clm:opt:binary-stopping-choices}
For $\uDet{1}<u\le \uDet{2}$, use $R=u$ and the maximizing point of
\cref{eq:opt:stopping-positive-branch}.  For
$\uDet{2}\le u\le \uDet{3}$, use the positive root $\Rquad(u)$ and
$\alpha=(u^2-u+1)/(u^2-u+1+(u-1)\Rquad(u))$; the zero-maximum identity makes the
maximum zero.  The minimizing completed-long fraction $b$ is positive in the
interior and becomes zero
at $u=\uDet{3}$.  For $u\ge \uDet{3}$, use the zero-$b$ maximizing point from
\cref{eq:opt:zero-b-tangency}.  Each choice satisfies
\[
 \min_{0\le b\le\alpha} f_{u,R}(\alpha,b)\ge0.
\]
\end{claim}

\begin{proof}
For $\uDet{1}<u\le\uDet{2}$, the assertion follows from
\cref{clm:opt:binary-first-transition}.  For
$\uDet{2}\le u\le\uDet{3}$, the defining equation for $\Rquad(u)$ is the
zero-maximum identity in
\cref{clm:opt:binary-b-minimization}; the minimizing completed-long fraction
is positive in the interior and becomes zero at $u=\uDet{3}$.  For
$u\ge\uDet{3}$, the zero-$b$ formula and its feasibility are exactly
\cref{eq:opt:zero-b-tangency} and the final assertion of
\cref{clm:opt:binary-b-minimization}.
\end{proof}

We can now prove the proposition.

\begin{proof}[Proof of \Cref{prop:opt:binary-lower}]
For $u\le1$, the lower bound $\RdetRO(u)\ge1$ is immediate.  For $u>1$,
use the target ratio and stopping fraction from
\cref{clm:opt:binary-stopping-choices} in
\cref{lem:opt:hidden-stopping}, and let the number of jobs tend to infinity.
The recording step in that lemma converts every adaptive interaction into a
fixed binary instance.  This proves all four formulas in
\cref{eq:opt:binary-curve}.
\end{proof}

\subsubsection{The harmonic construction for short jobs}
\label{sec:lower}

This subsection develops the harmonic construction used for the lower bounds
on $\uDet{4}\le u\le\uDet{5}$ and on $u\ge\uDet{5}$.  Its largest
processing time will be below two.  Consequently, whenever the construction
is later used with a
raw-execution cap $u>3$, every job satisfies $1+p_j<u$, so the offline cap
$u$ is inactive.
Its notation is local to this subsection.
We first record its extremal specialization and then extract the scaled form
used when $\uDet{4}\le u\le\uDet{5}$.  The construction
has many equally large groups of positive jobs, revealed from longest to
shortest, followed by one large group of zero jobs.  The only design choice is
the gap between two consecutive positive processing times.  Choosing these
gaps harmonically will make all decisions of the online algorithm disappear
from the leading term.

\begin{theorem}[Extremal harmonic core]
\label{thm:lower}
For every deterministic online algorithm $\mathcal B$ and every $\eps>0$,
there are arbitrarily large fixed instances $I$ with unit tests and rational
processing times such that
\[
  \ALG_{\mathcal B}(I)\ge (\RdetOT-\eps)\cdot\OPT(I).
\]
\end{theorem}

We prove the theorem through five claims.  They establish, in order, that the
adaptive reveal order defines a fixed input, that a shortest-first relaxation
lower-bounds every phase, that this relaxation gives an explicit finite
accounting formula, that harmonic processing-time gaps remove all
algorithm-dependent terms, and that the resulting one-variable limit is
maximized at $\RdetOT$.

\paragraph{Common construction.}
Fix an integer $K$, a rational number $\xi>0$, and a rational parameter
$\gamma\in(0,1)$; one may keep $\gamma=1/2$ throughout.  For a large scaling
integer $H$, the instance will contain $H$ jobs of each of $K$ positive types
$0,1,\ldots,K-1$ and $\xi H$ zero jobs.  We choose $H$ to clear all
denominators.

The adversary assigns processing times according to the order in which tests
are started.  The first $H$ tested jobs receive processing time $p_0$, the
next $H$ receive $p_1$, and so on; after the $K$ positive blocks, every
remaining job receives processing time zero.  Thus
\[
  p_0>p_1>\cdots>p_{K-1}>1,
\]
so the algorithm sees the positive types from longest to shortest and learns
about the zero jobs only at the end.

\begin{claim}[The reveal order defines a fixed input]
\label{clm:lower-fixed-input}
For every deterministic algorithm, assigning each answer produced by this
reveal order to the label of the tested job gives a fixed input on which the
algorithm repeats the same tests, answers, and completion decisions.
\end{claim}

\begin{proof}
This adaptive description is only a convenient way to construct the input.
After running the interaction against a fixed deterministic algorithm, assign
to every job label the value revealed when that label was tested.  Rerunning
the algorithm on this fixed assignment produces exactly the same transcript.
If the algorithm never tests or completes some job, its cost is infinite and
the lower bound is immediate; otherwise every job label is recorded.
Hence it is enough to lower-bound the cost under this adaptive reveal order.
\end{proof}

\begin{claim}[A shortest-first lower bound in each phase]
\label{clm:lower-phase-relaxation}
Let $T_\ell$ be the start of the first test in positive block $\ell$, and let
$T_K$ be the start of the first test in the zero block.  Phase $\ell$ contains
the completions in $(T_\ell,T_{\ell+1}]$; a completion exactly at $T_\ell$
belongs to the preceding phase.  In this construction, the mass of a group
is its number of jobs divided by $H$.  For group masses $q_1,\ldots,q_m$ and
remaining lengths $v_1\le\cdots\le v_m$, define
\begin{equation}
 \mathcal S((q_g,v_g)_{g=1}^m)
 =\sum_{g=1}^m
   \left(\frac{v_gq_g^2}{2}
    +q_g\sum_{h<g}v_hq_h\right).
 \label{eq:lower-S}
\end{equation}
If groups of size $q_gH$ finish within one phase, their relative completion
cost is at least $H^2\mathcal S+O(H)$.  If
\begin{equation}
  p_0-p_{K-1}<1.
  \label{eq:lower-span}
\end{equation}
then the nondecreasing remaining lengths in a positive phase are
\[
  p_{\ell-1}<p_{\ell-2}<\cdots<p_0<1+p_\ell.
\]
In the final phase, the order starts with the zero jobs of remaining length
one and continues with the positive types from shortest to longest.
\end{claim}

\begin{proof}
Each $T_\ell$ is the boundary between two nonpreemptive operations.
Consequently, a previously tested job that is unfinished at $T_\ell$ has
not started processing and still needs its entire processing time.

Suppose jobs with remaining work $v_1,\ldots,v_m$ complete after some time
$T$.  If their completion order is $\pi_1,\ldots,\pi_m$, then the $g$th relative
completion time is at least $\sum_{h\le g}v_{\pi_h}$.  Summing these prefix
inequalities shows that their total relative completion time is minimized by
arranging the $v_i$ in nondecreasing order.  For groups of size $q_gH$, the
within-group sums give $v_gq_g^2H^2/2+O(H)$, and the waiting caused by shorter
groups gives the cross term in \cref{eq:lower-S}.

In phase $\ell$, an unfinished job of an older type $j<\ell$ has remaining
length $p_j$, whereas a type-$\ell$ job completed in its own phase needs its
test and processing, of total length $1+p_\ell$.  The inequalities in the
claim follow from $p_0-p_{K-1}<1$.  At $T_K$, each untested zero job has
remaining length one, which proves the asserted final-phase order.
\end{proof}

\begin{claim}[Finite phase accounting]
\label{clm:lower-finite-accounting}
For each positive type $j$, let
\begin{itemize}
\item $s_j\in[0,1]$ be the mass of type-$j$ jobs completed in their own
      phase;
\item $d_{j,\ell}\ge0$ be the mass of type-$j$ jobs completed in the later
      phase $\ell>j$; and
\item
\[
  X_{j,\ell}=s_j+\sum_{h=j+1}^{\ell-1}d_{j,h}\le1
  \qquad(\ell>j)
\]
be the mass of type $j$ completed before phase $\ell$.  Thus
$X_{j,\ell+1}$ is the mass completed by the end of phase $\ell$.
\end{itemize}
Then the start of phase $\ell$ satisfies
\begin{equation}
  \frac{T_\ell}{H}
  \ge \ell+\sum_{j<\ell}p_jX_{j,\ell}.
  \label{eq:lower-start}
\end{equation}
For each phase, charge its completions their common start time from
\cref{eq:lower-start}, and charge their relative completion times using
\cref{eq:lower-S}.  Do the same at $T_K$ for the remaining positive jobs
and the $\xi H$ zero jobs.  Expanding and collecting equal variables gives
\begin{align}
 \frac{\ALG}{H^2}
 \ge{}& C_K
 +\sum_{j=0}^{K-1}\left(\frac{s_j^2}{2}-\theta_js_j\right)
 \nonumber\\
 &+\sum_{j<\ell}d_{j,\ell}
 \left[
   \ell-j-\theta_j
   -\sum_{j<k<\ell}(p_j-p_k)X_{k,\ell+1}
 \right]-O(1/H),
 \label{eq:lower-accounting}
\end{align}
where $K,\xi,\gamma$ are fixed while $H\to\infty$, and
\begin{align}
 \theta_j
   &=1-\xi(p_j-1)-\sum_{k>j}(p_j-p_k-1),
   \label{eq:lower-theta}\\
 C_K
   &=\frac{\xi^2}{2}+2\xi K+K^2+P_K,
 &
 P_K&=\sum_{j=0}^{K-1}\left(j+\frac12\right)p_j.
 \label{eq:lower-CP}
\end{align}
\end{claim}

\begin{proof}
At time $T_\ell$, the machine has already performed the tests of the $\ell H$
jobs in earlier positive blocks and has processed every earlier job that has
completed.  This gives \cref{eq:lower-start}.

For each phase, charge its completions their common start time from
\cref{eq:lower-start}, and charge their relative completion times using
\cref{clm:lower-phase-relaxation}.  In phase $\ell$, the masses in the
shortest-first relaxation are
\[
 d_{\ell-1,\ell},d_{\ell-2,\ell},\ldots,d_{0,\ell},s_\ell
\]
with respective lengths
\[
 p_{\ell-1},p_{\ell-2},\ldots,p_0,1+p_\ell.
\]
Multiply their total mass by the right-hand side of
\cref{eq:lower-start} and add the functional
\cref{eq:lower-S}.  In the final phase use masses
\[
 \xi,\ 1-X_{K-1,K},\ldots,1-X_{0,K}
\]
and lengths $1,p_{K-1},\ldots,p_0$.  The constant terms sum to $C_K$.
The terms containing only $s_j$ are $s_j^2/2-\theta_js_j$; the coefficient
of each $d_{j,\ell}$ is precisely the bracket in
\cref{eq:lower-accounting}.  Thus the displayed identity is just the
phasewise shortest-first bound with the phase start times retained.  No
optimality assumption about the online schedule is hidden in it.
\end{proof}

The meaning of the expression is useful.  Completing an $s_j$ fraction of a
type in its own phase creates a quadratic completion cost but saves later
waiting cost, producing $s_j^2/2-\theta_js_j$.  A completion postponed to a
later positive phase is represented by $d_{j,\ell}$ and its bracket.  We next
choose the processing times so that every such bracket is nonnegative and all
$\theta_j$ are equal.

\begin{claim}[Harmonic gaps remove the algorithm-dependent terms]
\label{clm:lower-harmonic-cancellation}
Assume \cref{eq:lower-span} and choose the processing times by
\begin{equation}
  p_{K-1}=1+\frac{\gamma}{\xi},
  \qquad
  p_i=p_{i+1}+\frac1{\xi+K-1-i}
  \quad(0\le i<K-1).
  \label{eq:lower-harmonic}
\end{equation}
Then every $\theta_j$ in \cref{eq:lower-theta} equals
\begin{equation}
 \theta_j=1-\gamma,
 \label{eq:lower-theta-constant}
\end{equation}
every coefficient of $d_{j,\ell}$ in \cref{eq:lower-accounting} is
nonnegative, and
\begin{equation}
 \liminf_{H\to\infty}\frac{\ALG}{H^2}
 \ge Q_K
 \defeq C_K-\frac K2(1-\gamma)^2.
 \label{eq:lower-QK}
\end{equation}
The corresponding offline coefficient is
\begin{equation}
 D_K=\frac{\xi^2}{2}+\xi K+\frac{K^2}{2}+P_K,
 \qquad
 \lim_{H\to\infty}\frac{\OPT}{H^2}=D_K.
 \label{eq:lower-DK}
\end{equation}
Thus $Q_K/D_K$ is a lower bound for the limiting ratio of the finite-$K$
construction.
\end{claim}

\begin{proof}
Measured upward from the shortest type, the consecutive gaps in
\cref{eq:lower-harmonic} are
\[
 \frac1{\xi+1},\frac1{\xi+2},\ldots,\frac1{\xi+K-1}.
\]
If $m=K-1-j$, then
\begin{align*}
 p_j-1
   &=\frac\gamma \xi+\sum_{h=1}^m\frac1{\xi+h},\\
 \sum_{k>j}(p_j-p_k)
   &=\sum_{h=1}^m\frac{h}{\xi+h}.
\end{align*}
Adding the first line multiplied by $\xi$ to the second line makes each
denominator cancel:
\begin{equation}
 \xi(p_j-1)+\sum_{k>j}(p_j-p_k)
 =\gamma+\sum_{h=1}^m\frac{\xi+h}{\xi+h}
 =\gamma+m.
 \label{eq:lower-telescope}
\end{equation}
Substitution in \cref{eq:lower-theta} gives
\cref{eq:lower-theta-constant}.  This cancellation explains the harmonic
increments: the factor $\xi$ from the zero block and the multiplicity $h$
from the $h$ shorter positive types add up to the denominator $\xi+h$.

The delayed-completion terms are now harmless as well.  Since
$X_{k,\ell+1}\le1$, their brackets are at least
\begin{align*}
 &\ell-j-(1-\gamma)
   -\sum_{j<k<\ell}(p_j-p_k)\\
 &\hspace{20mm}
 =\gamma+\sum_{j<k<\ell}\bigl(1-(p_j-p_k)\bigr)>0,
\end{align*}
where the final inequality follows from
\cref{eq:lower-span}.  We may therefore drop every term containing
$d_{j,\ell}$.  The remaining minimization is not a multivariable optimization
problem: complete one square for each $j$,
\[
  \frac{s_j^2}{2}-(1-\gamma)s_j
  =\frac12\bigl(s_j-(1-\gamma)\bigr)^2
   -\frac12(1-\gamma)^2.
\]
Because $1-\gamma\in(0,1)$, this minimum is feasible for
$s_j\in[0,1]$, which proves \cref{eq:lower-QK}.

Knowing all processing times,
the optimum runs the zero jobs first and the positive types from shortest to
longest.  Equivalently, the pairwise contribution of two jobs is
$1+\min\{p_i,p_j\}$.
The four terms respectively account for zero--zero pairs, zero--positive
pairs, the unit-test part of positive--positive pairs, and the processing
part of positive--positive pairs.  This gives \cref{eq:lower-DK} and
completes the proof.
\end{proof}

\begin{claim}[Limit and one-variable maximization]
\label{clm:lower-smooth-limit}
Fix $0<\alpha<e-1$ and set $\xi=K/\alpha$.  As $K\to\infty$, the
processing times in \cref{eq:lower-harmonic} satisfy
\cref{eq:lower-span}, and
\begin{equation}
 \lim_{K\to\infty}\frac{Q_K}{D_K}
 =R(\alpha)
 \defeq 1+
 \frac{\alpha+\alpha^2/2}
 {\frac12+\frac\alpha2+\frac{\alpha^2}{4}
  +\frac12(1+\alpha)^2\ln(1+\alpha)}.
 \label{eq:lower-Ralpha}
\end{equation}
The function $R(\alpha)$ has a unique maximum at
\begin{equation}
 \alpha_\star=\sqrt{\uDet{5}}-1<e-1,
 \qquad R(\alpha_\star)=\RdetOT.
 \label{eq:lower-alpha-star}
\end{equation}
\end{claim}

\begin{proof}
The total span is
\[
 p_0-p_{K-1}
 =\sum_{h=1}^{K-1}\frac1{\xi+h}
 <\ln\!\left(1+\frac K\xi\right)
 =\ln(1+\alpha)<1,
\]
which verifies \cref{eq:lower-span}.

The processing contribution $P_K$ has the exact form
\begin{equation}
 P_K
 =\frac{K^2}{2}\left(1+\frac\gamma \xi\right)
  +\frac12\sum_{h=1}^{K-1}\frac{(K-h)^2}{\xi+h}.
 \label{eq:lower-P-exact}
\end{equation}
Indeed, the common baseline $1+\gamma/\xi$ contributes the first term, and
swapping the order of summation shows that the increment $1/(\xi+h)$ receives
total weight $(K-h)^2/2$.  After division by $\xi^2$, the sum becomes a
Riemann integral:
\begin{align}
 \frac{P_K}{\xi^2}&\longrightarrow I(\alpha)
   =\frac{\alpha^2}{2}
    +\frac12\int_0^\alpha\frac{(\alpha-t)^2}{1+t}\,dt
   \nonumber\\
 &=\int_0^\alpha(\alpha-t)\bigl(1+\ln(1+t)\bigr)\,dt
   \nonumber\\
 &=\frac12(1+\alpha)^2\ln(1+\alpha)
   -\frac\alpha2-\frac{\alpha^2}{4}.
 \label{eq:lower-I}
\end{align}
The second integral form follows either by integration by parts or by summing
the limiting processing-time curve $1+\ln(1+t)$ against the remaining
positive mass $\alpha-t$.
The subtraction in \cref{eq:lower-QK} is only $O(K)$, whereas
$C_K,D_K=\Theta(K^2)$.  Since $C_K-D_K=\xi K+K^2/2$, we obtain
\cref{eq:lower-Ralpha}.

Put $z=(1+\alpha)^2$.  In terms of $z$, the limiting ratio is
\begin{equation}
  \widehat R(z)=1+\frac{2(z-1)}{1+z+z\ln z}.
  \label{eq:lower-Rz}
\end{equation}
The derivative of the fraction has the sign of
\[
  3-z+\ln z.
\]
For $z>1$ this expression is strictly decreasing, is positive at $z=1$, and
tends to $-\infty$.  Hence it has one zero $\uDet{5}$, characterized by
\[
  \uDet{5}-3=\ln \uDet{5}.
\]
The derivative is positive before $\uDet{5}$ and negative afterward, so this is
the global maximum.  At the root,
\[
  1+\uDet{5}+\uDet{5}\ln \uDet{5}=(\uDet{5}-1)^2,
\]
and therefore
\[
  \widehat R(\uDet{5})=1+\frac{2}{\uDet{5}-1}=\RdetOT.
\]
Moreover, $\sqrt{\uDet{5}}-1=1.12\ldots<e-1$, so the condition
$\alpha<e-1$ holds at the maximizer.  This proves
\cref{eq:lower-alpha-star}.
\end{proof}

We can now prove the theorem.

\begin{proof}[Proof of \Cref{thm:lower}]
By \cref{clm:lower-smooth-limit}, choose a rational $\alpha$ sufficiently close to
$\alpha_\star=\sqrt{\uDet{5}}-1$, and choose any rational
$\gamma\in(0,1)$.  For a sufficiently large $K$, set $\xi=K/\alpha$ and use
\cref{eq:lower-harmonic}; then $Q_K/D_K>\RdetOT-\eps/2$.
All parameters are rational.  Let $H$ tend to infinity through multiples of
their common denominators.
\Cref{clm:lower-finite-accounting,clm:lower-harmonic-cancellation} make the actual ratio exceed
$\RdetOT-\eps$ for all sufficiently large $H$.  Finally, freeze
the adaptively revealed answers on their job labels as in
\cref{clm:lower-fixed-input}.  The
resulting fixed input produces the same transcript for the given
deterministic algorithm.
\end{proof}

The lower bound for $\uDet{4}\le u\le\uDet{5}$ needs the same construction
at an arbitrary total scale and at a prescribed point of its limiting curve.
The next lemma packages this
rescaled form so that no phase accounting has to be repeated later.

\begin{lemma}[Scaled harmonic core]
\label{lem:opt:harmonic-block}
Fix $s>0$ and a parameter $t\in[1,e)$.  Against
every deterministic obligatory-testing algorithm, there are scales
$M_\nu\to\infty$ and finite rational multilevel instances with
$sM_\nu+o(M_\nu)$ jobs such that
\begin{equation}
 \frac{\OPT}{M_\nu^2}\longrightarrow
 O_0=\frac{s^2}{4}(1+t^{-2}+2\ln t),
 \qquad
 \liminf_{\nu\to\infty}\frac{\ALG}{M_\nu^2}
 \ge A_0^{\rm lb}
 \defeq O_0+\frac{s^2}{2}(1-t^{-2}).
 \label{eq:opt:harmonic-coefficients}
\end{equation}
Their total positive processing time, divided by $M_\nu$, converges to
\begin{equation}
 \frac1{M_\nu}\sum_{j:p_j>0}p_j\longrightarrow L=s\ln t,
 \label{eq:opt:harmonic-processing-mass}
\end{equation}
and their largest processing time converges to $1+\ln t<2$.
Freezing each revealed answer on its job label, as in
\cref{clm:lower-fixed-input}, converts the adaptive description into a fixed
instance.
\end{lemma}

\begin{proof}
The case $t=1$ is obtained by continuity from instances with vanishing
positive mass, so suppose $1<t<e$.  In the construction from the proof of
\cref{thm:lower}, set
$\alpha=t-1=K/\xi$.  Before rescaling, the number of jobs divided by $H$ is
$K+\xi=\xi t$.  Multiply every normalized group size by the common factor
$s/(\xi t)$.

If every normalized group size is multiplied by a common factor, each term
in the phase accounting, harmonic cancellation, and completion of squares in
\cref{eq:lower-accounting,eq:lower-telescope,eq:lower-QK} is multiplied by
the square of that factor.  The Riemann-sum limit
\cref{eq:lower-I} substituted into \cref{eq:lower-DK} therefore gives
\[
 O_0=\frac{s^2}{4}(1+t^{-2}+2\ln t).
\]
Likewise $C_K-D_K=\xi K+K^2/2$, while the $O(K)$ loss from completing the square in
\cref{eq:lower-QK} disappears after division by the quadratic scale.  After
the same rescaling this gives the guaranteed lower benchmark
\[
 A_0^{\rm lb}-O_0=\frac{s^2}{2}(1-t^{-2}).
\]
Finally, summing the limiting processing-time curve
$1+\ln(1+x)$ over the positive mass gives $L=s\ln t$, and
\cref{eq:lower-harmonic} gives $p_0\to1+\ln t$.  Choose successively finer
rational approximations, clear their denominators, and apply the recording
argument from \cref{clm:lower-fixed-input}: assign each revealed answer to
its job label and rerun the deterministic algorithm on the resulting fixed
input.  This produces arbitrarily large fixed rational instances with the
stated limit and liminf bounds.
\end{proof}

\subsubsection{Combining jobs with \texorpdfstring{$p=u$}{p = u} and a harmonic block}
\label{sec:opt:middle-lower}

For $\uDet{4}<u<\uDet{5}$, neither the binary family nor the pure harmonic
family is tight.  The adversary first reveals a prescribed fraction of jobs
with $p=u$.  It then reveals the scaled harmonic instance from
\cref{lem:opt:harmonic-block}, chosen so that every later job satisfies
$1+p\le u$.  Two lemmas
show that neither raw execution nor processing one of the initial jobs early
can avoid the resulting lower bound.  The later part is exactly the
construction from \cref{lem:opt:harmonic-block}, so its accounting need not
be repeated here.

\begin{lemma}[Deferring jobs with $p=u$]
\label{lem:opt:cap-deferral}
Prepend a fraction $m$ of jobs with $p=u$ to a later block occupying the
remaining fraction $s=1-m$ and having total normalized processing time $L$.
Suppose every job in the later block satisfies $1+p\le u$ and
\begin{equation}
 s(u-1)-L-m\ge0.
 \label{eq:opt:cap-deferral-condition}
\end{equation}
Then processing any initial job with $p=u$ before all later jobs are complete
cannot decrease the leading online cost.
\end{lemma}

\begin{proof}
Consider an initial job with $p=u$ that is completed after all initial jobs
have been tested but before a later job of value $p$.  The pair contributes
$1+u$ to the online cost.  If we postpone the initial job until after the
later job, the same pair contributes $2+p$.  The postponement lowers the cost
by $u-1-p\ge0$, so the postponed schedule is a valid lower bound on the
algorithm's actual cost.

It remains to consider jobs completed while the initial group of $p=u$ jobs
is still being tested.  Let $q$ be the fraction of the full instance formed
by these completed jobs.  Their
pairs with all later jobs add
$q[s(u-1)-L]$.  Their largest possible saving against not-yet-tested jobs
with $p=u$ is $qm-q^2/2$.  Relative to leaving every such job pending, the
change is therefore
at least
\begin{equation}
 q[s(u-1)-L-m]+\frac{q^2}{2}\ge0.
\end{equation}
The omitted one-job terms are lower order.
\end{proof}

\begin{lemma}[Accounting for the initial $p=u$ jobs]
\label{lem:opt:combined-accounting}
Let $O_0$, $A_0^{\rm lb}$, and $L$ be the offline coefficient, guaranteed
online coefficient, and total normalized positive processing time of the
later block, where every job satisfies $1+p\le u$.  After postponing all
initial jobs with $p=u$, the corresponding
coefficients $O$ and $A^{\rm lb}$ of the combined instance satisfy
\begin{align}
 A^{\rm lb}={}&A_0^{\rm lb}+m(1-m)+m(1+L)+\frac u2m^2,
 \label{eq:opt:combined-A}\\
 O={}&O_0+m(1-m+L)+\frac u2m^2.
 \label{eq:opt:combined-O}
\end{align}
In particular, $A^{\rm lb}-O=(A_0^{\rm lb}-O_0)+m$.
\end{lemma}

\begin{proof}
A pair consisting of one initial job with $p=u$ and one later job of value $p$
contributes $2+p$ online and $1+p$ offline.  Summed over the initial fraction $m$
and later fraction $1-m$, these give the cross terms involving $m(1-m)$ and
$mL$.
A pair of initial jobs with $p=u$ contributes $2+u$ online and $u$ offline,
which gives the terms proportional to $m^2$.  Adding the coefficients of the
later block yields the two displayed formulas.
Subtracting the two identities gives the final assertion.
\end{proof}

\begin{lemma}[Quota adversary allowing raw execution]
\label{lem:opt:cap-quota}
Fix $m\in(0,1)$ and a later block satisfying the hypotheses of
\cref{lem:opt:cap-deferral}, in particular $1+p\le u$ for each of its jobs.
Let $A^{\rm lb}$ and $O$ be the combined coefficients from
\cref{eq:opt:combined-A,eq:opt:combined-O}, and put
$D=A^{\rm lb}-R O$.  If $R\le u$ and $D\ge0$, then the quota construction
has nonnegative asymptotic excess at ratio $R$.

More precisely, suppose that along a sequence of quota crossings the
remaining block has size $S$ and $S/n\to\sigma$.  Up to a vanishing
finite-instance error, the excess divided by $n^2$ is at least
\begin{equation}
 \frac{u-R}{2}(1-\sigma^2)+\sigma^2D.
 \label{eq:opt:cap-quota-correct-excess}
\end{equation}
In particular, this lemma asserts preservation of the sign of the later
excess; it does not assert that the full coefficient $D$ survives when
$S<n$.
\end{lemma}

\begin{proof}
In the initial phase, a tested fresh job receives $p=u$, while a job run raw
is assigned hidden value zero.  If $v$ raw jobs and $\ell$ tested jobs with
$p=u$ have
been touched, stop at the first crossing
\begin{equation}
 \frac{\ell}{n-v}\ge m.
 \label{eq:opt:cap-quota}
\end{equation}
With $S=n-v$, the overshoot is less than one, so $\ell=mS+O(1)$.  If no
crossing occurs, no test can have occurred: after all first touches, any
positive $\ell$ would make the ratio in \cref{eq:opt:cap-quota} equal one.
Thus every job was run raw and the all-zero hidden instance has ratio
$u\ge R$.

At a crossing, scale the prescribed later instance to the remaining
$S-\ell$ first touches.  The crossing exceeds the target by fewer than one
job; absorb this error and integer rounding into the zero group.  Fix an
order of the job types in this later
block.  A test reveals the next value in that order; a raw action assigns the
next value to the touched job but does not reveal it to the algorithm.

Move every initial raw-zero completion to the start.  It has the same length
$u$ as processing a known job with $p=u$, and crossing a unit test only
lowers completion cost.
With
\begin{equation}
 Z_v=vn-\frac{v(v-1)}2,
\end{equation}
this prefix contributes $uZ_v$ online and $Z_v$ offline.  Subtract these
known contributions from the comparison.  The
remaining $S$ labels consist, up to one-job rounding, of a fraction $m$ of
jobs with $p=u$ and the prescribed later block satisfying $1+p\le u$.

To compare with the harmonic lower-bound construction, run a virtual copy of
the revealing algorithm.  If this virtual algorithm later runs a fresh job
of assigned value $p$ raw, the obligatory-testing simulation instead tests
and
immediately processes that job in time $1+p\le u$, hides the answer from the
virtual copy, and idles until the virtual raw block of length $u$ ends.
Inductively every simulated obligatory completion is no later than the
corresponding virtual revealing completion.  Since $1+p\le u$ on the later
block, its offline lengths are $1+p$ in both the revealing and obligatory
models.  Write $\ALG_{\rm hyb}$ and $\OPT_{\rm hyb}$ for the online and
offline costs
of this physical obligatory-testing simulation.  Then
\begin{equation}
 \ALG_{\mathsf{RO}}-\,R\cdot\OPT_{\mathsf{RO}}
 \ge (u-R)Z_v+
 \bigl(\ALG_{\rm hyb}-\,R\cdot\OPT_{\rm hyb}\bigr)-O_u(n).
 \label{eq:opt:cap-raw-decomposition}
\end{equation}
If $S/n\to\sigma$, then
\[
 \frac{Z_v}{n^2}=\frac{1-\sigma^2}{2}+o(1),
 \qquad
 \frac{\ALG_{\rm hyb}-\,R\cdot\OPT_{\rm hyb}}{n^2}
 \ge \sigma^2D-o(1).
\]
Substitution in \cref{eq:opt:cap-raw-decomposition} gives
\cref{eq:opt:cap-quota-correct-excess}.  Both terms there are nonnegative
when $R\le u$ and $D\ge0$.  If no quota crossing occurs, the all-zero
instance has ratio $u\ge R$, as observed above.

Finally, record every assigned value on its job label.  Replaying this fixed
input produces the same transcript because the values learned only by the
physical simulation were hidden from the virtual algorithm.  This proves
legality and turns each adaptive construction into a fixed input.
\end{proof}

\begin{proposition}[Lower bound for $\uDet{4}\le u\le\uDet{5}$]
\label{prop:opt:middle-lower}
For every $\uDet{4}\le u\le\uDet{5}$, every deterministic revealing
algorithm has size-asymptotic competitive ratio at least $1+c(u)$.
Equivalently,
\begin{equation}
 \RdetRO(u)\ge1+c(u).
 \label{eq:opt:middle-lower-final}
\end{equation}
\end{proposition}

The construction has two continuous parameters: $m$ is the fraction of
initial jobs with $p=u$, and $t$ selects the scaled harmonic block that
follows.
We first write its online/offline ratio, then solve the first-order conditions
for a worst parameter pair.  The remaining two claims verify that the pair is
feasible and can be approximated by fixed finite instances.

\begin{claim}[Ratio of the combined construction]
\label{clm:opt:combined-ratio}
For $m\in[0,1)$ and $t\in[1,e)$, apply
\cref{lem:opt:harmonic-block} with $s=1-m$ and prepend a fraction $m$ of jobs
with $p=u$.  The resulting guaranteed lower ratio is
\begin{equation}
 \mathcal R_u(m,t)=1+
 \frac{m+\frac{(1-m)^2}{2}(1-t^{-2})}
 {\frac{(1-m)^2}{4}(1+t^{-2}+2\ln t)
  +m(1-m)(1+\ln t)+\frac u2m^2}.
 \label{eq:opt:two-parameter-ratio}
\end{equation}
\end{claim}

\begin{proof}
Substitute $s=1-m$, $L=(1-m)\ln t$, and the coefficients from
\cref{eq:opt:harmonic-coefficients} into
\cref{eq:opt:combined-A,eq:opt:combined-O}, and divide the online coefficient
by the offline coefficient.  Collecting the excess over one gives exactly
\cref{eq:opt:two-parameter-ratio}.
\end{proof}

\begin{claim}[$(c,m)$-parametrization of a stationary point]
\label{clm:opt:cm-stationary-point}
Suppose the ratio in \cref{eq:opt:two-parameter-ratio} has an interior stationary
point with value $1+c$.  Combining this value equation with the two
vanishing partial derivatives gives
\begin{equation}
 \ln t=u-2-\frac1c,
 \qquad
 t^2=\frac{1-m}{1+m}\left(1+\frac2c\right).
 \label{eq:opt:cm-stationarity}
\end{equation}
Eliminating $t$ gives precisely
\cref{eq:opt:cm-parametrization}.  Conversely, those equations make
$\mathcal R_u(m,t)=1+c$.
\end{claim}

\begin{proof}
Write \cref{eq:opt:two-parameter-ratio} as $1+N/D$.  At a stationary
point of value $1+c$, the three equations are
\[
 N=cD,\qquad \partial_mN=c\,\partial_mD,\qquad
 \partial_tN=c\,\partial_tD.
\]
The $t$-derivative equation is
\[
 (1-m)^2t^{-3}
 =c\left[\frac{(1-m)^2}{2}(t^{-1}-t^{-3})
          +\frac{m(1-m)}t\right],
\]
which rearranges to the second identity in
\cref{eq:opt:cm-stationarity}.  Substituting it into the remaining two
equations and simplifying gives the first identity.  Eliminating $t$ from
these identities gives \cref{eq:opt:cm-parametrization}.  Reversing the
substitutions verifies both the value and the two vanishing derivatives.
\end{proof}

\begin{claim}[Feasibility and uniqueness of the parameters]
\label{clm:opt:cm-feasibility}
For each $u\in[\uDet{4},\uDet{5}]$, the equations in
\cref{eq:opt:cm-parametrization,eq:opt:cm-stationarity} determine a unique
triple $(c,m,t)$.  As $c$ decreases from $1/\varphi$ to $r$, this triple runs
continuously from
\begin{equation}
 (u,m,t)=\left(\uDet{4},\frac1\varphi,1\right)
 \quad\text{to}\quad
 (u,m,t)=(\uDet{5},0,\sqrt{\uDet{5}}).
 \label{eq:opt:cm-endpoints}
\end{equation}
It satisfies the deferral condition
\cref{eq:opt:cap-deferral-condition}; moreover, the finite harmonic block can
be chosen with $1+p<u$ for every one of its jobs.
\end{claim}

\begin{proof}
The parameter relations give $u\ge\uDet{4}$, $\ln t<1$, and
$m\le1/\varphi$.  Therefore
\begin{equation}
 (1-m)(u-1-\ln t)
 >(1-m)\varphi\ge\frac1\varphi\ge m,
\end{equation}
so \cref{eq:opt:cap-deferral-condition} holds; also $1+p<3<u$ for every
later positive value in a sufficiently fine finite construction.

In the first equation of \cref{eq:opt:cm-parametrization}, the left-hand
side $\operatorname{artanh}m-m$ has derivative
$m^2/(1-m^2)>0$, while the right-hand side has derivative
$2/[c^2(c+2)]>0$.  Hence it determines a unique $m=m(c)$, which is strictly
increasing.  Consequently
\begin{equation}
 u(c)=1+\frac2c-m(c)
\end{equation}
is strictly decreasing.  At $c=r$ both sides of the first equation vanish,
so $m=0$; at $c=1/\varphi$, direct substitution gives $m=1/\varphi$.
Together with $t$ from \cref{eq:opt:cm-stationarity}, the boundary value
$c=r$ gives $(u,m,t)=(\uDet{5},0,\sqrt{\uDet{5}})$, while
$c=1/\varphi$ gives
$(u,m,t)=(\uDet{4},1/\varphi,1)$, as recorded in
\cref{eq:opt:cm-endpoints}.
Consequently every
$u\in[\uDet{4},\uDet{5}]$ has exactly one parameter pair.
\end{proof}

\begin{claim}[Approximation by fixed finite instances]
\label{clm:opt:cm-finite-instances}
Fix $\uDet{4}<u<\uDet{5}$ and the associated parameters $(c,m,t)$ from
\cref{clm:opt:cm-feasibility}.  Against every deterministic revealing
algorithm, for every $\eps>0$ there are arbitrarily large fixed finite
instances $I$ such that
\begin{equation}
 \ALG(I)\ge(1+c-\eps)\OPT(I).
 \label{eq:opt:cm-finite-excess}
\end{equation}
\end{claim}

\begin{proof}
Fix $\eps>0$ and write $R_\eps=1+c-\eps$.  For the limiting combined
construction, \cref{clm:opt:combined-ratio} gives
\begin{equation}
 A^{\rm lb}-R_\eps O=\eps O>0.
 \label{eq:opt:cm-positive-margin}
\end{equation}
Choose rational approximations of $(m,t)$ that retain $1+p<u$ and make the
left side of \cref{eq:opt:cm-positive-margin} at least $\eps O/2$.

We spell out why the size of the harmonic remainder, which is selected by
the algorithm, does not invalidate the limiting construction.  Choose a
small constant $\delta>0$.  If the quota crossing leaves
$S\ge\delta n$ labels, use a finite harmonic approximation whose depth and
rational accuracy are chosen as functions of the available size $S$.
For accuracy level $j$, first fix the harmonic depth and rational data, and
then choose $N_j$ so that every $S\ge N_j$ realizes all block sizes with
rounding and one-job errors at most $j^{-1}S^2$.  Take
\[
 j(S)=\max\{j:N_j\le S\},
\]
after replacing $(N_j)$ by a strictly increasing sequence.  Thus
$j(S)\to\infty$, the harmonic depth tends to infinity, the rational data
converge to $(m,t)$, and all finite errors are $o(S^2)$ simultaneously.
The positive margin in \cref{eq:opt:cm-positive-margin} and
\cref{lem:opt:cap-quota} then give nonnegative excess at ratio $R_\eps$.

If instead $S<\delta n$, the raw prefix in
\cref{eq:opt:cap-raw-decomposition} contributes
\[
 (u-R_\eps)Z_v
 \ge \frac{u-R_\eps}{2}(1-\delta^2)n^2-O_u(n).
\]
The unfinished part has only $S$ jobs and all its operation lengths are
bounded in terms of $u$, so its possible negative excess is
$O_u(S^2+n)$.  Taking $\delta$ sufficiently small makes the displayed raw
slack dominate this loss.  This also covers values of $S$ too small for the
next harmonic accuracy level.

The quota overshoot, harmonic depth, rational approximation, and input size
have now been diagonalized in that order.  Recording every revealed value
on its label, as in \cref{lem:opt:cap-quota}, produces a fixed input and
proves \cref{eq:opt:cm-finite-excess} for arbitrarily large $n$.
\end{proof}

We can now prove the proposition.

\begin{proof}[Proof of \Cref{prop:opt:middle-lower}]
For $\uDet{4}<u<\uDet{5}$,
\cref{clm:opt:combined-ratio,clm:opt:cm-stationary-point} identify the target
ratio as $1+c(u)$, \cref{clm:opt:cm-feasibility} verifies the hypotheses of
the construction, and \cref{clm:opt:cm-finite-instances} supplies fixed
instances approaching that ratio.  At $u=\uDet{4}$ the same value is supplied
directly by the binary construction in \cref{prop:opt:binary-lower}.  At
$u=\uDet{5}$, \cref{lem:opt:obligatory-transfer,thm:lower} give the
obligatory lower bound $1+r$, which equals $1+c(\uDet{5})$ by
\cref{eq:opt:cm-endpoints}.  No passage between models by continuity is
used at either endpoint.
\end{proof}

\subsubsection{Transferring the lower bound when raw execution is no shorter}

The harmonic construction uses only values with $1+p_j\le u$, for which the
offline cap $u$ is inactive.  The next lemma shows that any revealing
algorithm with raw-execution cap $u$ induces an obligatory-testing algorithm
on these inputs without changing either the transcript comparison or the
offline optimum.

\begin{lemma}[Transfer from jobs with $1+p\le u$]
\label{lem:opt:obligatory-transfer}
Let $u\ge1$, and let $L_{\mathsf{OT}}(P)$ denote the deterministic
size-asymptotic value in obligatory testing, restricted to
$0\le p_j\le P$.  Then
\begin{equation}
 \RdetRO(u)\ge L_{\mathsf{OT}}(u-1).
 \label{eq:opt:obligatory-transfer}
\end{equation}
In particular,
\begin{equation}
 \RdetRO(u)\ge \RdetOT
 \qquad
 \left(u\ge1+\frac1r=2.75\ldots\right).
 \label{eq:opt:high-lower}
\end{equation}
\end{lemma}

\begin{proof}
Given a revealing algorithm with raw-execution cap $u$, construct an
obligatory-testing algorithm on instances with $p_j\le u-1$ by running a
virtual copy of the revealing algorithm.  Copy
each test and processing action.  When the virtual algorithm runs an
untouched job raw for $u$ time, physically test and immediately process it in
$1+p_j\le u$ time, hide the answer from the virtual copy, and wait for the
virtual clock.  Every physical completion is no later than its virtual
completion.  Moreover, $\min\{u,1+p_j\}=1+p_j$, so the revealing and
obligatory optima coincide.  Thus we obtain the following direct comparison.
For every revealing algorithm $\mathcal A$ with raw-execution cap $u$, the
simulation produces an obligatory-testing algorithm $\mathcal B$ and,
pointwise on every instance with $p_j\le u-1$,
\[
 \ALG_{\mathcal A}\ge\ALG_{\mathcal B},
 \qquad \OPT_u=\OPT_\infty.
\]
Consequently the size-asymptotic worst-case ratio of $\mathcal A$ is at
least the corresponding ratio of $\mathcal B$ restricted to
$p_j\le u-1$, which is at least $L_{\mathsf{OT}}(u-1)$ by its definition.
Taking the infimum over $\mathcal A$ proves
\cref{eq:opt:obligatory-transfer}.

The harmonic family attaining $\RdetOT$ has limiting largest
processing time
\begin{equation}
 1+\ln(1+y_*)
 =1+\frac12\ln \uDet{5}
 =\frac{\uDet{5}-1}{2}=\frac1r,
 \qquad y_*=\sqrt{\uDet{5}}-1.
\end{equation}
Choosing finite constructions whose largest processing time approaches
$1/r$ from below includes the equality case $u=1+1/r$ and proves
\cref{eq:opt:high-lower}.
\end{proof}

\subsection{Assembly of the curve}
\label{sec:det-curve-assembly}

It remains to verify the shape and joins of the six formulas in
\cref{eq:opt:full-curve} and to
pair each upper bound with its matching lower bound.  We do this explicitly
before transferring the uniform plateau witnesses to obligatory testing.

\subsubsection{Shape of the candidate curve}

Before matching the bounds, we verify the shape of the six formulas in
\cref{eq:opt:full-curve}.  The only monotonicity that needs proof is that of
$\Rquad(u)$ on $[\uDet{2},\uDet{3}]$; the formulas on
$[\uDet{3},\uDet{4}]$ and $[\uDet{4},\uDet{5}]$ follow directly from the
reciprocal formula and the $(c,m)$-parametrization in
\cref{eq:opt:cm-parametrization}.

\begin{lemma}[Monotonicity of the deterministic curve]
\label{lem:opt:curve-monotonicity}
The function $\Rquad(u)$ is strictly decreasing on
$[\uDet{2},\uDet{3}]$.  The function $1+1/\sqrt{u-1}$ is strictly
decreasing on $[\uDet{3},\uDet{4}]$, and $1+c(u)$ is strictly decreasing
on $[\uDet{4},\uDet{5}]$.  Consequently $\RdetRO(u)$ from
\cref{eq:opt:full-curve} attains its global maximum only at $u=\uDet{2}$.
\end{lemma}

\begin{proof}
Implicit differentiation of $T_s(\Rquad(u))=0$, with $s=u-1$, gives
\begin{equation}
 \Rquad'(u)=
 -\frac{R^2+(3s^2+2s)R-3(s+1)^2}
 {2sR+s^3+s^2+1}.
 \label{eq:opt:Rquad-derivative}
\end{equation}
Using $T_s(R)=0$, its numerator factors as
\[
 \frac{2s^3+s^2-1}{s}
 \left(R-
 \frac{2s^3+3s^2-2}{2s^3+s^2-1}\right).
\]
On the relevant interval $s\ge \uDet{2}-1>2/3$, the prefactor is
positive, while the fraction in parentheses is at most $5/3$ because
$4s^3-4s^2+1>0$.  The explicit inequality
\eqref{eq:opt:T-five-thirds} gives $R>5/3$, and hence
$\Rquad'(u)<0$.  The function $1+1/\sqrt{u-1}$ is visibly decreasing.  In
the $(c,m)$-parametrization on $[\uDet{4},\uDet{5}]$, $m(c)$ is increasing and
$u(c)=1+2/c-m(c)$ is decreasing, so $1+c$ decreases with $u$.
Thus the formula $\RdetRO(u)=u$ increases up to $u=\uDet{2}$; the formulas
$\Rquad(u)$, $1+1/\sqrt{u-1}$, and $1+c(u)$ then decrease until
$u=\uDet{5}$; and $\RdetRO(u)=\RdetOT$ is constant for
$u\ge\uDet{5}$.  This proves the global-maximum assertion.
\end{proof}

\subsubsection{Exact joins and completion of the proof}

We next check that adjacent formulas agree at every transition point.  Once
this consistency is recorded, the theorem follows by pairing the upper and
lower result already proved for each interval.

\begin{lemma}[Exact joins]
\label{lem:opt:curve-joins}
The six formulas in \cref{eq:opt:full-curve} agree at all five transition
points.  More explicitly, the formulas $1$ and $u$ agree at
$\uDet{1}=1$, and
\begin{align}
 \Rquad(\uDet{2})&=\uDet{2},
 \notag\\
 \Rquad(\uDet{3})&=1+\frac1{\sqrt{s_0}}
              =1.68\ldots,
 \notag\\
 1+\frac1{\sqrt{\uDet{4}-1}}&=\varphi
              =1+c(\uDet{4}),
 \notag\\
 1+c(\uDet{5})&=1+r=\RdetOT.
 \label{eq:opt:curve-joins}
\end{align}
\end{lemma}

\begin{proof}
The equality at $\uDet{1}=1$ is immediate.  Substitution of the defining
equations for $\uDet{2}$ and $s_0=\uDet{3}-1$ into $T_s$ gives the first
two displayed identities; uniqueness of the positive root identifies the
values with $\Rquad$.  Since
$\uDet{4}-1=\varphi+1=\varphi^2$, the formula
$1+1/\sqrt{u-1}$ equals $1+1/\varphi=\varphi$ at $u=\uDet{4}$.
The identities $1+c(\uDet{4})=\varphi$ and
$1+c(\uDet{5})=1+r$ follow from the parameter triples at
$u=\uDet{4}$ and $u=\uDet{5}$ in \cref{eq:opt:cm-endpoints}.
\end{proof}

\begin{proof}[Proof of \Cref{thm:optional-curve}]
For the upper bounds, \cref{lem:opt:raw-upper} covers
$0<u\le\uDet{2}$, \cref{lem:opt:ute-endpoint-game} covers
$\uDet{2}\le u\le\uDet{3}$, and \cref{lem:opt:zero-prefix} covers
$\uDet{3}\le u\le\uDet{4}$.  The guarantees on
$\uDet{4}\le u\le\uDet{5}$ and $u\ge\uDet{5}$ follow from
\cref{thm:opt:optional-middle-upper}.

For the matching lower bounds, \cref{prop:opt:binary-lower} covers
$0<u\le\uDet{4}$, \cref{prop:opt:middle-lower} covers
$\uDet{4}\le u\le\uDet{5}$, and
\cref{lem:opt:obligatory-transfer,thm:lower} cover $u\ge\uDet{5}$ with
instances satisfying $1+p_j<u$, so the offline cap $u$ is inactive.  The
exact joins are given by
\cref{lem:opt:curve-joins}, and the claimed shape follows from
\cref{lem:opt:curve-monotonicity}.  The finite upper bounds have the stated
$O_u(n)$ remainders, while
\cref{thm:opt:optional-middle-upper} gives one absolute $O(n)$ remainder on
the plateau.  This proves every assertion of the theorem.
\end{proof}

\begin{remark}[Representative values for $\uDet{4}\le u\le\uDet{5}$]
Some values for $u\in[\uDet{4},\uDet{5}]$ are
\begin{center}
\begin{tabular}{@{}rrrr@{}}
\toprule
$u$ & $\RdetRO(u)$ & $m$ & $t$\\
\midrule
$3.80$ & $1.60$ & $0.54$ & $1.14$\\
$4.00$ & $1.58$ & $0.43$ & $1.33$\\
$4.07$ & $1.58$ & $0.38$ & $1.41$\\
$4.30$ & $1.57$ & $0.20$ & $1.73$\\
\bottomrule
\end{tabular}
\end{center}
\end{remark}

\begin{remark}[Asymptotic scope]
\label{rem:opt:curve-audit}
All statements above concern the fixed-$u$ size-asymptotic ratio.  For
$0<u<\uDet{5}$, the upper bounds have additive error $O_u(n)$.  For
$u\in[\uDet{5},\infty)$ the same algorithm has one uniform $O(n)$ error.
The unbounded obligatory endpoint is not included in the statement of the
curve; it is derived in \cref{sec:obligatory}.
\end{remark}

\subsection{The obligatory endpoint}
\label{sec:obligatory}

The obligatory model is the unbounded endpoint of revealing optimization,
but the passage to that endpoint is not a formal substitution
$u=\infty$: both the algorithm and the additive error must be uniform in
$u$.  The finite curve has just established exactly those two properties,
so the endpoint now follows by a short transfer.

\begin{proof}[Proof of \Cref{thm:det-obligatory}]
For the upper bound, use \textsc{AdaptiveThreshold} with $c=r$.  This
algorithm never runs a job raw and is independent of $u$.  Given a finite
obligatory instance, choose any finite
\[
 u>\max\{\uDet{5},1+\max_j p_j\}.
\]
Its execution is identical in the revealing model with raw-execution cap
$u$, and every
offline effective length is
$\min\{u,1+p_j\}=1+p_j$.  Thus the revealing and obligatory optima coincide.
The constant $B$ in \cref{eq:opt:uniform-high-upper} is independent of $u$,
so that inequality also bounds the obligatory instance by
$\ALG\le\RdetOT\cdot\OPT+Bn$.

For the lower bound, regard an arbitrary obligatory algorithm $\mathcal B$
as a revealing algorithm with the fixed raw-execution cap $u=\uDet{5}$ that
simply never uses raw execution.  The plateau lower bound in
\cref{thm:optional-curve} may be
witnessed by instances with $p_j<\uDet{5}-1$, so the offline cap
$u=\uDet{5}$ is inactive.  On those instances the
online transcript is unchanged and
$\min\{\uDet{5},1+p_j\}=1+p_j$, so again the two offline optima coincide.  The
revealing lower inequality at $u=\uDet{5}$ therefore gives the claimed
obligatory lower bound for arbitrarily large instances.
\end{proof}

\section{Randomized revealing optimization: the exact four-piece curve}
\label{sec:random-common-upper}

\Cref{thm:common-upper-instance} identified the optimal leading online cost
for every fixed input in the revealing-optimization model.  We now compare that cost with the
clairvoyant optimum and maximize over all inputs.  The raw-execution cap $u$
is fixed, but
the algorithm is not told the job lengths in the input.  This section proves
\cref{thm:random-common-upper}, the rigorous randomized half of the informal
curve theorem \cref{thm:intro-common-upper-curves}.  Its three transition
points are
\begin{equation}
 \uRand{1}=1,
 \qquad
 \uRand{2}=5.04\ldots,
 \qquad
 \uRand{3}=\frac{25}{4}.
 \label{eq:rand-transition-points}
\end{equation}
Here $\uRand{2}>1$ is the unique root of
$u^3-6u^2+5u-1=0$ above one.

\begin{theorem}[Randomized revealing-optimization curve]
\label{thm:random-common-upper}
For every fixed $u>0$, the optimal randomized size-asymptotic competitive
ratio against an oblivious adversary is
\begin{equation}
 \RrandRO(u)=
 \begin{cases}
  1,&0<u\le\uRand{1},\\[2pt]
  \displaystyle\frac{u^3}{u^2+(u-1)^3},
    &\uRand{1}\le u\le\uRand{2},\\[8pt]
  \displaystyle 1+\frac{1}{2(\sqrt u-1)},
    &\uRand{2}\le u\le\uRand{3},\\[8pt]
  \displaystyle\frac43,&u\ge\uRand{3}.
 \end{cases}
 \label{eq:intro-random-common-upper-curve}
\end{equation}
More precisely, there are randomized nonanticipating algorithms
$\mathcal A_{n,u}$ and numbers $\eps_u(n)\to0$ such that every
input $p\in[0,u]^n$ satisfies
\begin{equation}
 \EE\ALG_{\mathcal A_{n,u}}(p)
 \le \RrandRO(u)\cdot\OPT_u(p)+\eps_u(n)n^2.
 \label{eq:intro-random-common-upper}
\end{equation}
Conversely, for every randomized online algorithm $\mathcal A$ and every
$\eps>0$, there
are arbitrarily large $n$ and a fixed input $p\in[0,u]^n$ for which
\begin{equation}
 \EE\ALG_{\mathcal A}(p)
 \ge\bigl(\RrandRO(u)-\eps\bigr)\cdot\OPT_u(p).
 \label{eq:random-common-upper-lower}
\end{equation}
The curve is continuous and has global maximum
\begin{equation}
 \max_{u>0}\RrandRO(u)
 =\frac{27+6\sqrt3}{23}
 =1.62\ldots,
 \label{eq:intro-random-common-maximum}
\end{equation}
attained at $u=(3+\sqrt3)/2$.
\end{theorem}

The full curve is shown in \cref{fig:randomized-revealing-curve}.

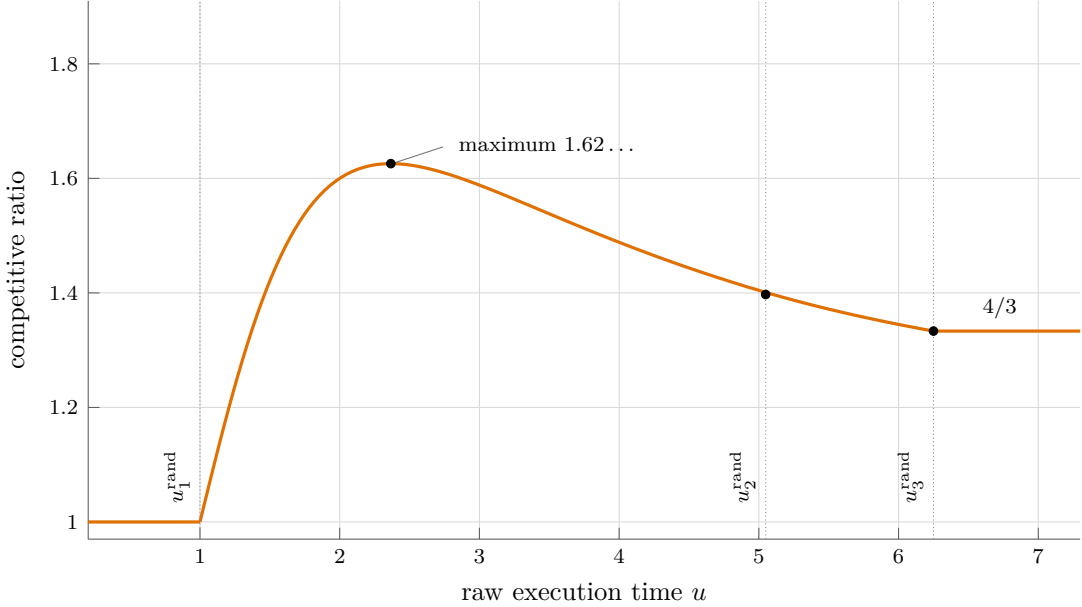
\begin{figure}[htbp]
\centering
\begingroup
\tracinglostchars=0
\begin{tikzpicture}
\begin{axis}[
  width=0.98\linewidth,
  height=0.58\linewidth,
  xmin=0.2,
  xmax=7.3,
  ymin=0.97,
  ymax=1.91,
  xlabel={raw execution time $u$},
  ylabel={competitive ratio},
  xtick={1,2,3,4,5,6,7},
  ytick={1,1.2,1.4,1.6,1.8},
  tick label style={font=\scriptsize},
  label style={font=\small},
  grid=major,
  grid style={black!15},
  axis line style={black!65},
  tick style={black!65},
  axis x line*=bottom,
  axis y line*=left,
  samples=120,
  no markers,
  clip=true,
]

\addplot[orange!88!black,very thick,domain=0.2:1] {1};
\addplot[orange!88!black,very thick,domain=1:5.04891733952]
  {x^3/(x^2+(x-1)^3)};
\addplot[orange!88!black,very thick,domain=5.04891733952:6.25]
  {1+1/(2*(sqrt(x)-1))};
\addplot[orange!88!black,very thick,domain=6.25:7.3] {4/3};
\addplot[black!45,densely dotted,thin] coordinates {(1,0.97) (1,1.91)};
\node[font=\scriptsize,rotate=90,anchor=south]
  at (axis cs:1,1.08) {$\uRand{1}$};
\addplot[black!45,densely dotted,thin]
  coordinates {(5.04891733952,0.97) (5.04891733952,1.91)};
\node[font=\scriptsize,rotate=90,anchor=south]
  at (axis cs:5.04891733952,1.08) {$\uRand{2}$};
\addplot[black!45,densely dotted,thin]
  coordinates {(6.25,0.97) (6.25,1.91)};
\node[font=\scriptsize,rotate=90,anchor=south]
  at (axis cs:6.25,1.08) {$\uRand{3}$};
\addplot[only marks,mark=*,mark size=1.6pt,black] coordinates {
  (2.36602540378,1.62575238458)
  (5.04891733952,1.39719767662)
  (6.25,1.33333333333)
};
\node[font=\scriptsize,anchor=west] at (axis cs:2.78,1.66)
  {maximum $1.62\ldots$};
\draw[black!55,thin] (axis cs:2.74,1.655) -- (axis cs:2.39,1.627);
\node[font=\scriptsize,anchor=west] at (axis cs:6.53,1.375) {$4/3$};

\end{axis}
\end{tikzpicture}
\endgroup
\caption{The exact randomized revealing-optimization curve.  The two
binary hard-instance families switch at the marked transition points.}
\label{fig:randomized-revealing-curve}
\end{figure}

\Cref{thm:common-upper-instance} identifies the best asymptotic online cost
on each input.  Here we compare that value with the clairvoyant SPT cost and
maximize over all empirical distributions.  Replacing an initial part of the
survival function by its average reduces this optimization to distributions
supported on two values.  Two one-variable maximizations give the four
formulas in \cref{eq:intro-random-common-upper-curve}.  A sublinear sample
removes advance knowledge of the input,
and shuffled binary inputs give matching lower bounds.

\subsection{Offline and stationary fluid coefficients}

Assume initially that $u>1$ and write
\begin{equation}
 s=u-1.
 \label{eq:rcu-s}
\end{equation}
Let $D$ be a finite distribution on $[0,u]$ and use the survival function
\begin{equation}
 S_D(t)=D((t,u]),
 \qquad S_D(t)=0\quad(t>u).
 \label{eq:rcu-survival}
\end{equation}
The convention at point masses does not affect any integral below.  Since a
clairvoyant job of length $p$ has effective length $1+\min\{s,p\}$, write a
nonnegative value $x$ as $\int_0^\infty\one_{t<x}\,dt$ in the symmetric pair
identity \cref{eq:unified-opt-symmetric}.  This gives
\begin{align}
 \Omega_{\mathsf{RO},u}(D)
 &\defeq\frac12\iint
       \min\{u,1+p,1+q\}\,dD(p)dD(q)\notag\\
 &=\frac12\left(1+\int_0^s S_D(t)^2\,dt\right),
 \label{eq:rcu-offline-fluid}\\
 \OPT_u(p)
 &=n^2\Omega_{\mathsf{RO},u}(D[p])
   +\frac12\sum_i\min\{u,1+p_i\}
 \label{eq:rcu-offline-finite}
\end{align}
for every empirical distribution $D[p]$.

We next specialize the maximum-density module
\cref{lem:maximum-density-module} to a form suited to the offline cap $u$.
If its
threshold prefix has mass $a$ and first moment $\ell$, put
\begin{equation}
 \tau_D=\frac{1+\ell}{a}.
 \label{eq:rcu-tau-density}
\end{equation}
The common threshold equation and the layer-cake identity give
\begin{equation}
 \int(\tau_D-p)^+\,dD(p)=1,
 \qquad
 \int_0^{\tau_D} S_D(t)\,dt=\tau_D-1.
 \label{eq:rcu-threshold-equation}
\end{equation}

The stationary algorithm tests all jobs in a uniform random order.  It
immediately processes a positive job if its revealed value lies in the
selected threshold prefix, and after the last test it processes all remaining
jobs in SPT order.  The following formula is the
stationary pair identity specialized to the finite offline cap $u$.

\begin{lemma}[Stationary survival formula]
\label{lem:rcu-stationary-survival}
The doubled leading coefficient of the stationary algorithm is
\begin{equation}
 \mathcal T_u(D)=
 \begin{cases}
  \displaystyle \tau_D+\int_{\tau_D}^u S_D(t)^2\,dt,&\tau_D\le u,\\[6pt]
  \tau_D,&\tau_D>u.
 \end{cases}
 \label{eq:rcu-stationary-survival}
\end{equation}
Consequently the better of the raw and stationary algorithms has leading
ratio
\begin{equation}
 \frac{\min\{u,\mathcal T_u(D)\}}
      {1+\int_0^sS_D(t)^2\,dt}.
 \label{eq:rcu-announced-ratio}
\end{equation}
\end{lemma}

\begin{proof}
Let $h=1-a$ be the residual mass.  Direct pair accounting gives the leading
coefficient
\begin{equation}
 (1+\ell)\left(1-\frac a2\right)+\SPT(\nu).
 \label{eq:rcu-stationary-pair}
\end{equation}
If $\tau_D\le u$, the threshold property gives
\[
 \SPT(\nu)=\frac{\tau_D h^2}{2}
            +\frac12\int_{\tau_D}^uS_D(t)^2\,dt,
 \qquad
 (1+\ell)\left(1-\frac a2\right)
 =\frac{\tau_D}2(1-h^2).
\]
Their sum is one half of the first line of
\eqref{eq:rcu-stationary-survival}.  If $\tau_D>u$, the whole distribution is
in the prefix and \eqref{eq:rcu-stationary-pair} equals $\tau_D/2$.  Raw
execution has leading coefficient $u/2$, and division by
\eqref{eq:rcu-offline-fluid} proves \eqref{eq:rcu-announced-ratio}.
\end{proof}

This is where the instance-optimal theorem enters the curve calculation.
In \cref{eq:cui-F}, $q=0$ is the raw algorithm and $q=1$ is precisely the
stationary algorithm using the selected threshold prefix.  Hence
\begin{equation}
 2\PhiRO(D)\le\min\{u,\mathcal T_u(D)\},
 \label{eq:rcu-Phi-endpoint-upper}
\end{equation}
and \cref{eq:rcu-announced-ratio} is an upper bound on the
instance-specific ratio $\PhiRO(D)/\Omega_{\mathsf{RO},u}(D)$.  Interior values of $q$ can
strictly improve individual inputs; the point of the next two subsections is
that the two choices $q=0$ and $q=1$ already give the sharp value after maximizing
over all $D$.  The binary lower bounds at the end apply to every algorithm,
not only to algorithms restricted to $q\in\{0,1\}$.

\subsection{Flattening the survival function}

We now maximize \eqref{eq:rcu-announced-ratio} over arbitrary $D$.  Suppose
first that $1\le\tau_D\le u$ and put
\begin{equation}
 y=\frac{\tau_D-1}{\tau_D}.
 \label{eq:rcu-y}
\end{equation}
By \eqref{eq:rcu-threshold-equation}, $y$ is the average of $S_D$ on
$[0,\tau_D]$.  Since $S_D$ is nonincreasing, its average on every initial
subinterval is at least $y$.  Thus Cauchy--Schwarz gives
\begin{equation}
 \int_0^aS_D(t)^2\,dt\ge ay^2
 \qquad(0\le a\le\tau_D),
 \label{eq:rcu-initial-cs}
\end{equation}
and $S_D(t)\le y$ almost everywhere for $t>\tau_D$.

Replace $S_D$ on $(0,\tau_D)$ by the constant $y$ and leave its tail unchanged.
This remains a nonincreasing survival function, preserves
\eqref{eq:rcu-threshold-equation} and $\mathcal T_u(D)$, and by
\eqref{eq:rcu-initial-cs} can only decrease the offline coefficient.  Hence
it can only increase the ratio.  This replacement is illustrated in
\cref{fig:survival-flattening}.  We may therefore assume
\begin{equation}
 S_D(t)=y\quad(0<t<\tau_D).
 \label{eq:rcu-flat-prefix}
\end{equation}

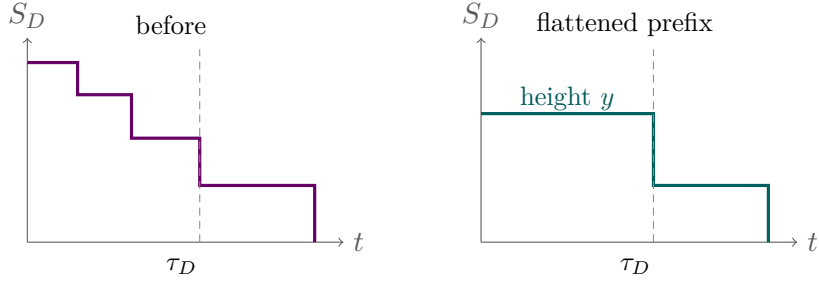
\begin{figure}[htbp]
\centering
\begin{tikzpicture}[x=0.95cm,y=2.5cm]
  \begin{scope}
    \draw[->,black!65] (0,0)--(4.4,0) node[right] {$t$};
    \draw[->,black!65] (0,0)--(0,1.08) node[above] {$S_D$};
    \draw[violet!80!black,very thick]
      (0,0.95)--(0.7,0.95)--(0.7,0.78)--(1.45,0.78)--
      (1.45,0.55)--(2.4,0.55)--(2.4,0.3)--(4.0,0.3)--(4.0,0);
    \draw[black!45,densely dashed] (2.4,0)--(2.4,1.02);
    \node[font=\small] at (2.15,-0.13) {$\tau_D$};
    \node[font=\small] at (2.0,1.15) {before};
  \end{scope}
  \begin{scope}[xshift=6cm]
    \draw[->,black!65] (0,0)--(4.4,0) node[right] {$t$};
    \draw[->,black!65] (0,0)--(0,1.08) node[above] {$S_D$};
    \draw[teal!75!black,very thick]
      (0,0.68)--(2.4,0.68)--(2.4,0.3)--(4.0,0.3)--(4.0,0);
    \draw[black!45,densely dashed] (2.4,0)--(2.4,1.02);
    \node[font=\small] at (2.15,-0.13) {$\tau_D$};
    \node[font=\small] at (2.0,1.15) {flattened prefix};
    \node[font=\small,teal!70!black] at (1.2,0.76) {height $y$};
  \end{scope}
\end{tikzpicture}
\caption{Survival-function flattening preserves the threshold constraint
and stationary numerator while decreasing the offline denominator.}
\label{fig:survival-flattening}
\end{figure}

If $u-1\le\tau_D\le u$, the offline integral lies entirely in the flat part,
and $S_D(t)^2\le y^2$ on the remaining tail.  It follows that
\begin{equation}
 \frac{\min\{u,\mathcal T_u(D)\}}
      {1+\int_0^sS_D(t)^2\,dt}
 \le
 B_u(\tau_D)
 \defeq
 \frac{\tau_D+(u-\tau_D)y^2}{1+(u-1)y^2}.
 \label{eq:rcu-B-first}
\end{equation}
The numerator displayed on the right is at most $u$, so the raw minimum is
inactive.

Now suppose $\tau_D<u-1$.  Put
\begin{equation}
 L=u-1-\tau_D,
 \quad E=\int_{\tau_D}^{u-1}S_D(t)^2\,dt,
 \quad G=\int_{u-1}^uS_D(t)^2\,dt.
 \label{eq:rcu-EG}
\end{equation}
Monotonicity gives
\begin{equation}
 0\le E\le Ly^2,
 \qquad G\le E/L.
 \label{eq:rcu-E-range}
\end{equation}
Consequently
\begin{equation}
 \frac{\min\{u,\mathcal T_u(D)\}}
      {1+\int_0^sS_D(t)^2\,dt}
 \le
 \frac{\tau_D+\dfrac{u-\tau_D}{u-1-\tau_D}E}
      {1+\tau_D y^2+E}.
 \label{eq:rcu-linear-fractional-E}
\end{equation}
For fixed $\tau_D$ this is linear-fractional in $E$, so its derivative has a
constant sign.  Its maximum on $[0,Ly^2]$ therefore occurs at $E=0$ or
$E=Ly^2$.  Substitution of these two values gives, respectively,
\begin{align}
 A(\tau_D)&\defeq\frac{\tau_D}{1+\tau_D y^2}
              =\frac{\tau_D^2}{\tau_D^2-\tau_D+1},
 \label{eq:rcu-A}\\
 B_u(\tau_D)&=\frac{\tau_D+(u-\tau_D)y^2}{1+(u-1)y^2}.
 \label{eq:rcu-B}
\end{align}
They are attained by rectangular survival functions: respectively a binary
distribution on $\{0,\tau_D\}$ and one on $\{0,u\}$.

Thus the two scalar candidates have a concrete interpretation.  The
$B_u$ family, supported on $\{0,u\}$, determines the curve from $u=1$
through $u=25/4$ (with its optimizer moving from the boundary
$\tau_D=u$ to an interior value at $\uRand{2}$).  The $A$ family, supported
on $\{0,\tau_D\}$, takes over at $u=25/4$; its choice
$\tau_D=2$ gives the obligatory $4/3$ plateau.  The calculations below
verify these dominance ranges.

It remains to cover $\tau_D>u$.  Equation
\eqref{eq:rcu-threshold-equation} and the zero extension of $S_D$ give
$\int_0^uS_D(t)\,dt=\tau_D-1\ge u-1$.  The average of a nonincreasing function on
$[0,u-1]$ is no smaller than its average on $[0,u]$, and hence
\begin{equation}
 \int_0^{u-1}S_D(t)^2\,dt
 \ge (u-1)\left(\frac{\tau_D-1}{u}\right)^2
 \ge\frac{(u-1)^3}{u^2}.
 \label{eq:rcu-tau-large}
\end{equation}
Here the raw algorithm is chosen, and its ratio is at most $B_u(u)$.
We have proved the reduction
\begin{equation}
 \sup_D\frac{\min\{u,\mathcal T_u(D)\}}
                 {1+\int_0^sS_D(t)^2\,dt}
 \le
 \max\left\{
  \sup_{1\le\tau\le u-1}A(\tau),
  \sup_{1\le\tau\le u}B_u(\tau)
 \right\},
 \label{eq:rcu-binary-reduction}
\end{equation}
where the first supremum is omitted for $u<2$.  In fact equality holds in
\eqref{eq:rcu-binary-reduction}, because the binary distributions on
$\{0,\tau_D\}$ and $\{0,u\}$ identified before the reduction attain the
$A$ and $B_u$ families, respectively.

\subsection{The two scalar maximizations}

The first family satisfies
\begin{equation}
 A'(\tau)=\frac{\tau(2-\tau)}{(\tau^2-\tau+1)^2}.
 \label{eq:rcu-A-derivative}
\end{equation}
Thus its maximum is attained at $\tau=\min\{2,u-1\}$ and equals $4/3$ as
soon as $u\ge3$.

For the second family, substitution of \eqref{eq:rcu-y} gives
\begin{equation}
 B_u(\tau)=
 \frac{\tau^2u+2\tau^2-2\tau u-\tau+u}
      {\tau^2u-2\tau u+2\tau+u-1}.
 \label{eq:rcu-B-rational}
\end{equation}
The sign of its derivative is the sign of
\begin{equation}
 g_u(\tau)=(2\tau-1)^2-u(\tau-1)^2.
 \label{eq:rcu-B-derivative}
\end{equation}
For $u\le4$ this is positive on $[1,\infty)$.  For $u>4$, its unique root
above one is
\begin{equation}
 \tau_*(u)=\frac{u-2+\sqrt u}{u-4}
           =\frac{\sqrt u-1}{\sqrt u-2},
 \label{eq:rcu-tau-star}
\end{equation}
and $B_u$ increases before this point and decreases afterwards.  The
identity $\tau_*(u)=u$ is equivalent to
\begin{equation}
 u^3-6u^2+5u-1=0.
 \label{eq:rcu-transition-cubic}
\end{equation}
The cubic is strictly increasing beyond its last critical point and changes
sign there exactly once; its unique root above one is $\uRand{2}$ from
\eqref{eq:rand-transition-points}.  Hence
\begin{align}
 \max_{1\le\tau\le u}B_u(\tau)
 &=B_u(u)=\frac{u^3}{u^3-2u^2+3u-1}
       =\frac{u^3}{u^2+(u-1)^3},&&u\le\uRand{2},
 \label{eq:rcu-B-at-u}\\
 &=B_u(\tau_*(u))
       =1+\frac1{2(\sqrt u-1)},&&u\ge\uRand{2}.
 \label{eq:rcu-B-at-star}
\end{align}

It remains to compare the candidate $A$ with the values obtained from $B_u$.
On $2\le u\le3$, clearing the positive
denominators in $B_u(u)-A(u-1)$ and putting $v=u-2$ leaves
\[
 v^4+3v^3+3v^2+3v+3>0.
\]
On $3\le u\le\uRand{2}$, one has $B_u(u)>4/3$.  On
$\uRand{2}\le u\le\uRand{3}$, \eqref{eq:rcu-B-at-star} is at least $4/3$,
with equality exactly at $u=\uRand{3}$; afterwards $A(2)=4/3$ is larger.  Combining
these observations with \eqref{eq:rcu-binary-reduction} proves the announced
upper bound
\begin{equation}
 \frac{\min\{u,\mathcal T_u(D)\}}
      {1+\int_0^sS_D(t)^2\,dt}
 \le\RrandRO(u)
 \label{eq:rcu-announced-upper}
\end{equation}
for every finite distribution $D$.

The junction statements now follow directly.
\Cref{eq:rcu-transition-cubic,eq:rcu-B-at-u,eq:rcu-B-at-star} give
continuity at $\uRand{2}$, and the formula
$1+1/[2(\sqrt u-1)]$ equals $4/3$ at $u=\uRand{3}$.
Differentiating \eqref{eq:rcu-B-at-u} shows that its only
interior maximum above one occurs at $u=(3+\sqrt3)/2$; substitution gives
\eqref{eq:intro-random-common-maximum}.

\subsection{Worst-case upper bound from instance optimality}

No second learning argument is needed.  Apply the universal unannounced
algorithm of \cref{thm:common-upper-instance} to $D=D[p]$.  The endpoint
comparison \cref{eq:rcu-Phi-endpoint-upper} and the scalar bound
\cref{eq:rcu-announced-upper} give
\[
 \PhiRO(D[p])
 \le \RrandRO(u)\cdot\Omega_{\mathsf{RO},u}(D[p]).
\]
The exact empirical identity \cref{eq:rcu-offline-finite} has a nonnegative
linear correction.  Therefore the instance-optimal upper bound immediately
yields
\begin{equation}
 \EE\ALG_{\mathcal A_{n,u}}
 \le\RrandRO(u)\cdot\OPT_u+\,o_u(n^2),
 \label{eq:rcu-final-upper}
\end{equation}
which is \eqref{eq:intro-random-common-upper}.  For $0<u\le1$, raw execution
is also optimal for the clairvoyant scheduler, so no learning is required.

\subsection{Binary oblivious lower bounds}

It remains to show that no adaptive algorithm improves the resulting
worst-case upper bound.  We first record an
envelope for shuffled binary inputs.

\begin{lemma}[Binary-input envelope at the offline cap $u$]
\label{lem:rcu-binary-lower}
Fix $u>1$, $x\in[0,1]$, and put $z=1-x$.  Let $p^{(n)}$ have $xn+O(1)$
jobs of value $u$ and $zn+O(1)$ zero jobs.  For every
randomized nonanticipating algorithm, even an announced one,
\begin{equation}
 \frac{\EE\ALG}{n^2}\ge
 \begin{cases}
  \displaystyle\frac u2-o_u(1),&x\ge (u-1)/u,\\[6pt]
  \displaystyle\frac{1+x+ux^2}{2}-o_u(1),&x\le(u-1)/u.
 \end{cases}
 \label{eq:rcu-binary-fluid-lower}
\end{equation}
Every labeling has
\begin{equation}
 \frac{\OPT_u}{n^2}
 =\frac{1+(u-1)x^2}{2}+O_u(1/n).
 \label{eq:rcu-binary-offline}
\end{equation}
\end{lemma}

\begin{proof}
First fix a deterministic algorithm.  If it does not complete all jobs, the
claim is immediate.  Otherwise place the individual jobs uniformly on the
public labels, treating jobs of the same value as distinct.  Reveal a job's
value in the analysis only when the algorithm first touches its label.
Conditional on all earlier reveals, the unrevealed jobs remain uniformly
permuted, and the algorithm has already made its test/raw choice before the
next value is revealed.

Take $\delta_n=n^{-1/4}$.  The predictable-permutation bound
\cref{lem:predictable-permutation-sampling}, applied to the zero indicator
and the test selector, gives one event of probability $1-o(1)$ on which
every prefix of the order in which labels are first touched that ends with at
least $\delta_n n$ untouched labels and contains $t$ tests
satisfies
\begin{equation}
 \bigl|N_0(t)-zt\bigr|\le
 \Delta_n,
 \qquad
 \Delta_n=O(\delta_n^{-1}\sqrt{n\ln n})=o(n).
 \label{eq:rcu-predictable-urn}
\end{equation}
Here $N_0(t)$ is the number of tested zeros.  The event is simultaneous, so
we do not condition on the algorithm's adaptive stopping time.

Edit the suffix by running raw every still-untouched job that the original
schedule would later test, deleting its later known-processing operation.
Only $\delta_n n+O(1)$ jobs change, so the edit changes total completion cost
by $o_u(n^2)$.  Let $q$ be the final tested fraction of the edited schedule.
On \eqref{eq:rcu-predictable-urn}, every prefix of its completion curve is
dominated, up to $o(1)$ vertical error, by the divisible items
\begin{center}
\begin{tabular}{@{}lcc@{}}
\toprule
item & completion capacity & work per completion\\
\midrule
tests that reveal zero jobs & $zq$ & $1/z$\\
processing of tested $p=u$ jobs & $xq$ & $u$\\
raw jobs & $1-q$ & $u$\\
\bottomrule
\end{tabular}
\end{center}
When $z=0$, the first row is simply omitted.
Indeed, $t$ tests complete at most $zt+\Delta_n$ zeros and expose at most
$xt+\Delta_n$ jobs with $p=u$; raw execution and processing of a known
$p=u$ job both use $u$ work per completion.  Ordering these three divisible
blocks by work per completion,
as in \cref{lem:divisible-spt-area}, therefore upper-bounds completed mass at every
work value.  Integrating the remaining mass gives a lower bound on total
completion time.

If $z\le1/u$, every displayed item costs at least $u$ per completion, so the
normalized area is at least $u/2-o_u(1)$.  If $z\ge1/u$, the tests that
reveal zero jobs come first.  For fixed $q$ the resulting remaining-mass area is
\begin{equation}
 F_z(q)=q-\frac{zq^2}{2}+\frac u2(1-zq)^2,
 \qquad
 F_z'(q)=(1-zq)(1-uz).
 \label{eq:rcu-binary-F}
\end{equation}
It is minimized at $q=1$, giving
$F_z(1)=(1+x+ux^2)/2$.  Restoring the edited final $\delta_n n+O(1)$ jobs,
the $o(1)$ error in completed mass, and the sampling event's failure
probability together contribute only $o_u(n^2)$.

The argument holds for each deterministic algorithm.  For a randomized algorithm,
fix its seed and then average over the seed.  The private shuffle is already
part of the input convention.  Finally, the effective lengths of the two job types are
$1$ and $u$, so \eqref{eq:rcu-binary-offline} follows from
\cref{lem:offline}.
\end{proof}

We now select the binary masses.  If $\uRand{1}<u\le\uRand{2}$, take
\begin{equation}
 x=\frac{u-1}{u}.
 \label{eq:rcu-lower-x-crossing}
\end{equation}
This is the boundary $z=1/u$ in \cref{lem:rcu-binary-lower}, and its online
and offline coefficients have ratio
\begin{equation}
 \frac{u}{1+(u-1)((u-1)/u)^2}
 =\frac{u^3}{u^2+(u-1)^3}.
 \label{eq:rcu-lower-first-branch}
\end{equation}
If $\uRand{2}\le u\le\uRand{3}$, take
\begin{equation}
 x=\frac1{\sqrt u-1}.
 \label{eq:rcu-lower-x-stationary}
\end{equation}
The definition of $\uRand{2}$ is precisely the point from which
$x\le(u-1)/u$.  Since $(u-1)x^2=1+2x$, the second line of
\eqref{eq:rcu-binary-fluid-lower} divided by
\eqref{eq:rcu-binary-offline} becomes
\begin{equation}
 1+\frac x2=1+\frac1{2(\sqrt u-1)}.
 \label{eq:rcu-lower-second-branch}
\end{equation}
Irrational masses are approximated by integer job counts divided by $n$; the
rounding error and the one-job terms omitted from the leading pair
coefficient are $O_u(n)$.

For $u\ge\uRand{3}$, the lower bound $4/3$ uses a fixed bounded
distribution supported on $\{0,2\}$ instead of the binary distribution on
$\{0,u\}$ used for smaller $u$.

\begin{lemma}[Balanced zero--two lower bound]
\label{lem:rcu-zero-two-lower}
For every $u\ge\uRand{3}$ and every randomized algorithm, arbitrarily large even
$n$ admit an input with $n/2$ zero jobs and $n/2$ jobs of value two such that
\begin{equation}
 \EE\ALG\ge n^2-o_u(n^2),
 \qquad
 \OPT_u=\frac34n^2+\,O(n).
 \label{eq:rcu-zero-two-costs}
\end{equation}
\end{lemma}

\begin{proof}
Expose values when their labels are first touched, apply the
predictable-permutation event \cref{eq:rcu-predictable-urn} simultaneously
to every prefix of that order, and replace tests among the final $o(n)$
untouched jobs by raw execution as in the proof of
\cref{lem:rcu-binary-lower}.  For a realized tested fraction
$q$, the three divisible blocks now have capacities and costs
\begin{center}
\begin{tabular}{@{}lcc@{}}
\toprule
item & completion capacity & work per completion\\
\midrule
tests that reveal zero jobs & $q/2$ & $2$\\
processing of tested jobs with $p=2$ & $q/2$ & $2$\\
raw jobs & $1-q$ & $u$\\
\bottomrule
\end{tabular}
\end{center}
Every item costs at least two per completion.  Thus the completion envelope
lies below $\min\{1,t/2\}$ independently of $q$, whose remaining-mass area is
one.  This proves the online estimate after the $o_u(n^2)$ repairs and
averaging over the algorithm seed.  Since $u>3$, the two effective lengths
are one and three; their equally weighted divisible SPT cost gives offline
coefficient $3/4$.
\end{proof}

Equations \eqref{eq:rcu-lower-first-branch} and
\eqref{eq:rcu-lower-second-branch}, together with
\cref{lem:rcu-zero-two-lower}, match \eqref{eq:rcu-final-upper} on
$1<u\le\uRand{2}$, $\uRand{2}\le u\le\uRand{3}$, and
$u\ge\uRand{3}$, respectively.  For $0<u\le1$, raw execution is both an
online and a clairvoyant
optimum.  This completes the proof of the exact curve in
\cref{thm:random-common-upper}.

\subsection{The randomized obligatory endpoint}
\label{sec:random-obligatory}

We now pass from the randomized curve with a finite raw-execution cap $u$ to
the unbounded obligatory endpoint.  The adversary fixes an input
$p\in[0,\infty)^n$ before the algorithm
draws its random bits; the input convention then privately shuffles the
jobs.  Their lengths are not announced, and no upper bound on
the processing times is assumed.  The zero--two lower bound at the finite
raw-execution cap $u=\uRand{3}$
transfers directly.  The upper bound needs one additional step: the
algorithm groups every sampled value above a threshold into one category and
switches to a safe test-all schedule when the sample does not identify a
dense enough early group.  Letting this threshold grow slowly
also makes the same algorithm instance-optimal on every bounded class.  This
subsection proves
\cref{thm:rand-obligatory},
the rigorous randomized part of the informal obligatory-testing result
\cref{thm:intro-obligatory} from the introduction.

\begin{theorem}[Randomized obligatory testing]
\label{thm:rand-obligatory}
There are a universal constant $C$ and a randomized nonanticipating
algorithm $\mathcal A_n^*$, depending only on $n$, such that every fixed
input $p\in[0,\infty)^n$ satisfies
\begin{equation}
 \EE\ALG_{\mathcal A_n^*}(p)
 \le \frac43\cdot\OPT(p)+Cn^{19/10}
 \le\left(\frac43+2Cn^{-1/10}\right)\cdot\OPT(p).
 \label{eq:rand-obligatory-upper}
\end{equation}
A fixed-cutoff companion $\mathcal A_n^{(32)}$, using $B=32$ and the more
conservative fallback criterion $\widehat\kappa<2/B$, satisfies the sharper
explicit estimate
\begin{equation}
 \EE\ALG_{\mathcal A_n^{(32)}}(p)
 \le \frac43\cdot\OPT(p)+20378n^{7/4}.
 \label{eq:rand-obligatory-fixed-upper}
\end{equation}
Conversely, for every randomized algorithm $\mathcal A$ and every $\eps>0$,
arbitrarily large even $n$ admit an input with $n/2$ zero jobs and
$n/2$ jobs of value two such that
\begin{equation}
 \EE\ALG_{\mathcal A}(p)
 \ge\left(\frac43-\eps\right)\cdot\OPT(p),
 \qquad
 \OPT(p)=\frac34n^2+n.
 \label{eq:rand-obligatory-lower}
\end{equation}
Hence the exact randomized size-asymptotic competitive ratio against an
oblivious adversary is $\RrandOT=4/3$.
\end{theorem}

When the multiset of values is revealed in advance, the algorithmic core is
the prefix of the shortest processing times that maximizes completions per
unit of testing and processing work.  Testing jobs in a uniformly random
order gives an exact pair identity
whose nonnegative terms prove the factor $4/3$.  The universal algorithm
learns an approximately separated prefix from a sublinear sample; the
threshold and the safe test-all schedule described below protect it from
unbounded values missed by the sample.  The matching lower
bound is the zero--two lower bound at the finite
raw-execution cap $u=\uRand{3}$.  We first prove the uniform upper bound and
then give the short transfer.

\subsubsection{Normalization of the offline pair identity}

For a distribution $D$ on $[0,\infty)$, define the leading offline
coefficient
\begin{equation}
 \Omega_{\mathsf{OT}}(D)
 \defeq\frac12+\SPT(D).
 \label{eq:rand-omega-ot}
\end{equation}
For a fixed input $p$, the finite SPT identity gives
\begin{equation}
 \Omega_{\mathsf{OT}}(D[p])
 =\frac12+\frac{\SPT([n])}{n^2}
  -\frac1{2n^2}\sum_i p_i,
 \qquad
 \Delta_{\mathsf{OT}}(p)=\frac{n+\sum_i p_i}{2}.
 \label{eq:rand-KO}
\end{equation}
Combining \cref{eq:unified-opt-symmetric} with the finite SPT identity, with
$\lambda_i=1+p_i$, gives
\begin{equation}
 \OPT(p)=n^2\Omega_{\mathsf{OT}}(D[p])+\Delta_{\mathsf{OT}}(p).
 \label{eq:rand-opt-normalized}
\end{equation}
The linear correction $\Delta_{\mathsf{OT}}(p)$ will be kept rather than
hidden: this makes the analysis valid for unbounded processing times.

\subsubsection{Stationary prefixes}

For a fixed input $p$, let $E\subseteq[n]$ contain every zero job and suppose
that every processing time in $E$ is at most every processing time outside
$E$.  Put $T=[n]\setminus E$ and
\begin{equation}
 e=|E|,\qquad a=\frac en,\qquad h=1-a,\qquad
 \ell_X=\frac1n\sum_{i\in X}p_i\quad(X\subseteq[n]).
 \label{eq:stationary-prefix-notation}
\end{equation}
We call $E$ a \emph{processing-time prefix}.  Since every job in $E$ is no
longer than every job in $T$, the SPT pair formula gives
\begin{equation}
 \SPT([n])=\SPT(E)+n^2h\ell_E+\SPT(T).
 \label{eq:stationary-spt-split}
\end{equation}

Consider the stationary algorithm that tests the jobs in a uniformly random
order, processes each positive $E$-job immediately, and after the last test
processes $T$ in SPT order.  Here \emph{stationary} means that $E$ and this
immediate-or-deferred rule are fixed before the random order is exposed.
We write its leading coefficient as
\begin{equation}
 \operatorname{STAT}(E)
 \defeq(1+\ell_E)\left(1-\frac a2\right)
   +\frac{\SPT(T)}{n^2}-\frac{\ell_T}{2n}.
 \label{eq:stationary-leading-coefficient}
\end{equation}
This notation has a direct exact interpretation.  Let
$W=n+\sum_{i\in E}p_i$ and $W_E=e+\sum_{i\in E}p_i$.  Every test or immediate
processing block precedes a fixed $E$-job with probability $1/2$, so the
$E$-jobs contribute $(eW+W_E)/2$ in expectation.  Every $T$-job then sees
the common offset $W$, and the final phase contributes
$(n-e)W+\SPT(T)$.  Adding these quantities gives
\begin{equation}
 \EE\ALG_{\Stat}(p,E)
 =n^2\operatorname{STAT}(E)+\frac{e+\sum_i p_i}{2}.
 \label{eq:stationary-prefix-cost}
\end{equation}

For a nonempty prefix define its completion density by
\begin{equation}
 \kappa(E)=\frac{a}{1+\ell_E}.
 \label{eq:rand-density}
\end{equation}
Choose a maximum-density prefix and put
\begin{equation}
 \tau=\frac{1+\ell_E}{a}=\frac1{\kappa(E)}.
 \label{eq:rand-tau}
\end{equation}
All zeros belong to every maximizer: adding a zero increases the numerator
without increasing the denominator.
For an occurrence of length $p$ outside $E$, adding it does not decrease
$a/(1+\ell_E)$ exactly when $p\le(1+\ell_E)/a=\tau$; the analogous removal
calculation applies to an occurrence inside $E$.  We may therefore choose a
maximizer with
\begin{equation}
 \tau=\tau_{D[p]},\qquad
 p_i\le\tau\ (i\in E),\qquad p_i\ge\tau\ (i\notin E),
 \label{eq:rand-density-threshold}
\end{equation}
where tied occurrences may be assigned to either side.

The next identity is the algebraic source of the constant $4/3$.

\begin{lemma}[Exact identity and robust bound]
\label{lem:rand-four-thirds-certificate}
For every ordered prefix split with $a>0$ and
$\tau=(1+\ell_E)/a$,
\begin{equation}
 \begin{aligned}
  &2\tau\bigl(
     4\Omega_{\mathsf{OT}}(D[p])-3\operatorname{STAT}(E)\bigr)\\
  &\quad=(\tau-2)^2
   +4\left[\tau\left(
      \frac{2\SPT(E)}{n^2}-\frac{\ell_E}{n}\right)-\ell_E^2\right]\\
  &\qquad
   +\tau\left(
      \frac{2\SPT(T)}{n^2}-\frac{\ell_T}{n}-\tau h^2\right).
 \end{aligned}
 \label{eq:rand-certificate}
\end{equation}
Consequently:
\begin{enumerate}[label=\textup{(\roman*)}]
\item if values in $E$ are at most $\tau$ and values in $T$ are at least
      $\tau$, then $\operatorname{STAT}(E)\le4\Omega_{\mathsf{OT}}(D[p])/3$;
\item if values in $E$ are at most $\tau+s$ and values in $T$ are at least
      $\tau-s$, where $s\ge0$, then
      \begin{equation}
       \operatorname{STAT}(E)\le\frac43\Omega_{\mathsf{OT}}(D[p])+\frac23s.
       \label{eq:rand-robust-certificate}
      \end{equation}
\end{enumerate}
\end{lemma}

\begin{proof}
Expansion using $\ell_E=a\tau-1$, $h=1-a$, and
\cref{eq:stationary-spt-split} gives \cref{eq:rand-certificate}.  Under exact
separation, $xy\le\tau\min\{x,y\}$ for $x,y\in E$, hence
\[
 \ell_E^2\le\tau\left(
  \frac{2\SPT(E)}{n^2}-\frac{\ell_E}{n}\right).
\]
For $x,y\in T$, one has $\min\{x,y\}\ge\tau$, hence
\[
 \frac{2\SPT(T)}{n^2}-\frac{\ell_T}{n}\ge\tau h^2.
\]
This proves (i).

Under the relaxed inequalities,
\[
 \tau\left(\frac{2\SPT(E)}{n^2}-\frac{\ell_E}{n}\right)-\ell_E^2
 \ge-s\left(\frac{2\SPT(E)}{n^2}-\frac{\ell_E}{n}\right),
\]
and
\[
 \frac{2\SPT(T)}{n^2}-\frac{\ell_T}{n}-\tau h^2\ge-sh^2.
\]
Moreover
$2\SPT(E)/n^2-\ell_E/n\le a\ell_E\le a^2\tau$.
Substitution into
\cref{eq:rand-certificate} yields
\[
 \operatorname{STAT}(E)
 \le\frac43\Omega_{\mathsf{OT}}(D[p])+\frac23sa^2+\frac16sh^2
 \le\frac43\Omega_{\mathsf{OT}}(D[p])+\frac23s,
\]
because $(2/3)a^2+(1/6)(1-a)^2\le2/3$ on $[0,1]$.
\end{proof}

Combining the exact identity with
\cref{eq:stationary-prefix-cost,eq:rand-opt-normalized} already gives a finite
statement useful below.

\begin{corollary}[Announced multiset]
\label{cor:rand-announced-four-thirds}
If the multiset is announced, the stationary algorithm for a maximum-density
prefix satisfies
\[
 \EE\ALG_{\Stat}(p,E)\le\frac43\cdot\OPT(p).
\]
\end{corollary}

Indeed, the correction in \cref{eq:stationary-prefix-cost} is at most
$\Delta_{\mathsf{OT}}(p)$.
Equality in the quadratic coefficient of \cref{eq:rand-certificate} forces
$\tau=2$, zero mass $1/2$, and all positive mass at two.  This predicts the
lower-bound instance before any probabilistic argument is used.

We also need an estimate whose error is independent of the values outside
the early set.  Direct expansion gives
\begin{equation}
 \operatorname{STAT}(E)-\Omega_{\mathsf{OT}}(D[p])
 =\frac12\left(
 h+a\ell_E-\frac{2\SPT(E)}{n^2}+\frac{\ell_E}{n}\right).
 \label{eq:rand-crude-identity}
\end{equation}
If all values in $E$ are at most $L\ge1$, then $\ell_E\le aL$ and therefore
\begin{equation}
 \operatorname{STAT}(E)\le \Omega_{\mathsf{OT}}(D[p])+\frac L2.
 \label{eq:rand-bounded-early}
\end{equation}

\subsubsection{Learning the prefix from a sublinear sample}

For the main algorithm, set the slowly growing cutoff
\begin{equation}
 B=B_n=32+n^{1/20}.
 \label{eq:rand-cutoff}
\end{equation}
For an integer $d$, let $\eta=B/d$.  Keep zero as a separate category,
partition $(0,B]$ into the consecutive bins
$((b-1)\eta,b\eta]$, and add one overflow category for values above $B$.
A finite positive value is rounded upward to the endpoint $q_b=b\eta$.

Let $S$ be a uniformly random subset of $k$ labels.  If $\mu_b$ and
$\widehat\mu_b$ are the population and sample fractions of category $b$, put
\begin{equation}
 \Delta=\sum_b|\widehat\mu_b-\mu_b|.
 \label{eq:rand-hist-error}
\end{equation}

\begin{lemma}[Histogram error]
\label{lem:rand-histogram}
For $1\le k<n$,
\[
 \EE\Delta\le\sqrt{\frac{d+2}{k}}.
\]
\end{lemma}

\begin{proof}
This is \cref{eq:without-replacement-histogram}, with the zero bin and the
overflow bin included among the $d+2$ categories.
\end{proof}

For a finite category prefix $J$, define
\[
 \widehat a(J)=\sum_{b\in J}\widehat\mu_b,\qquad
 \widehat\ell(J)=\sum_{b\in J}q_b\widehat\mu_b,\qquad
 \widehat\kappa(J)=\frac{\widehat a(J)}{1+\widehat\ell(J)}.
\]
Let $\widehat\kappa$ be the maximum and
$\widehat\tau=1/\widehat\kappa$.  If the sample has no mass in any finite
category (equivalently, every sampled value is in the overflow category), set
\begin{equation}
 \widehat\kappa=0,
 \qquad \widehat\tau=+\infty.
 \label{eq:rand-empty-finite-sample}
\end{equation}
Otherwise, starting with a maximizing category prefix, also include every
finite category whose endpoint is at most $\widehat\tau$.  Thus the
algorithm uses
\begin{equation}
 \widehat J=\{\text{finite categories }b:q_b\le\widehat\tau\}.
 \label{eq:rand-threshold-closure}
\end{equation}

\begin{lemma}[Closing a maximum-density prefix at its threshold]
\label{lem:rand-threshold-closure}
The prefix $\widehat J$ from \cref{eq:rand-threshold-closure} has the same
maximum sample density.  Every category in it has endpoint at most
$\widehat\tau$, and
every finite category outside it has endpoint greater than
$\widehat\tau$, including categories with zero sample count.
\end{lemma}

\begin{proof}
For every positive-mass sample category, adding it preserves or improves the
density exactly when its endpoint is at most $\widehat\tau$; removing it
preserves or improves the density under the reverse inequality.  Empty
categories do not change the sample statistics, and adding mass at equality
preserves the identity
$1+\widehat\ell=\widehat a\widehat\tau$.  Adding all categories up to the
threshold therefore preserves optimality and makes the two pointwise
assertions immediate.
\end{proof}

For $n\ge12^4$, set
\begin{equation}
 k=\lfloor n^{3/4}\rfloor,\qquad
 d=\lfloor n^{1/4}\rfloor,\qquad \eta=B/d.
 \label{eq:rand-parameters}
\end{equation}
The algorithm draws a uniform permutation and first tests its first $k$ jobs,
leaving positive sample jobs pending.  It computes $\widehat J$ according to
\cref{eq:rand-threshold-closure}.
If
\begin{equation}
 \widehat\kappa<\frac1B,
 \label{eq:rand-fallback-test}
\end{equation}
it enters \emph{fallback mode}: test every remaining job and then process all
positive jobs in SPT order.  Otherwise it enters \emph{learned mode}: first
process the positive sampled jobs whose categories lie in $\widehat J$;
then test the remaining jobs in random order and immediately process each
positive job whose category lies in $\widehat J$; finally process all other
positive tested jobs in SPT order.  The overflow category never belongs to
$\widehat J$.
For $n<12^4$, define the algorithm to use fallback mode immediately.  This
small-$n$ branch makes $k<n$ whenever the histogram lemma is invoked; its
finite cost is absorbed by the universal constant in the asymptotic bound.

\subsubsection{Analysis on good and bad samples}

Put
\begin{equation}
 \gamma=\frac1{B(B+1)},\qquad G=\{\Delta\le\gamma\}.
 \label{eq:rand-good-event}
\end{equation}
Markov's inequality, recalled in \cref{sec:probabilistic-tools}, and
\cref{lem:rand-histogram} give
\begin{equation}
 \Prob(G^c)\le B(B+1)\sqrt{\frac{d+2}{k}}.
 \label{eq:rand-bad-probability}
\end{equation}

Assume that $G$ holds and the algorithm uses the learned schedule.  Let $E$
contain all jobs whose categories belong to $\widehat J$, and let
$a,\ell_E,\tau=(1+\ell_E)/a$ be its true statistics.
Since $\widehat\tau\le B$ and
$\widehat a=(1+\widehat\ell)/\widehat\tau\ge1/B$, we have
$a\ge1/(B+1)$.  For
this fixed category set,
\[
 |a-\widehat a|\le\Delta,\qquad
 |\ell_E-\widehat\ell|\le B\Delta+\eta.
\]
Consequently
\begin{equation}
 |\tau-\widehat\tau|
 \le2B(B+1)\Delta+(B+1)\eta.
 \label{eq:rand-threshold-transfer}
\end{equation}
Every value in $E$ is at most $\widehat\tau$, while every value outside $E$
is greater than $\widehat\tau-\eta$.  Thus the robust bound applies
with
\begin{equation}
 s=2B(B+1)\Delta+(B+2)\eta.
 \label{eq:rand-robust-s}
\end{equation}

The algorithm tests the sample before all other jobs, so its order is not
exactly the stationary schedule in which all jobs share one uniform order.
The difference is only lower order.  Conditional on the unordered sample,
put $e_S=|E\cap S|$, $e_T=|E\setminus S|$, and
\[
 W_S=k+\sum_{i\in E\cap S}p_i,\qquad
 W_T=n-k+\sum_{i\in E\setminus S}p_i.
\]
Comparing the two independent random internal orders with one fully random
order, the difference in the cost contributed by jobs in $E$ is
\[
 \frac{e_TW_S-e_SW_T}{2}
 \le\frac12nk(1+B).
\]
Postponing the positive jobs in $E\cap S$ until all sample tests finish adds
at most $k^2(1+B)$.  Jobs outside $E$ are processed identically at the end
because every zero belongs to $E$.  Therefore, on event $G$ when the learned
schedule is used,
\begin{equation}
 \EE[\ALG\mid S]
 \le\frac43\cdot\OPT
 +n^2\!\left[\frac43B(B+1)\Delta+\frac23(B+2)\eta\right]
 +(1+B)(nk/2+k^2).
 \label{eq:rand-good-learned}
\end{equation}

On event $G$ when the fallback schedule is used, comparison of sample and
population densities shows that the reciprocal of the true maximum density
satisfies
$\tau_{D[p]}\ge B/4\ge8$.
Here are the constants, including the grid and floor effects.  If, to the
contrary, $\tau_{D[p]}<B/4$, let $E$ be a true maximum-density threshold
prefix, of mass $a$, and let $J$ contain all rounded finite categories met
by $E$.  Upward rounding changes every selected value by at most $\eta$;
closing $J$ can add only endpoints no larger than
$\tau_{D[p]}+\eta$.  Hence the population reciprocal density of $J$ is at
most $\tau_{D[p]}+\eta$, while
\begin{equation}
 a=\frac{1+\ell_E}{\tau_{D[p]}}>\frac4B.
 \label{eq:rand-fallback-prefix-mass}
\end{equation}
For $n\ge12^4$, the floor in $d=\lfloor n^{1/4}\rfloor$ gives
$d\ge12$ and hence $\eta\le B/12$.  On $G$, the sample mass of $J$ is at
least $a-\gamma$, and its sample work numerator is at most its population
numerator plus $B\gamma$.  Consequently its sample reciprocal density is
at most
\begin{align}
 \frac{a(B/4+\eta)+B\gamma}{a-\gamma}
 &\le \frac B3+
   \frac{(4B/3)\gamma}{4/B-\gamma}
 \notag\\
 &=\frac B3+\frac{4B}{3(4B+3)}<B,
 \label{eq:rand-fallback-mesh-bridge}
\end{align}
where $\gamma=1/[B(B+1)]$.  Thus the sample maximum density exceeds
$1/B$, contradicting \cref{eq:rand-fallback-test}.  This proves
$\tau_{D[p]}\ge B/4$ with the mesh, floors, and constants accounted for.
The exact threshold inequalities imply
\begin{equation}
 \frac{2\SPT([n])}{n^2}-\frac1{n^2}\sum_i p_i
 \ge\frac{(\tau_{D[p]}-1)^2}{\tau_{D[p]}},
 \qquad
 \Omega_{\mathsf{OT}}(D[p])
 \ge\frac{\tau_{D[p]}-1+1/\tau_{D[p]}}{2}\ge\frac{57}{16}.
 \label{eq:rand-fallback-opt-large}
\end{equation}
For every test order, test-all-then-SPT satisfies
\begin{equation}
 \ALG_{\rm fb}\le\OPT+\frac{n(n-1)}2.
 \label{eq:rand-fallback-crude}
\end{equation}
Indeed, if there are $z$ zeros, moving their completions from the first $z$
SPT positions to their test positions increases their sum of completion
times by at most $zn-z(z+1)/2$.  Delaying the remaining $n-z$ positive jobs
until all tests are finished adds at most $(n-z)(n-z-1)/2$; the two bounds
sum to $n(n-1)/2$.  Combining
\cref{eq:rand-fallback-opt-large,eq:rand-fallback-crude} gives
$\ALG_{\rm fb}<\frac43\cdot\OPT$.

Because processing times are unbounded, the estimate on $G^c$ must control
excess over $\frac43\cdot\OPT$, not raw cost.  Under the learned schedule
every value in $E$ is at most $B$, so \cref{eq:rand-bounded-early} bounds the
stationary excess by $(B/2)n^2$; under the fallback schedule
\cref{eq:rand-fallback-crude} bounds it by
$n^2/2$.  Averaging \cref{eq:rand-good-learned} and these bounds on $G^c$,
and using \cref{eq:rand-bad-probability}, yields the uniform estimate
\begin{equation}
 \EE\ALG
 \le\frac43\cdot\OPT
 +n^2\!\left[
   \frac{B(B+1)(3B+8)}6\sqrt{\frac{d+2}{k}}
   +\frac{2B(B+2)}{3d}\right]
 +(1+B)(nk/2+k^2).
 \label{eq:rand-explicit-upper}
\end{equation}
For $n\ge12^4$, the floors in \cref{eq:rand-parameters} give
$\sqrt{(d+2)/k}\le2n^{-1/4}$, $1/d\le2n^{-1/4}$, and
$k\le n^{3/4}$.  Since $B_n=O(n^{1/20})$, the three error terms in
\cref{eq:rand-explicit-upper} are respectively
$O(n^{19/10})$, $O(n^{37/20})$, and $O(n^{9/5})$.  Enlarging one universal
constant covers smaller $n$ and proves the first inequality in
\cref{eq:rand-obligatory-upper}.  The bound $\OPT\ge n^2/2$ gives its
multiplicative form.

For completeness, the fixed-cutoff companion uses $B=32$ and the more
conservative fallback criterion $\widehat\kappa<2/B$.  Substituting these
fixed-cutoff choices into the error calculation gives
\[
 \EE\ALG
 \le\frac43\cdot\OPT
 +n^2\!\left[9472\sqrt{\frac{d+2}{k}}+\frac{704}{d}\right]
 +17(nk/2+k^2).
\]
The estimates $k\le n^{3/4}$, $1/d\le2n^{-1/4}$, and
$\sqrt{(d+2)/k}\le2n^{-1/4}$ give
$\EE\ALG\le\frac43\cdot\OPT+20378n^{7/4}$, with the same deliberately loose
constant covering smaller $n$.  This proves
\cref{eq:rand-obligatory-fixed-upper}.

\subsubsection{Transfer of the zero--two lower bound}

Fix the finite raw-execution cap $u=\uRand{3}$.  Any randomized obligatory
algorithm can be viewed as a revealing-optimization algorithm at this value
of $u$ that simply never uses raw execution.  Apply
\cref{lem:rcu-zero-two-lower}.  Its fixed hard
input has only values zero and two, so every revealing effective length is
\[
 \min\{\uRand{3},1+p_j\}=1+p_j.
\]
The online transcript and the offline optimum are therefore exactly the
same in the revealing model with raw-execution cap $u=\uRand{3}$ and in the
obligatory model.  For even $n$, the common
offline value is
\[
 \OPT=\frac34n^2+n,
\]
whereas the selected labeling has expected online cost
$n^2-o(n^2)$.  Its ratio tends to $4/3$, which proves
\cref{eq:rand-obligatory-lower}.  Together with
\cref{eq:rand-obligatory-upper}, this completes the proof of
\cref{thm:rand-obligatory}.

\section{Blind optimization: speedup without revelation}
\label{sec:blind-optimization}

We now separate the two benefits that were combined in the
revealing-optimization model.  Raw execution still takes the known time $u$, and one
unit of optimization still reduces the processing time to $p_i\in[0,u]$.
The difference is that optimization does not reveal $p_i$.  An optimized job
must be started blindly and its duration becomes known only when it finishes.
We call this the \emph{blind-optimization} model $\mathsf{BO}(u)$.

We study deterministic worst-case algorithms, and randomized algorithms against
an oblivious adversary.  The deterministic optimum has a particularly
simple form: run everything raw up to $u=2$, and optimize everything after
that point.  To prove that no adaptive mixture improves these two algorithms,
the adversary assigns value zero to a raw job and value $u$ to an optimized
job, and stops after a prescribed fraction has been optimized.  Randomization
is still essential for learning the input:
a private random sample is representative after the initial shuffle,
whereas an adversary can make a deterministic exploration history atypical.

This small change removes the threshold decision entirely.  It also shows
that the bounded randomized ratio in the revealing model comes from the
revealed information.  The two curves are compared in
\cref{fig:blind-optimization-curves}.

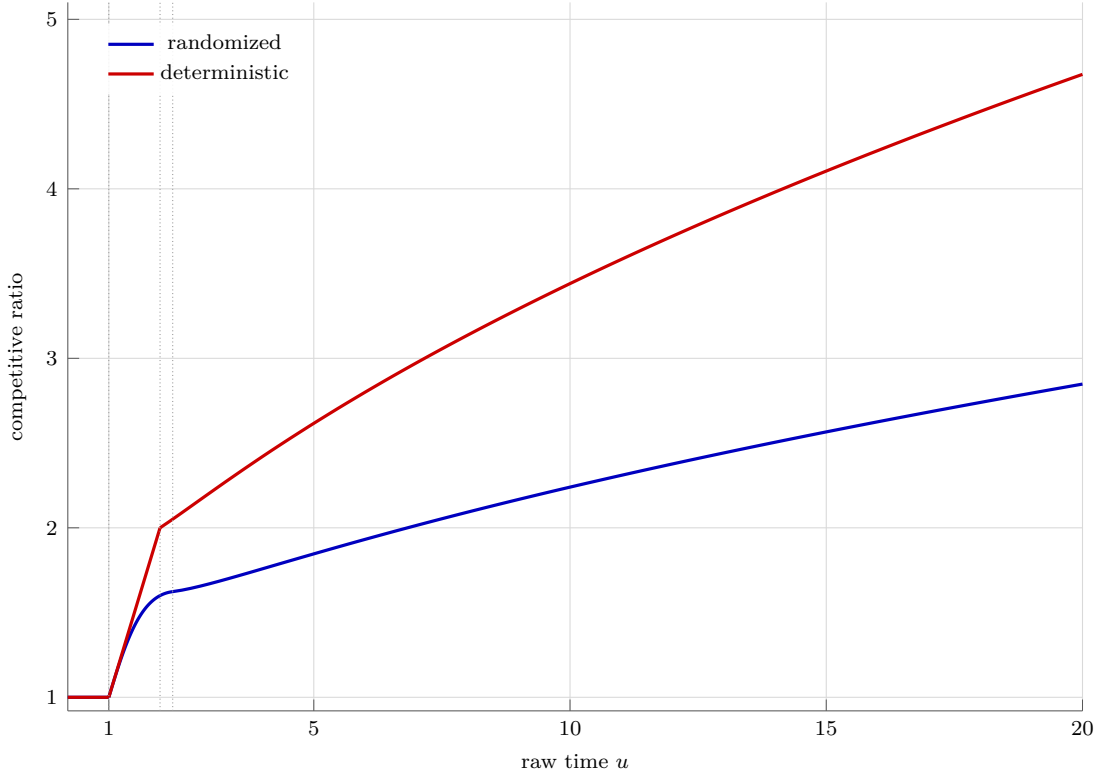
\begin{figure}[htbp]
\centering
\begingroup
\tracinglostchars=0
\begin{tikzpicture}
\begin{axis}[
  width=\linewidth,
  height=0.73\linewidth,
  xmin=0.2,
  xmax=20,
  ymin=0.92,
  ymax=5.1,
  xlabel={raw time $u$},
  ylabel={competitive ratio},
  xtick={1,5,10,15,20},
  ytick={1,2,3,4,5},
  tick label style={font=\scriptsize},
  label style={font=\scriptsize},
  grid=major,
  grid style={black!15},
  axis line style={black!65},
  tick style={black!65},
  axis x line*=bottom,
  axis y line*=left,
  samples=100,
  no markers,
  legend style={
    font=\scriptsize,
    draw=none,
    fill=white,
    fill opacity=0.9,
    text opacity=1,
    at={(0.03,0.97)},
    anchor=north west,
  },
]
\addplot[blue!75!black,very thick,domain=0.2:1] {1};
\addlegendentry{randomized}
\addplot[forget plot,blue!75!black,very thick,domain=1:2.24697960372]
  {x^3/(x^2+(x-1)^3)};
\addplot[forget plot,blue!75!black,very thick,domain=2.24697960372:20]
  {0.5*(1+sqrt((x^2+x-1)/(x-1)))};

\addplot[red!80!black,very thick,domain=0.2:1] {1};
\addlegendentry{deterministic}
\addplot[red!80!black,very thick,domain=1:2] {x};
\addplot[red!80!black,very thick,domain=2:20]
  {(sqrt(4*x^3-4*x+1)-1)/(2*(x-1))};

\addplot[black!40,densely dotted,thin] coordinates {(1,0.92) (1,5.1)};
\addplot[black!40,densely dotted,thin] coordinates {(2,0.92) (2,5.1)};
\addplot[black!40,densely dotted,thin]
  coordinates {(2.24697960372,0.92) (2.24697960372,5.1)};
\end{axis}
\end{tikzpicture}
\endgroup
\caption{Exact deterministic and randomized blind-optimization curves.
Unlike their revealing counterparts, both diverge with the raw cap $u$.}
\label{fig:blind-optimization-curves}
\end{figure}

The section first proves that deterministic instance optimality is impossible
and determines the exact deterministic worst-case curve.  It then proves two
closely related randomized statements.  The first, \cref{thm:bo-instance},
is the formal version of the instance-optimality theorem
\cref{thm:intro-bo-instance}.  The second, \cref{thm:bo-curve}, is the
formal version of \cref{thm:intro-bo-curve} and maximizes that
instance-specific value relative to the clairvoyant optimum.  Keeping the
results together avoids repeating the same reduction to contiguous
optimization-and-processing blocks.

For an input $p$, put
\[
 \mu_p=\frac1n\sum_{i=1}^n p_i
\]
and define
\begin{equation}
 \PhiBO(p)=\frac12\min\{u,1+\mu_p\}.
 \label{eq:bo-phi}
\end{equation}
The two terms have an immediate meaning.  Running every job raw uses $u$
units per completion.  Optimizing and then blindly processing every job uses
$1+\mu_p$ units per completion on average.

Randomization is necessary for attaining this benchmark on every input.

\begin{theorem}[No deterministic instance-optimal algorithm]
\label{thm:bo-det-not-instance-optimal}
For every unannounced deterministic nonanticipating algorithm $\mathcal A$ and
every even $n$, there are a fixed labeled input $I\in\{0,2\}^n$ and an
announced randomized algorithm $\mathcal A'$ such
that
\begin{equation}
 \ALG_{\mathcal A}(I)
 \ge \EE\ALG_{\mathcal A'}(I)+\frac18n^2.
 \label{eq:bo-det-instance-gap}
\end{equation}
Consequently no deterministic algorithm is asymptotically instance-optimal in
$\mathsf{BO}(2)$.
\end{theorem}

\begin{proof}
Run $\mathcal A$ against the following deterministic revelation rule.  Give
a job value zero if $\mathcal A$ runs it raw, and value two if $\mathcal A$
optimizes it.  Optimization reveals nothing, and the value two is seen only
when the optimized job is subsequently processed.  Since $\mathcal A$ is
deterministic, recording the resulting value of every label produces a fixed
labeled input with the same transcript.

Let $m$ be the number of optimized jobs on this input.  Every raw job consumes
two units before completion, whereas every optimized job consumes three.
Irrespective of interleaving, the completion times are therefore bounded
below by the shortest-first schedule of $n-m$ blocks of length two and $m$
blocks of length three.  Hence
\begin{equation}
 \ALG_{\mathcal A}(I)
 \ge n^2+n+\frac{m^2+m}{2}.
 \label{eq:bo-det-instance-online}
\end{equation}

The multiset has $n-m$ zeros and $m$ twos.  If $m\le n/2$, let
$\mathcal A'$ privately permute the jobs and optimize and process them one at
a time in that order.  Every position sees the average block length
$(n+2m)/n$, so
\begin{equation}
 \EE\ALG_{\mathcal A'}(I)=\frac{(n+1)(n+2m)}2.
 \label{eq:bo-det-instance-comparator-opt}
\end{equation}
Subtracting this from \cref{eq:bo-det-instance-online} gives
$(n-m)(n-m+1)/2\ge n^2/8$.

If $m>n/2$, let $\mathcal A'$ run every job raw.  Its cost is $n(n+1)$, and
the difference in \cref{eq:bo-det-instance-online} is
$(m^2+m)/2>n^2/8$.  This proves \cref{eq:bo-det-instance-gap}.
\end{proof}

An optimization operation reveals no information.  We will repeatedly use
the following normal form.

\begin{lemma}[Contiguous-block normal form for optimized jobs]
\label{lem:bo-normal-form}
For every nonanticipating algorithm $\mathcal A$ in $\mathsf{BO}(u)$ there is a
nonanticipating algorithm $\mathcal A^{\rm nf}$ which completes every optimized
job as one contiguous block of length $1+p_i$ and satisfies
\begin{equation}
 \ALG_{\mathcal A^{\rm nf}}(I)\le \ALG_{\mathcal A}(I)
 \label{eq:bo-normal-form}
\end{equation}
on every input $I$.  If $\mathcal A$ is deterministic, then so is
$\mathcal A^{\rm nf}$.
\end{lemma}

\begin{proof}
Maintain a virtual execution of $\mathcal A$.  When the virtual algorithm
optimizes a job, record that fact and advance the virtual clock by one unit,
but perform no physical work.  Immediately before the virtual algorithm later
processes that job, perform its postponed optimization physically and then
process it.  Observations that occur earlier in physical time are buffered
until the corresponding virtual completion.  Thus the virtual algorithm sees
exactly its original history and makes the same decisions.  Moving a unit of
optimization to the right leaves the optimized job and every later job no
later than before, while jobs completed in between move one unit earlier.
Applying this to every separated optimization proves
\eqref{eq:bo-normal-form}.
\end{proof}

\subsection{The deterministic worst-case curve}

For $u>1$ define
\begin{equation}
 G_{\mathsf{BO}}(u)
 \defeq \frac{\sqrt{4u^3-4u+1}-1}{2(u-1)}.
 \label{eq:bo-det-optall-ratio}
\end{equation}
The two deterministic algorithms needed below are \textsc{Raw}, which runs
every job for time $u$, and \textsc{OptimizeAll}, which fixes an arbitrary
label order and completes every job by one contiguous optimization-and-processing
block of length $1+p_i$.

\begin{theorem}[Exact deterministic blind-optimization curve]
\label{thm:bo-det-curve}
The optimal deterministic size-asymptotic competitive ratio in
$\mathsf{BO}(u)$ is
\begin{equation}
 \RdetBO(u)=
 \begin{cases}
  1,&0<u\le1,\\[2pt]
  u,&1\le u\le2,\\[5pt]
  \displaystyle G_{\mathsf{BO}}(u),&u\ge2.
 \end{cases}
 \label{eq:bo-det-curve}
\end{equation}
The algorithm \textnormal{\textsc{Raw}} is optimal for $u\le2$, and
\textnormal{\textsc{OptimizeAll}} is optimal for $u\ge2$; both are optimal
at $u=2$.
In particular,
\begin{equation}
 \RdetBO(u)=\sqrt u+o(1)
 \qquad (u\to\infty).
 \label{eq:bo-det-asymptotic}
\end{equation}
\end{theorem}

We first prove the adaptive lower bound.  Its stopping rule is the direct
analogue of \cref{lem:opt:hidden-stopping} for optimization without
revelation.

\begin{lemma}[Hidden binary stopping for optimized jobs]
\label{lem:bo-det-hidden-stopping}
Fix $u>1$, a target $R\le u$, and $\alpha\in(0,1)$.  Against every
deterministic algorithm there is a fixed binary input with $p_j\in\{0,u\}$ on
which
\begin{equation}
 \ALG-\,R\cdot\OPT
 \ge \frac{n^2}{2}\left((u-R)(1-\sigma^2)
       +\sigma^2 f^{\mathsf{BO}}_{u,R}(y)\right)-O_u(n),
 \label{eq:bo-det-stopping-decomposition}
\end{equation}
where $\sigma\in[0,1]$, $0\le y-\alpha<1/(\sigma n)$ when
$\sigma>0$, and
\begin{equation}
 f^{\mathsf{BO}}_{u,R}(y)
 =1-R+2uy-\bigl(u+R(u-1)\bigr)y^2.
 \label{eq:bo-det-stopping-f}
\end{equation}
Consequently, if $f^{\mathsf{BO}}_{u,R}(\alpha)\ge0$, then every
deterministic algorithm has size-asymptotic ratio at least $R$.
\end{lemma}

\begin{proof}
By \cref{lem:bo-normal-form}, it is enough to lower-bound an algorithm that
optimizes and processes each such job in one contiguous block.  Until the
stopping condition
below is met, give value zero to every job completed raw and value $u$ to
every optimized job.
If $v$ raw jobs and $L$ optimized jobs with $p=u$ have been completed, stop at the
first completion for which $L>0$ and
\begin{equation}
 L\ge\alpha(n-v).
 \label{eq:bo-det-stopping-line}
\end{equation}
Give value zero to all remaining jobs.  The answers are deterministic
functions of the transcript and therefore define one fixed labeled binary
input that reproduces the same execution.

If the algorithm fails to complete all jobs, its ratio is infinite.  Otherwise,
if the line is never crossed, then $L=0$ after all jobs are completed:
otherwise $n-v=L$ at the end and $L\ge\alpha L$.  Thus the algorithm ran
every job raw on the all-zero input and has ratio $u\ge R$.

Suppose the line is crossed.  Put $x=n-v-L$.  For a lower bound, discard
idle time, put the already completed raw jobs before the already completed
optimized jobs with $p=u$, reveal the future, and let all $x$ remaining zero
jobs be optimized.  These exchanges only favor the algorithm: among the
operations already completed, shortest first is optimal, while no untouched
future job could
have completed before the crossing.  The resulting exact cost and the
clairvoyant optimum are
\begin{align}
 A_{\rm ex}={}&
 u\left(vn-\frac{v(v-1)}2\right)
 +(u+1)\left(L(n-v)-\frac{L(L-1)}2\right)
 +\frac{x(x+1)}2,
 \label{eq:bo-det-stopping-Aexact}\\
 O_{\rm ex}={}&
 \frac{n(n+1)}2+(u-1)\frac{L(L+1)}2.
 \label{eq:bo-det-stopping-Oexact}
\end{align}
Here the offline optimum first optimizes and processes the $n-L$ zero jobs,
using one unit per job, and then runs the $L$ jobs with $p=u$ raw.

Set
\begin{equation}
 \sigma=\frac{n-v}{n},
 \qquad y=\frac{L}{n-v}.
 \label{eq:bo-det-stopping-coordinates}
\end{equation}
Twice the quadratic coefficients of
\cref{eq:bo-det-stopping-Aexact,eq:bo-det-stopping-Oexact} are
\begin{align*}
 \mathcal A
 &=u(1-\sigma^2)+\sigma^2(1+2uy-uy^2),\\
 \mathcal O
 &=(1-\sigma^2)+\sigma^2(1+(u-1)y^2).
\end{align*}
Subtracting $R\cdot\mathcal O$ from $\mathcal A$ gives exactly the bracket in
\eqref{eq:bo-det-stopping-decomposition}.  Immediately before the crossing,
$L-\alpha(n-v)<0$.  An optimized completion increases this expression by
one and a raw completion by $\alpha$, so
\begin{equation}
 0\le y-\alpha<\frac1{n-v}=\frac1{\sigma n}.
 \label{eq:bo-det-stopping-overshoot}
\end{equation}
The quadratic $f^{\mathsf{BO}}_{u,R}$ is uniformly Lipschitz on $[0,1]$ for
fixed $u$ and $R\le u$.  Hence
$\sigma^2f^{\mathsf{BO}}_{u,R}(y)
\ge\sigma^2f^{\mathsf{BO}}_{u,R}(\alpha)-O_u(1/n)$.
The exact terms contributed by individual jobs add another $O_u(n)$ before
normalization.
If $u-R$ and $f^{\mathsf{BO}}_{u,R}(\alpha)$ are nonnegative, this proves
the claimed lower bound.
\end{proof}

\begin{proof}[Proof of \Cref{thm:bo-det-curve}]
If $u\le1$, raw execution is an offline-optimal action for every job, so the
ratio is one.  If $1<u\le2$, use
\begin{equation}
 R=u,
 \qquad \alpha=\frac1u.
\end{equation}
Then
$f^{\mathsf{BO}}_{u,u}(1/u)=2-u\ge0$, and
\cref{lem:bo-det-hidden-stopping} gives the lower bound $u$.  Conversely,
every offline effective length is at least one, whereas \textsc{Raw} has
cost $u n(n+1)/2$, so its ratio is at most $u$.

Now let $u\ge2$ and put $R=G_{\mathsf{BO}}(u)$.  The definition
\eqref{eq:bo-det-optall-ratio} is equivalently the positive solution of
\begin{equation}
 (R-1)\bigl(u+R(u-1)\bigr)=u^2.
 \label{eq:bo-det-tangency}
\end{equation}
Moreover $R\le u$, with equality only at $u=2$.  Choose
\begin{equation}
 \alpha=\frac{u}{u+R(u-1)}.
 \label{eq:bo-det-alpha}
\end{equation}
This is the maximizer of \cref{eq:bo-det-stopping-f}, and
\eqref{eq:bo-det-tangency} gives
\begin{equation}
 f^{\mathsf{BO}}_{u,R}(\alpha)
 =1-R+\frac{u^2}{u+R(u-1)}=0.
\end{equation}
Thus \cref{lem:bo-det-hidden-stopping} proves the lower bound $R$.

It remains to analyze \textsc{OptimizeAll}.  Put $a(p)=1+p$.  For a fixed
multiset, its worst label order has the
$a(p_i)$ in nonincreasing order.  Direct pair accounting gives
\begin{align}
 \ALG_{\textsc{OptimizeAll}}
 &\le \frac12\sum_{i,j}\max\{a(p_i),a(p_j)\}
      +\frac12\sum_i a(p_i),
 \label{eq:bo-det-optall-online}\\
 \OPT
 &=\frac12\sum_{i,j}\min\{u,1+p_i,1+p_j\}
      +\frac12\sum_i\min\{u,1+p_i\}.
 \label{eq:bo-det-optall-offline}
\end{align}
Write $S_D(t)=D((t,u])$ for $0\le t\le u$.  Expressing a nonnegative value
$x$ as $\int_0^\infty\one_{t<x}\,dt$ yields
\begin{align}
 \iint\max\{a(p),a(q)\}\,dD(p)dD(q)
 &=1+\int_0^u(2S_D(t)-S_D(t)^2)\,dt,
 \label{eq:bo-det-optall-layer-online}\\
 \iint\min\{u,1+p,1+q\}\,dD(p)dD(q)
 &=1+\int_0^{u-1}S_D(t)^2\,dt.
 \label{eq:bo-det-optall-layer-offline}
\end{align}
Set
\begin{equation}
 b=\frac1{u-1}\int_0^{u-1}S_D(t)\,dt.
\end{equation}
Cauchy--Schwarz bounds the square integral on $[0,u-1]$ from below by
$(u-1)b^2$.  Since $S_D$ is nonincreasing,
\begin{equation}
 T\defeq\int_{u-1}^uS_D(t)\,dt\le S_D(u-1)\le b.
\end{equation}
The function $h\mapsto2h-h^2$ is increasing and concave on $[0,1]$.
Cauchy--Schwarz on the prefix and Jensen's inequality on the unit tail
therefore give
\begin{align}
 \iint\max\{a(p),a(q)\}\,dD(p)dD(q)
 &\le1+2ub-ub^2,\\
 \iint\min\{u,1+p,1+q\}\,dD(p)dD(q)
 &\ge1+(u-1)b^2.
\end{align}
Consequently their ratio is at most
\begin{equation}
 g_u(b)=\frac{1+2ub-ub^2}{1+(u-1)b^2}.
 \label{eq:bo-det-optall-envelope}
\end{equation}
The derivative of $g_u$ changes sign exactly once on $[0,1]$.  At its
maximum, if $R=g_u(b)$, the stationary equation
$g_u'(b)=0$ says
\begin{equation}
 u(1-b)=R(u-1)b,
 \qquad
 b=\frac{u}{u+R(u-1)}.
\end{equation}
Substitution into $g_u(b)=R$ gives
\eqref{eq:bo-det-tangency}.  That equation has the unique positive solution
$G_{\mathsf{BO}}(u)$; hence
\begin{equation}
 \max_{0\le b\le1}g_u(b)=G_{\mathsf{BO}}(u).
\end{equation}
All survival-function inequalities above are tight for a binary multiset
with a $b$-fraction of values $u$ and the remaining values zero, labeled so
that the optimized jobs with $p=u$ occur first.
Equations \eqref{eq:bo-det-optall-online}--%
\eqref{eq:bo-det-optall-envelope} prove
$\ALG_{\textsc{OptimizeAll}}\le R\cdot\OPT+\,O_u(n)$.
This matches the lower bound and proves \eqref{eq:bo-det-curve}.
The asymptotic formula \eqref{eq:bo-det-asymptotic} follows directly from
\eqref{eq:bo-det-optall-ratio}.
\end{proof}

\begin{theorem}[Blind-optimization instance optimality]
\label{thm:bo-instance}
For every fixed $0<u<\infty$ there exist a randomized nonanticipating algorithm
$\mathcal A_{\mathsf{BO}}$, depending only on $n,u$, and a sequence
$\eps_u(n)\to0$ such that every input $p\in[0,u]^n$ satisfies
\begin{equation}
 \EE\ALG_{\mathcal A_{\mathsf{BO}}}(p)
 \le n^2\PhiBO(p)+\eps_u(n)n^2.
 \label{eq:bo-instance-upper}
\end{equation}
Conversely, every $p\in[0,u]^n$ and every randomized nonanticipating algorithm
$\mathcal A$, even an announced one, satisfy
\begin{equation}
 \EE\ALG_{\mathcal A}(p)
 \ge n^2\PhiBO(p)-\eps_u(n)n^2.
 \label{eq:bo-instance-lower}
\end{equation}
Thus one algorithm is asymptotically optimal, at the $n^2$ scale, for every
input simultaneously.  Unlike $\PhiRO$ and $\PhiBE$, $\PhiBO$
depends only on the empirical mean.
\end{theorem}

The proof is simpler than in the revealing models.  Because optimization
reveals nothing, it can be moved next to the processing of its job.  Every
job is then completed either by raw execution of known length $u$ or by one
contiguous optimization-and-processing operation whose average length is
$1+\mu_p$.  A sampling bound that holds for every adaptive prefix rules out
an adaptive advantage, while a sample estimates the one number $\mu_p$.
We first prove this instance-optimal statement and then maximize its value
relative to the clairvoyant optimum.

\subsection{Proof of randomized instance optimality}

By the input convention, the jobs in $p$ are shuffled uniformly
before the algorithm starts.  The lower half simply says that no announced
algorithm beats $n^2\PhiBO(p)-o_u(n^2)$ in expectation.

By \cref{lem:bo-normal-form}, we may regard every optimized job as one
indivisible operation of hidden length $1+p_i$.  Raw execution has the known
length $u$.  Either action completes exactly one job.

\paragraph{The announced lower bound.}
Fix the private seed of an algorithm and place the announced occurrences
uniformly behind the labels.  We may assume that the algorithm completes every
job.  Treat equal-valued jobs as distinct and reveal their values in the
analysis in first-touch order.  By the principle of deferred decisions, this order is a
uniform permutation, and the decision whether to optimize the next job uses
only the values revealed by earlier optimized jobs.  Running a job raw does
not reveal its value to the algorithm; exposing it
afterwards only for the analysis does not affect the algorithm's decisions.

Fix $\delta>0$ and ignore the last $\delta n$ first touches for the moment.
The predictable-permutation estimate
\cref{lem:predictable-permutation-sampling}, applied to the processing times
and the indicator of the decision to optimize, gives one event of
probability $1-o(1)$ on which every earlier prefix containing $t n$
completed optimized jobs
has total hidden processing work
\begin{equation}
 \sum_{i\text{ optimized in the prefix}}p_i
   =\mu_p t n+o_\delta(n).
 \label{eq:bo-selected-work}
\end{equation}
The event is simultaneous over all prefixes, so a stopping fraction chosen
adaptively by the algorithm may be substituted only after the event has been
fixed.  Editing the last $\delta n$ jobs to either mode changes total
completion cost by $O_u(\delta n^2)$.

For one error bound that applies directly at every $n$, take
$\delta_n=n^{-1/4}$.  The bounded selected-work estimate used in
\cref{sec:common-lower-transfer} has deviation
\[
 O_u\!\left((1+\delta_n^{-1})\sqrt{n\ln(n+2)}\right)
 =O_u\!\left(n^{3/4}\sqrt{\ln(n+2)}\right)
\]
and failure probability $o(1)$, simultaneously over all relevant prefixes.
After division by $n$, both this deviation and the modification of the last
$\delta_n n$ jobs contribute
$O_u(n^{-1/4}\sqrt{\ln(n+2)})$ to the normalized completion-area bound.

Put $a=1+\mu_p$.  If the final optimized fraction is $q$, the best possible
order of the two divisible block types puts the smaller work per completion
first.  When $a\le u$
its area is
\begin{equation}
 F_{a,u}(q)
 =a\left(q-\frac{q^2}{2}\right)
   +\frac{u(1-q)^2}{2},
 \qquad
 F'_{a,u}(q)=(a-u)(1-q).
 \label{eq:bo-fluid-opt-first}
\end{equation}
It is minimized at $q=1$, with value $a/2$.  When $a\ge u$, jobs run raw go
first and the area is
\begin{equation}
 F_{a,u}(q)
 =\frac{u(1-q^2)}2+\frac{a q^2}2
 =\frac u2+\frac{a-u}{2}q^2,
 \label{eq:bo-fluid-raw-first}
\end{equation}
which is minimized at $q=0$, with value $u/2$.  The shortest-first argument
of \cref{lem:divisible-spt-area} proves that no interleaving has smaller area.
Equations
\eqref{eq:bo-selected-work}--\eqref{eq:bo-fluid-raw-first}, with the uniform
choice $\delta=\delta_n$, give
\[
 \frac1{n^2}\EE\ALG_{\mathcal A}(p)
 \ge\frac12\min\{u,1+\mu_p\}-o_u(1).
\]
The estimate is uniform in the deterministic algorithm.  Averaging it over the
algorithm seed proves \eqref{eq:bo-instance-lower}.  Together with the upper
bound below, one may
take $\eps_u(n)=O_u(n^{-1/4}\sqrt{\ln(n+2)})$ after enlarging its constant.

\paragraph{One algorithm for all inputs.}
Take $k=\lceil n^{2/3}\rceil$ uniformly random sample jobs.  Optimize each
sample job and process it immediately, thereby observing its duration, and let
$\widehat\mu$ be their sample mean.  If $1+\widehat\mu<u$, optimize and
immediately process every remaining job in fresh random order.  Otherwise
run every remaining job raw.

The sample contributes at most $O_u(nk)=o_u(n^2)$ beyond the main schedule.
The variance estimate \cref{eq:without-replacement-variance}, after scaling
from $[0,1]$ to $[0,u]$, gives
\[
 \EE|\widehat\mu-\mu_p|=O_u(k^{-1/2}+k/n)=o_u(1).
\]
Choosing the wrong mode can cost at most
$O(n^2)|\widehat\mu-\mu_p|$, while choosing the correct mode has leading
cost $n^2\PhiBO(p)$.  This proves
\eqref{eq:bo-instance-upper}.  In particular, the universal algorithm needs to
learn only one number, not a histogram or a threshold.

\subsection{The exact worst-case curve}

The clairvoyant scheduler knows $p_i$, uses the effective length
$\min\{u,1+p_i\}$ from \cref{eq:model-effective-length}, and orders these
blocks shortest first.  For a distribution $D$ on $[0,u]$, define its
leading offline coefficient by
\begin{equation}
 \Omega_{\mathsf{BO},u}(D)
 \defeq\frac12\iint
  \min\{u,1+p,1+q\}\,dD(p)dD(q).
 \label{eq:bo-offline-coefficient}
\end{equation}
Consequently the instance-specific ratio from
\cref{thm:bo-instance} is
\begin{equation}
 \alpha_{\mathsf{BO},u}(D)
 =\frac{\min\{u,1+\int p\,dD(p)\}}
        {\iint\min\{u,1+p,1+q\}\,dD(p)dD(q)}.
 \label{eq:bo-instance-ratio}
\end{equation}

Let $u_{\mathsf{BO}}>1$ denote the unique root greater than one of
\begin{equation}
 u^3-2u^2-u+1=0,
 \qquad
 u_{\mathsf{BO}}=1+2\cos(2\pi/7)=2.24\ldots .
 \label{eq:bo-transition}
\end{equation}
The factorization
\[
 v^6-6v^4+5v^2-1
 =(v^3-2v^2-v+1)(v^3+2v^2-v-1)
\]
also shows that $\uRand{2}=u_{\mathsf{BO}}^2$.

\begin{theorem}[Exact blind-optimization curve]
\label{thm:bo-curve}
The optimal randomized size-asymptotic competitive ratio in
$\mathsf{BO}(u)$ against an oblivious adversary is
\begin{equation}
 \RrandBO(u)=
 \begin{cases}
  1,&0<u\le1,\\[3pt]
  \displaystyle\frac{u^3}{u^2+(u-1)^3},
    &1\le u\le u_{\mathsf{BO}},\\[9pt]
  \displaystyle\frac12\left(
      1+\sqrt{\frac{u^2+u-1}{u-1}}\right),
    &u\ge u_{\mathsf{BO}}.
 \end{cases}
 \label{eq:bo-curve}
\end{equation}
The formulas $u^3/[u^2+(u-1)^3]$ and
$\frac12(1+\sqrt{(u^2+u-1)/(u-1)})$ agree at
$u_{\mathsf{BO}}$.  Moreover,
\begin{equation}
 \RrandBO(u)=\frac{\sqrt u}{2}+\frac12+o(1)
 \qquad (u\to\infty).
 \label{eq:bo-curve-asymptotic}
\end{equation}
\end{theorem}

\begin{proof}
The claim is trivial for $u\le1$, because raw execution is no longer than
the unit optimization operation.  Assume $u>1$ and let
\[
 S_D(t)=D((t,u]),\qquad 0\le t\le u.
\]
Writing each nonnegative minimum as an integral of its level indicators in
\eqref{eq:bo-offline-coefficient} gives
\begin{equation}
 2\Omega_{\mathsf{BO},u}(D)
 =1+\int_0^{u-1}S_D(t)^2\,dt.
 \label{eq:bo-layer-cake}
\end{equation}
Put $S=\int_0^{u-1}S_D(t)\,dt$.  Since $S_D$ is nonincreasing,
\[
 \int_{u-1}^uS_D(t)\,dt
 \le S_D(u-1)
 \le\frac{S}{u-1}.
\]
Therefore, writing $\mu=\int p\,dD(p)=\int_0^uS_D(t)\,dt$,
\begin{equation}
 S\ge\frac{u-1}{u}\mu.
 \label{eq:bo-survival-mean}
\end{equation}
Cauchy--Schwarz now yields
\begin{equation}
 \int_0^{u-1}S_D(t)^2\,dt
 \ge\frac{S^2}{u-1}
 \ge (u-1)\left(\frac\mu u\right)^2.
 \label{eq:bo-survival-square}
\end{equation}
Set $b=\mu/u$.  Equations
\eqref{eq:bo-instance-ratio}--\eqref{eq:bo-survival-square} imply
\begin{equation}
 \alpha_{\mathsf{BO},u}(D)
 \le f_u(b)
 \defeq\frac{\min\{u,1+ub\}}{1+(u-1)b^2}.
 \label{eq:bo-binary-envelope}
\end{equation}
Equality is attained by the binary distribution
\begin{equation}
 D(\{u\})=b,
 \qquad
 D(\{0\})=1-b.
 \label{eq:bo-binary-hard}
\end{equation}
Indeed, its survival function is constant $b$ on $[0,u)$, so both
inequalities in \eqref{eq:bo-survival-square} are equalities.

It remains to maximize the scalar function $f_u$.  Put
\[
 b_0=\frac{u-1}{u}.
\]
For $b\ge b_0$ the numerator is $u$ and $f_u$ decreases, so this region is
maximized at $b_0$.  For $b\le b_0$, differentiation gives the unique
unconstrained maximizer
\begin{equation}
 b_*(u)=
 \frac{u}
 {u-1+\sqrt{(u-1)(u^2+u-1)}}.
 \label{eq:bo-b-star}
\end{equation}
The equality $b_*(u)=b_0$ is equivalent to
$u^3-2u^2-u+1=0$.  Hence $b_0$ is optimal up to $u_{\mathsf{BO}}$ and $b_*(u)$ is
optimal afterwards.  Substitution gives
\[
 f_u(b_0)=\frac{u^3}{u^2+(u-1)^3},
 \qquad
 f_u(b_*)=\frac12\left(
  1+\sqrt{\frac{u^2+u-1}{u-1}}\right).
\]
This proves \eqref{eq:bo-curve}.  The shuffled binary inputs from
\eqref{eq:bo-binary-hard} give the lower bounds under the input convention.
Together with the universal upper algorithm of
\cref{thm:bo-instance}, they prove both sides of the competitive-ratio claim.
If the maximizing mass $b$ is irrational, finite inputs use
$\lfloor bn\rfloor$ copies of $u$; continuity gives the same asymptotic
ratio.
Finally,
$(u^2+u-1)/(u-1)=u+2+1/(u-1)$ gives
\eqref{eq:bo-curve-asymptotic}.
\end{proof}

The hard distribution explains the separation.  Most jobs have optimized
length zero, while a carefully chosen small fraction have length $u$.  Blind
optimization is attractive on average, so the online algorithm optimizes the
jobs.  It cannot recognize a job with $p=u$ before committing to it, and
processing such a job early delays many zeros.  The clairvoyant optimum
optimizes the zeros, runs the jobs with $p=u$ raw, and places them last.  In
the revealing model, by contrast, the same jobs can be deferred after their
tests.  The
gap between the two curves therefore measures the scheduling value of the
revealed information itself.

\section*{Acknowledgements}
\addcontentsline{toc}{section}{Acknowledgements}

An initial improvement over the then-current deterministic lower bound for
obligatory testing was obtained with the Bolzano open-source multi-agent
research system~\cite{GrebikEtAl2026Bolzano}.  That preliminary result
motivated the present project.  Most of the subsequent conjecture generation,
proof development, counterexample search, and exposition was carried out with
substantial assistance from OpenAI's GPT-5.6 through
interactive Codex sessions.  The author selected the research directions,
reviewed the arguments, and takes responsibility for the claims
in this paper.

\addcontentsline{toc}{section}{References}
\bibliographystyle{alphaurl}
\bibliography{references}

\end{document}